\documentclass[UKenglish,cleveref,numberwithinsect,thm-restate,a4paper,11pt]{article}

\RequirePackage{doi}
\usepackage{comment}
\usepackage[margin=1in]{geometry}
\usepackage{hyperref}
\usepackage{blkarray}
\usepackage{bm}
\usepackage{rotating}
\usepackage{afterpage}
\usepackage{diagbox}
\usepackage{microtype}
\usepackage[fleqn]{amsmath}
\usepackage{amssymb,amsfonts,amsthm}
\usepackage{mathtools}
\usepackage{enumerate}
\usepackage{booktabs}
\usepackage{enumitem}
\usepackage[mathlines]{lineno}
\usepackage{thmtools, thm-restate}
\usepackage{tabularx}
\usepackage{url}
\usepackage{xspace}
\usepackage{graphicx}

\usepackage{environ}

\makeatletter
\newcommand{\problemname}[1]{\gdef\@problemname{#1}}
\newcommand{\probleminput}[1]{\gdef\@probleminput{#1}}
\newcommand{\problemoutput}[1]{\gdef\@problemoutput{#1}}
\newcommand{\problemparameter}[1]{\gdef\@problemparameter{#1}}
\NewEnviron{parameterizedproblem}{
  \problemname{}\probleminput{}\problemoutput{}\problemparameter{}
  \BODY
  \par\addvspace{.5\baselineskip}
  \noindent
  \begin{tabularx}{\textwidth}{@{\hspace{\parindent}} l X c}
    \multicolumn{2}{@{\hspace{\parindent}}l}{\@problemname} \\
    {Input}: & \@probleminput \\
    {Output:} & \@problemoutput \\
    {Parameter:} & \@problemparameter
    \end{tabularx}
  \par\addvspace{.5\baselineskip}
}
\makeatother

\setenumerate{itemsep=0pt,topsep=4pt,parsep=4pt}
\setitemize{itemsep=0pt,topsep=4pt,parsep=4pt}

\usepackage{tikz}
\usetikzlibrary{quotes, backgrounds, positioning, calc, arrows, arrows.meta, decorations.pathreplacing, shapes, shapes.misc}
\tikzstyle{vertex}=[circle,solid,draw=black,minimum size=14pt,inner sep=0pt]
\tikzstyle{cut}=[solid,draw=red]
\tikzset{>=stealth',every on chain/.append style={join},every join/.style={->}}

\tikzset{every picture/.style={draw=black, line width=.5pt},font=\sffamily}
\tikzset{every node/.style={circle, inner sep=0pt, minimum size=14},font=\sffamily}
\tikzset{>={Latex[color=,length=6pt,width=2.5pt 3]},font=\sffamily}
\tikzstyle{graph} = [fill=black!3, fill opacity=1, draw=black,font=\sffamily]
\tikzstyle{source} = [fill=yellow!30, fill opacity=1, draw=black,font=\sffamily]
\tikzstyle{joint} = [fill=black!3, fill opacity=1, draw=black,font=\sffamily]
\tikzstyle{tree} = [fill=red!30, fill opacity=1, draw=black,font=\sffamily]
\tikzstyle{stain} = [fill=orange, draw=black, line width=.7pt, fill opacity=.15, densely dotted]
\tikzstyle{stainEdge} = [fill=green, draw=black, line width=.7pt, fill opacity=.15, densely dotted]

\newtheorem{theorem}{Theorem}
\newtheorem{lemma}[theorem]{Lemma}

\newtheorem{claim}[theorem]{Claim}
\newtheorem{fact}[theorem]{Fact}
\newtheorem{observation}[theorem]{Observation}

\newtheorem{definition}[theorem]{Definition}

\newtheorem*{conjecture*}{Conjecture}
\newtheorem*{question*}{Open Problem}

\newcommand{\N}{\ensuremath{\mathbb{N}}}
\newcommand{\Q}{\ensuremath{\mathbb{Q}}}
\newcommand{\R}{\ensuremath{\mathbb{R}}}

\newcommand{\scB}{\mathcal{B}}
\newcommand{\scC}{\mathcal{C}}

\newcommand{\scG}{\mathcal{G}}
\newcommand{\scH}{\mathcal{H}}

\newcommand{\ccFPT}{\ensuremath{\mathsf{FPT}}}
\newcommand{\W}[1]{\ensuremath{\mathsf{W[#1]}}}

\mathchardef\mhyphen="2D

\renewcommand{\to}{\!\rightarrow\!}

\newcommand{\homs}[2]{\mbox{\ensuremath{\mathsf{Hom}(#1 \to #2)}}}

\newcommand{\subs}[2]{\mbox{\ensuremath{\mathsf{Sub}(#1 \to #2)}}}

\newcommand{\indsubs}[2]{\mbox{\ensuremath{\mathsf{IndSub}(#1 \to #2)}}}

\newcommand{\homsprob}{\ensuremath{\textsc{Hom}}}
\newcommand{\dirhomsprob}{\ensuremath{\textsc{DirHom}}}
\newcommand{\dirhomsprobd}{\ensuremath{\textsc{DirHom}_\mathrm{d}}}
\newcommand{\dirsubsprobd}{\ensuremath{\textsc{DirSub}_\mathrm{d}}}
\newcommand{\dirsubsprob}{\ensuremath{\textsc{DirSub}}}
\newcommand{\dirindsubsprobd}{\ensuremath{\textsc{DirIndSub}_\mathrm{d}}}
\newcommand{\dirindsubsprob}{\ensuremath{\textsc{DirIndSub}}}

\newcommand{\cphomsproba}{\ensuremath{\textsc{cp-Hom}_{\mathrm{a}}}}
\newcommand{\cpdirhomsprobd}{\ensuremath{\textsc{cp-DirHom}_{\mathrm{d}}}}

\newcommand{\indsubsprob}{\ensuremath{\textsc{IndSub}}}

\newcommand{\subsprob}{\ensuremath{\textsc{Sub}}}

\newcommand{\CSP}{\ensuremath{\textsc{CSP}}}
\newcommand{\fptred}{\ensuremath{\leq^{\mathrm{fpt}}_{\mathrm{T}}}}

\newcommand{\lovasz}{Lov{\'{a}}sz}

\newcommand{\arity}{\ensuremath{\mathsf{a}}}

\newcommand{\aw}{\operatorname{aw}} 

\newcommand{\scE}{\mathcal{E}}

\newcommand{\reach}[1]{\mathcal{R}(#1)}
\newcommand{\reachv}[1]{R(#1)}
\newcommand{\dsplit}[1]{\vec{#1}^2}
\newcommand{\scount}[1]{\alpha_s(#1)}
\newcommand{\reduct}[1]{\Gamma(#1)}

\newcommand{\fwidth}[1]{{#1}\mhyphen\mathsf{width}}

\newcommand{\asub}{{\mathsf{sub}}}
\newcommand{\aindsub}{{\mathsf{indsub}}}

\newcommand{\SCCquot}[1]{\vec{#1}\mathbin{{/}{\sim}}\xspace}
\newcommand{\Hsim}{\SCCquot{H}}

\title{The Complexity of Pattern Counting in Directed Graphs, Parameterised by the Outdegree\footnote{This research was funded in whole, or in part, by the Royal Society project "RAISON DATA" (Project reference: RP\textbackslash R1\textbackslash 201074) and by the Google Focused Award ``Algorithms and Learning for AI'' (ALL4AI). For the purpose of Open Access, the authors have applied a CC BY public copyright licence to any Author Accepted Manuscript version arising from this submission. All data is provided in full in the results section of this paper.}}

\begin{document}

\author{Marco Bressan \\ Department of Computer Science\\ University of Milan\\ Italy 
\and
Matthias Lanzinger \\ Department of Computer Science \\ University of Oxford\\ United Kingdom 
\and
Marc Roth \\ Department of Computer Science \\ University of Oxford\\ United Kingdom 
}

\maketitle
\begin{abstract}
We study the fixed-parameter tractability of the following fundamental problem: given two directed graphs $\vec H$ and $\vec G$, count the number of copies of $\vec H$ in $\vec G$. The standard setting, where the tractability is well understood, uses only $|\vec H|$ as a parameter. In this paper we take a step forward, and adopt as a parameter $|\vec H|+d(\vec G)$, where $d(\vec G)$ is the maximum outdegree of $|\vec G|$. Under this parameterization, we completely characterize the fixed-parameter tractability of the problem in both its non-induced and induced versions through two novel structural parameters, the \emph{fractional cover number} $\rho^*$ and the \emph{source number} $\alpha_s$. On the one hand we give algorithms with running time $f(|\vec H|,d(\vec G)) \cdot |\vec G|^{\rho^*\!(\vec H)+O(1)}$ and $f(|\vec H|,d(\vec G)) \cdot |\vec G|^{\alpha_s(\vec H)+O(1)}$ for counting respectively the copies and induced copies of $\vec H$ in $\vec G$; on the other hand we show that, unless the Exponential Time Hypothesis fails, for any class $\vec C$ of directed graphs the (induced) counting problem is fixed-parameter tractable if and only if $\rho^*(\vec C)$ ($\alpha_s(\vec C)$) is bounded.
These results explain how the orientation of the pattern can make counting easy or hard, and prove that a classic algorithm by Chiba and Nishizeki and its extensions (Chiba, Nishizeki SICOMP 85; Bressan Algorithmica 21) are optimal unless ETH fails.

Our proofs consist of several layers of parameterized reductions that preserve the outdegree of the host graph. To start with, we establish a tight connection between counting \emph{homomorphisms} from $\vec H$ to $\vec G$ to \#CSP, the problem of counting solutions of constraint satisfactions problems, for special classes of patterns that we call \emph{canonical DAGs}. To lift these results from canonical DAGs to arbitrary directed graphs, we exploit a combination of several ingredients: existing results for \#CSPs (Marx JACM 13; Grohe, Marx TALG 14), an extension of graph motif parameters (Curticapean, Dell, Marx STOC 17) to our setting, the introduction of what we call \emph{monotone reversible minors}, and careful analysis of quotients of directed graphs in order to relate their adaptive width and fractional hypertree width as a function to our novel parameters.
Along the route we establish a novel bound of the integrality gap for the fractional independence number of hypergraphs based on adaptive width, which might be of independent interest.
\end{abstract}

\thispagestyle{empty}

\pagebreak

\setcounter{page}{1}

\section{Introduction}
We study the complexity of the following fundamental counting problem: given two directed graphs, $\vec H$ (the ``pattern'') and $\vec G$ (the ``host''), count the number of occurrences or induced occurrences of $\vec H$ in $\vec G$. This problem, known as \emph{subgraph counting}, \emph{motif counting}, or \emph{pattern counting}, has gained great popularity because of its apparent ubiquity in a diverse selection of fields, from social network analysis~\cite{UganderBK13} to network science~\cite{Miloetal02,Miloetal04}, and from database theory~\cite{DurandM15,ChenM16,DellRW19,Arenasetal21,FockeGRZ22} and data mining~\cite{AhmedNRD15,Babis17} to bioinformatics~\cite{Nogaetat08,Schilleretal15}, phylogeny~\cite{Kuchaievetal10}, and genetics~\cite{ShenOrrMM02,Tranetal13}. For this reason, subgraph counting in general has received significant attention from the theoretical community in the last two decades, with a flurry of novel techniques and exciting results~\cite{ArvindR02,FlumG04,Curticapean13,CurticapeanM14,JerrumM15,Meeks16,CurticapeanDM17,BrandDH18,LokshtanovBSZ21,BressanR21,BeraGLSS22,FockeR22}.

Since subgraph counting in general is hard (think of counting cliques), it is common to \emph{parameterise} the problem so as to allow for a ``bad'' dependence on some quantity that is believed to be small in practice~\cite{FlumG04,CurticapeanDM17}. The standard parameterisation is by the size of $\vec H$, that is, $|\vec H|=|V(\vec H)|+|E(\vec H)|$. In that case, one says the problem is \emph{fixed-parameter tractable}, or in the class \ccFPT, if for some (computable) function $f$ it admits an algorithm that runs in time $f(|\vec H|) \cdot |\vec G|^{O(1)}$ for all $\vec H$ and $\vec G$. This means one considers as efficient an algorithm with running time, say, $2^{|\vec H|} \cdot |\vec G|$, but not one with running time $|\vec G|^{|\vec H|}$. The rationale is that in practice $\vec H$ is often very small compared to $\vec G$, thus a running time of $2^{|\vec H|}\cdot |\vec G|$ is better than one of $|\vec G|^{|\vec H|}$. Under this parameterisation, the tractability of the problem is well understood: for the undirected version, both the induced and non-induced versions are in \ccFPT\ if and only if certain invariants of $H$ are bounded~\cite{ChenTW08,CurticapeanM14, CurticapeanDM17}, and it is not hard to show that the same holds for the directed case as well (see Section~\ref{sec:results}).

While the parameterisation by $|\vec H|$ is standard, it is also quite restrictive. Consider for instance the problem of counting the \emph{induced} copies of $\vec{H}$ in $\vec{G}$: when parameterised by $|\vec H|$, it is well-known that the problem is in $\ccFPT$ if and only if the pattern size $|\vec{H}|$ is bounded (see~\cite{ChenTW08} and Appendix~\ref{app:unbounded_degree}). Thus, under this parameterization, one can efficiently count the induced copies of just a \emph{finite} number of patterns. Suppose instead the parameter is $|\vec H|+d(\vec G)$, where $d(\vec G)$ is the maximum outdegree of $\vec G$; the problem is then considered tractable if for some (computable) function $f$ it admits an algorithm that runs in time $f(|\vec H|,d(\vec G)) \cdot |\vec G|^{O(1)}$ for all $\vec H$ and $\vec G$. It is not hard to see that, under this parameterization, the problem becomes $\ccFPT$ even for \emph{infinite} families of patterns. Let indeed $\vec H$ be the acyclic orientation of a $k$-clique: since $\vec H$ has only one source $s$ (a vertex of indegree $0$), one can first guess the image of $s$ in $\vec{G}$ and then iterate over all $(k-1)$-vertex subsets in the out-neighbourhood of $s$, which yields an algorithm with running time $O(d(\vec G)^{|\vec H|} \cdot |\vec{G}|)$. This idea was in fact extended to counting subgraphs in degenerate host graphs, which have orientations with bounded outdegree~\cite{Bera-ITCS20,BeraS20,Bera-SODA21,BressanR21,GishbolinerLSY22,BeraGLSS22} (see Section~\ref{sec:related_work} for a detailed discussion). Thus, adopting $|\vec H|+d(\vec G)$ as a parameter can open the door to a richer landscape of tractability.

The goal of the present work is to understand precisely what that landscape is; that is, to understand when the aforementioned problems, parameterised by $|\vec H|+d(\vec G)$, are in \ccFPT\ as a function of the pattern $\vec H$. In addition to the aforementioned example of pattern counting in degenerate graphs, there is another reason to consider $d(\vec G)$ as part of the parameter when counting directed subgraphs: several ``real-world'' directed graphs that are natural ``hosts'' have small or constant outdegree. This is true for many web graphs or online social network graphs, where the maximum outdegree is much smaller than the average degree or the maximum indegree; and it is true by construction in graphs produced by generative models such as preferential attachment~\cite{PrefAttach}.
As is customary, to express the dependence on the structure of $\vec H$, we formulate the problems as a function of a \emph{class} $\vec C$  of patterns --- for instance, one may let $\vec C$ be the class of all directed complete graphs, or of all directed trees. Let then $\vec C$ denote an arbitrary family of directed graphs, and for any $\vec H$ and $\vec G$ let $\#\subs{\vec H}{\vec G}$ and $\#\indsubs{\vec H}{\vec G}$ denote respectively the number of copies and induced copies of $\vec H$ in $\vec G$. 

\noindent Our parameterised counting problems are formally defined as follows:
\begin{parameterizedproblem}
\problemname{$\#\dirsubsprobd(\vec C)$}
\probleminput{a pair of digraphs $(\vec H, \vec G)$ with $\vec H \in \vec C$}
\problemoutput{\#\subs{\vec H}{\vec G}}
\problemparameter{$|\vec H|+d(\vec{G})$}
\end{parameterizedproblem}
\begin{parameterizedproblem}
\problemname{$\#\dirindsubsprobd(\vec C)$}
\probleminput{a pair of digraphs $(\vec H, \vec G)$ with $\vec H \in \vec C$}
\problemoutput{\#\indsubs{\vec H}{\vec G}}
\problemparameter{$|\vec H|+d(\vec{G})$}
\end{parameterizedproblem}
\noindent The goal of the present work is to understand which structural properties of the elements of $\vec C$ determine whether $\#\dirsubsprobd(\vec C)$ and $\#\dirindsubsprobd(\vec C)$ are in $\ccFPT$. 

The rest of this manuscript is organised as follows. Section~\ref{sec:results} gives a concise overview of our results and their significance. Section~\ref{sec:overview} gives a detailed overview of our proofs and their key technical insights. The complete proofs of all our claims can be found in the remaining sections.

\section{Results}\label{sec:results}
We give complete complexity classifications for $\#\dirsubsprobd(\vec C)$ and $\#\dirindsubsprobd(\vec C)$, into \ccFPT\ versus non-\ccFPT\ cases, as a function of $\vec C$. These complexity classifications, which are formally stated below, have the succinct form ``The problem is in \ccFPT\ if and only if $p(\vec C)$ is bounded'', where $p$ is some parameter measuring the structural complexity of the graphs in $\vec C$. The definition of those parameters is not elementary and requires the introduction of some ancillary notation and definitions, which we are going to do next. In order to understand \emph{why} those parameters are the right ones, instead, one should take the technical tour of Section~\ref{sec:overview}.

Let us then introduce our structural parameters. First, we need to define \emph{reachability hypergraphs} and \emph{contours}. Let $\vec H$ be a directed graph, and let $\mathcal S$ be the set of its strongly connected components. Denote by $\sim$ the equivalence relation over $V(\vec H)$ given by $\mathcal S$, and let $\Hsim$ be the quotient of $\vec H$ w.r.t.\ $\sim$; with a little abuse of notation we let $\mathcal S$ be the vertex set of $\Hsim$. A strongly connected component $S \in V(\Hsim)$ is a \emph{source} if it has indegree $0$ in $\Hsim$. Let $S_1,\ldots,S_k$ be the set of all such sources. For any $S \in V(\Hsim)$ let $R(S)$ be the set of vertices reachable from $S$ in $\vec H$.
\begin{definition}
The \emph{reachability hypergraph} of $\vec{H}$, denoted by $\reach{\vec{H}}$, is the hypergraph with vertex set $V(\vec{H})$ and edge set $\{\reachv{S_i} : i \in [k]\}$.
\end{definition}
\noindent Intuitively, $\reach{\vec{H}}$ measures the complexity of $\vec H$ in terms of ``reachability relationships''. However, to state our classifications correctly, we need to consider a slight modification of $\reach{\vec{H}}$.
\begin{definition}\label{def:contour_intro}
The \emph{contour} of $\vec{H}$, denoted by $\reduct{\vec{H}}$, is the hypergraph $\reach{\vec{H}} \setminus \cup_{i \in [k]}S_i$.
\end{definition}
\noindent For instance, if $\vec H_n$ is obtained by orienting the edges of the $1$-subdivision of the complete graph $K_n$ towards the original vertices, then $\reduct{\vec H_n}=K_n$.

\begin{figure}[t]
    \centering
    \includegraphics[scale=0.72]{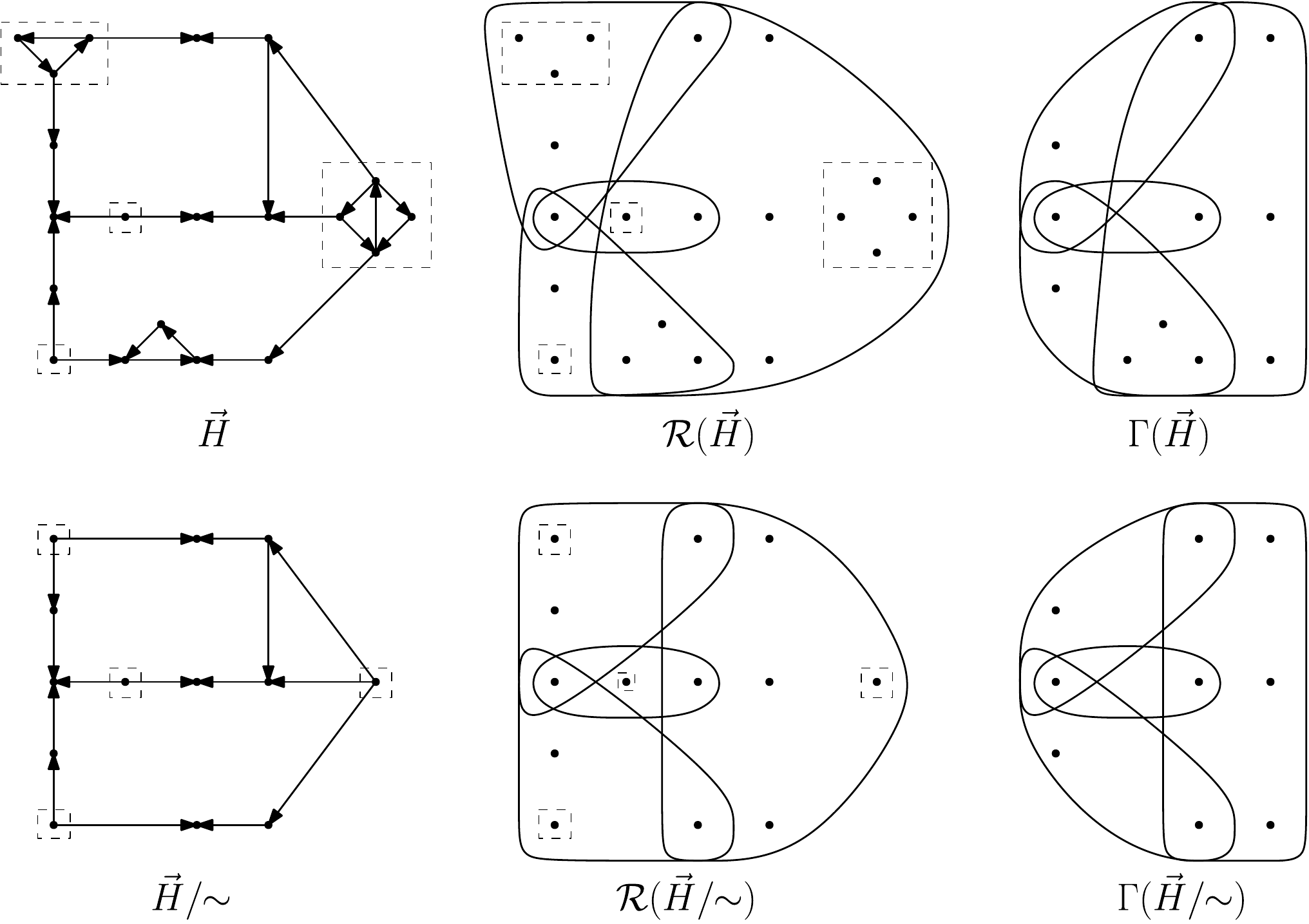}
    \caption{\small \emph{(Top:)} A directed graph $\vec{H}$, the sources of which are highlighted in dashed boxes, its reachability hypergraph $\reach{\vec{H}}$, and its contour $\reduct{\vec{H}}$. \emph{(Bottom:)} The same constructions for the DAG $\Hsim$ obtained from $\vec{H}$ by identifying all strongly connected components. Clearly, the source number is invariant under taking the quotient w.r.t.\ $\sim$, that is, $\scount{\vec{H}} =\scount{\Hsim}$. We will see that the same is true for the fractional cover number, that is, $\rho^\ast(\vec{H})=\rho^\ast(\Hsim)$. Consequently, it \emph{always} suffices to consider the DAG $\Hsim$ for determining the complexity of counting copies and induced copies of $\vec{H}$.}
    \label{fig:main_intro}
\end{figure}

Finally, we introduce our directed graph invariants, the \emph{fractional cover number} and the \emph{source number}. Let $\scH$ be a hypergraph. A function $\gamma:E(\scH) \to [0,\infty)$ is a \emph{fractional edge cover} of $\scH$ if for every $v \in V(\scH)$
\begin{equation}\label{eq:intro_frac_number}
    \sum_{e \in E(\scH) : v\in e} \gamma(e) \ge 1 \,.
\end{equation}
The weight of $\gamma$ is $\sum_{e \in E(\scH)} \gamma(e)$. The fractional edge cover number of $\scH$, denoted by $\rho^*(\scH)$, is the smallest weight of any fractional edge cover of $\scH$.
Then:
\begin{definition}
The \emph{fractional cover number} of $\vec H$ is $\rho^*(\vec H)=\rho^*(\reduct{\vec H})$. The \emph{source number} of $\vec H$ is the number of sources in $\mathcal S$, denoted by $\alpha_s(\vec H)$; in other words, the number of strongly connected components of $\vec H$ that are not reachable from any other connected component.

\end{definition}
\noindent Intuitively, both $\rho^*$ and $\alpha_s$ measure the complexity of covering $\vec H$ through its sources. Our main result, the following dichotomy theorem, says that such a ``covering complexity'' determines precisely the fixed-parameter tractability of our problems. For any class $\vec C$ of directed graphs let $\rho^*(\vec C) = \sup_{\vec H \in \vec C} \rho^*(H)$ and $\alpha_s(\vec C) = \sup_{\vec H \in \vec C} \alpha_s(H)$.
\begin{theorem}\label{thm:main_intro}
If the Exponential Time Hypothesis holds, then:
\begin{enumerate}\itemsep0pt
    \item $\#\dirsubsprobd(\vec{C}) \in \ccFPT$ if and only if $\rho^*(\vec{C}) < \infty$
    \item $\#\dirindsubsprobd(\vec{C}) \in \ccFPT$ if and only if $\alpha_s(\vec{C}) < \infty$
\end{enumerate}
\end{theorem}
\noindent Note that ETH is used only by the ``only if'' direction. While the statement of Theorem~\ref{thm:main_intro} is simple, its proof is nontrivial --- virtually all of this manuscript is devoted to it. To put the theorem into perspective, Table~\ref{tab:dichotomy} compares it to dichotomies for the other variants of the problem. We also observe that Theorem~\ref{thm:main_intro} can be slightly strengthened: we can show the hardness direction even for acyclic host graphs.

As a consequence of Theorem~\ref{thm:main_intro}, we can claim the optimality (in an FPT sense) of the well-known approach to counting the induced copies of a DAG $\vec H$ in a host $\vec G$ of bounded outdegree, used in several recent works on counting in hosts of bounded degeneracy~\cite{Bera-ITCS20,BeraS20,Bera-SODA21,BressanR21,GishbolinerLSY22,BeraGLSS22}. This approach consists in guessing the images of the sources of $\vec H$ in $\vec G$, and has running time $f(|\vec H|,d(\vec G)) \cdot |\vec G|^{\alpha_s(\vec H)+O(1)}$. By Theorem~\ref{thm:main_intro}, unless ETH fails the dependence on $\alpha_s(\vec H)$ at the exponent cannot be avoided, hence that approach is optimal in an FPT sense.

\begin{table}[t]\centering
    \begin{tabularx}{0.97\textwidth}{lccc}
        \toprule
        \begin{tabular}{@{}l}
             $~$
        \end{tabular}&\begin{tabular}{@{}l}
             $\#\subsprob(C)$
        \end{tabular} &\begin{tabular}{@{}l}
             $\#\dirsubsprob(\vec{C})$
        \end{tabular}& \begin{tabular}{@{}l}
            $\#\dirsubsprobd(\vec{C})$
        \end{tabular}\\
        \cmidrule(r{2ex}){2-4}
        \begin{tabular}{l}
             FPT criterion
        \end{tabular} & \begin{tabular}{c}
             $\mathsf{vc}(C)<\infty$ \\
             Curticapean, Marx~\cite{CurticapeanM14}
        \end{tabular} & \begin{tabular}{c}
             $\mathsf{vc}(\vec{C})<\infty$ \\
             Appendix~\ref{app:unbounded_degree}
        \end{tabular} & \begin{tabular}{c}
             $\rho^\ast(\vec{C})<\infty$ \\
             Theorem~\ref{thm:main_intro}
        \end{tabular} \\[8pt]
        \toprule
        \begin{tabular}{@{}l}
             $~$
        \end{tabular}&\begin{tabular}{@{}l}
             $\#\indsubsprob(C)$
        \end{tabular} &\begin{tabular}{@{}l}
             $\#\dirindsubsprob(\vec{C})$
        \end{tabular}& \begin{tabular}{@{}l}
            $\#\dirindsubsprobd(\vec{C})$
        \end{tabular}\\
        \cmidrule(r{2ex}){2-4}
        \begin{tabular}{l}
             FPT criterion
        \end{tabular} & \begin{tabular}{c}
             $|C|<\infty$ \\
             Chen, Thurley, Weyer~\cite{ChenTW08}
        \end{tabular} & \begin{tabular}{c}
             $|\vec{C}|<\infty$ \\
             Appendix~\ref{app:unbounded_degree}
        \end{tabular} & \begin{tabular}{c}
             $\scount{\vec{C}}<\infty$ \\
             Theorem~\ref{thm:main_intro}
        \end{tabular} \\[8pt]
        \bottomrule
        \end{tabularx}
     \caption[caption]{\small Our results for $\#\dirsubsprobd(\vec{C})$ and $\#\dirindsubsprobd(\vec{C})$ compared against $\#\indsubsprob(\vec{C})$ and $\#\dirindsubsprob(\vec{C})$, which are their counterparts parameterised by $|\vec{H}|$, and against $\#\subsprob(C)$ and $\#\indsubsprob(C)$, which are their \emph{undirected} counterparts parameterised by $|H|$. Here $\mathsf{vc}$ denotes the vertex cover number; for directed graphs this is is just the vertex cover number of the underlying undirected graph. The results for $\#\indsubsprob(\vec{C})$ and $\#\dirindsubsprob(\vec{C})$ are  folklore in the community.}
      \label{tab:dichotomy}
\end{table}

It shall be noted that, for $\#\dirindsubsprobd(\vec{C})$, the non-\ccFPT case in Theorem~\ref{thm:main_intro} also yields $\#\W{1}$-hardness (see Section~\ref{sec:prelims} for a definition of $\#\W{1}$). For $\#\dirsubsprobd(\vec{C})$ instead we do not prove $\#\W{1}$-hardness; the reason is that our proof uses a reduction from certain families of $\#\textsc{CSP}$ instances which by~\cite{Marx13} we know to be not in $\ccFPT$ if ETH holds, but we do not know if they are $\#\W{1}$-hard too.

When the problems in Theorem~\ref{thm:main_intro} are in $\ccFPT$, we can show simple algorithms that solve them in time $f(|\vec H|,d(\vec G)) \cdot |\vec G|^{p(\vec H)+O(1)}$ where $p \in \{\rho^*,\alpha_s\}$. Formally, we prove:
\begin{theorem}\label{thm:main_intro_bounds}
For some computable function $f$ there is an algorithm solving $\#\dirsubsprobd(\vec{C})$ in time $f(|\vec{H}|,d(\vec G))\cdot |\vec{G}|^{\rho^*\!(\vec H)+O(1)}$. The same holds for $\#\dirindsubsprobd(\vec{C})$ with $\alpha_s$ in place of $\rho^*$.
\end{theorem}
\noindent We point out that theorems \ref{thm:main_intro} and \ref{thm:main_intro_bounds} remain true in the (edge or vertex) weighted setting, too.

A simple example shows Theorem~\ref{thm:main_intro} and Theorem~\ref{thm:main_intro_bounds} in action.
Let $\Delta_1$ and $\Delta_2$ be respectively the cyclic and acyclic orientations of $K_3$, and for each $k \in \N$ let $\Delta_1^k$ and $\Delta_2^k$ consist of $k$ disjoint copies of respectively $\Delta_1$ and $\Delta_2$. Finally, let $\vec{C}_1=\{\Delta_1^k : k \in \N\}$ and $\vec{C}_2=\{\Delta_2^k : k \in \N\}$. Although the patterns are rather elementary, establishing the tractability of $\#\dirsubsprobd(\vec{C}_1)$ and $\#\dirsubsprobd(\vec{C}_2)$ ``by hand'' can be laborious. Theorem~\ref{thm:main_intro} and Theorem~\ref{thm:main_intro_bounds} answer immediately: $\rho^*(\Delta_1^k)=0$, since in $\Delta_1^k$ every vertex belongs to some source, hence $\#\dirsubsprobd(\vec{C}_1)$ is fixed-parameter tractable and solvable in time $f(|\vec H|,d(\vec G)) \cdot |\vec G|^{O(1)}$; but $\rho^*(\Delta_2^k)=k$, since $\reduct{\Delta_2^k}$ has $k$ disjoint hyperedges, hence $\#\dirsubsprobd(\vec{C}_2)$ is not fixed-parameter tractable unless ETH fails. One can also see that $\alpha_s(\Delta_1^k)=\alpha_s(\Delta_2^k)=k$; therefore, by Theorem~\ref{thm:main_intro}, under ETH both $\#\dirindsubsprobd(\vec{C}_1)$ and $\#\dirindsubsprobd(\vec{C}_2)$ are not fixed-parameter tractable.

Another example helps appreciating the different between our parameterization and the standard one, as well as the necessity of $\rho^*$ being fractional. Let $H_k$ be the graph defined as follows. The vertices of $H_k$ are $U_k \cup D_k$ where $U_k=\{1,\dots,2k\}$ and $D_k=\{A \subseteq U_k ~|~ |A|=k\}$; and for each $i\in U_k$, there is an edge between $i$ and $D\in D_k$ if and only if $i\in D$. 
Let $C$ be the class of all $H_k$. It is not hard to show that $H_k$ contains the subdivision of the $k$-clique as induced subgraph. Thus the vertex-cover number of $C$ is unbounded and, assuming ETH, Table~\ref{tab:dichotomy} yields that $\#\subsprob(C)$ and $\#\dirsubsprob(\vec{C})$ are not fixed-parameter tractable for any class $\vec{C}$ obtained by orienting the graphs in $C$. However, if we parameterise also by the outdegree of the host, then the situation becomes much more subtle. Let $\vec{C}$ be the class of digraphs obtained by orienting the edges in the $H_k$ from $U_k$ to $D_k$; an argument similar to~\cite[Example~4.2]{GroheM14} shows that $\rho^\ast(\vec{H}_k)\leq 2$ for each $\vec{H}_k \in \vec{C}$, thus $\#\dirsubsprobd(\vec{C})$ is fixed-parameter tractable by Theorem~\ref{thm:main_intro_bounds}. Moreover,~\cite[Example~4.2]{GroheM14} show that any non-fractional cover of $\vec H_k$ has super-constant weight; this proves that considering the \emph{fractional} cover number $\rho^*$ is crucial; its integral counterpart cannot work. 

We conclude this section with a result of independent interest developed in our proofs. Let $\scH$ be a hypergraph. The independence number $\alpha(\scH)$ of $\scH$ is the size of the largest subset of $V(\scH)$ such that no two of its elements are contained in a common edge. The natural relaxation of this definition yields the \emph{fractional} independence number $\alpha^*(\scH)$. Our result is that the ratio between $\alpha^*(\scH)$ and $\alpha(\scH)$, i.e.\ the integrality gap of $\alpha$, is bounded by the \emph{adaptive width}\footnote{Note that adaptive width is equivalent to submodular width~\cite{Marx13}.} of~$\scH$~\cite{Marx10}.
\begin{theorem}\label{thm:independence_Gap_Intro}
Every hypergraph $\scH$ satisfies
$\alpha(\scH) \ge \frac{1}{2}+\frac{\alpha^{\ast}(\scH)}{4\aw(\scH)}$.
\end{theorem}

\section{Technical Overview}\label{sec:overview}
This section gives an overview of the tools and techniques behind the results of Section~\ref{sec:results}. The overview focuses on $\#\dirsubsprobd(\vec C)$, but similar arguments apply to $\#\dirindsubsprobd(\vec C)$. Before digging into the most technical part, let us give the high-level idea of our proof strategy.

At the root of all our results is a standard connection between copies and homomorphisms, explained in Section~\ref{sec:hom_basis}. It is well known indeed that $\#\subs{\vec H}{\vec G}$ can be expressed as a linear combination of homomorphism counts, $\sum_{\vec F} a_{\vec H}(\vec F) \cdot \#\homs{\vec F}{\vec G}$, where $a_{\vec H}(\vec F) > 0$ and $\vec F$ ranges over a certain set of \emph{quotients} of $\vec H$. Here, a quotient of $\vec{H}$ is a directed graph obtained from $\vec{H}$ by contracting (not necessarily connected) vertex subsets into single vertices (see Section~\ref{sec:prelims} for the formal definition). It is also known that the complexity of computing $\#\subs{\vec H}{\vec G}$ equals, up to $f(|\vec H|)$ factors, that of computing the hardest $\#\homs{\vec F}{\vec G}$ term. Therefore we can reduce $\#\dirsubsprobd(\vec C)$ \emph{to} and \emph{from} its homomorphism counting version $\#\dirhomsprobd(\vec Q)$, where $\vec Q$ consists of certain quotients of $\vec C$. Armed with these results, we proceed as follows.

First, in Section~\ref{sec:ub_intro} we prove that $\rho^*(\vec C) < \infty$ implies $\#\dirsubsprobd(\vec C) \in \ccFPT$. To this end we prove that if $\vec F$ is a quotient of $\vec H$ then the \emph{fractional hypertreewidth} of the contour of $\vec F$ satisfies $\mathsf{fhtw}(\reduct{\vec F)} \le \rho^*(\vec H)$. Therefore, $\mathsf{fhtw}(\reduct{\vec Q)} \le \rho^*(\vec C) < \infty$. We then show that computing $\#\homs{\vec F}{\vec G}$ can be reduced in FPT time to counting the homomorphisms from $\reduct{\vec F}$ to a hypergraph $\scG$ obtained from $\vec F$ and $\vec G$. As hypergraph homomorphism counting is in $\ccFPT$ when the pattern has bounded fractional hypertreewdith, this proves the claim.

Next, in Section~\ref{sec:lb_intro} we prove that $\rho^*(\vec C) = \infty$ implies $\#\dirsubsprobd(\vec C) \notin \ccFPT$, or ETH fails.
To start with, we suppose $\vec C$ contains only \emph{canonical DAGs}, directed graphs of a particularly simple type. We can prove that the aforementioned problem of counting homomorphisms between hypergraphs can be reduced to $\#\dirhomsprobd(\vec C)$ if the considered hypergraph patterns belong to the contours of $\vec C$. By existing results this implies that, unless ETH fails, $\#\dirhomsprobd(\vec C) \notin \ccFPT$ whenever the contours of $\vec C$ have unbounded \emph{adaptive width}~\cite{MarxSW17}.
It remains to lift these results from canonical DAGs to abitrary DAGs and, ultimately, to arbitrary directed graphs. To this end, we introduce what we call \emph{monotone reversible minors} (MRMs). Intuitively, $\vec H'$ is an MRM of $\vec H$ if there exists an FPT reduction from counting copies $\vec H'$ to counting copies of $\vec H$, and if $\vec H'$ preserves some parameters of interest (like $\rho^*$). We show that \emph{every} directed graph $\vec H$ has an MRM $\vec H'$ that is a canonical DAG, so counting $\vec H$ is at least as hard as counting $\vec H'$. Next, we show that counting copies of $\vec H'$ is hard. To this end we show that, if $\rho^*(\vec H')$ is large, then its reduct $\reduct{\vec H'}$ has large adaptive width or large independence number. By employing arguments from the homomorphism connection above and from~\cite{BressanR21}, this implies that counting the copies of $\vec H'$ is hard unless ETH fails, which concludes our proof.

In what follows we use standard terminology as much as possible; in any case, all concepts and terms are defined formally in Section~\ref{sec:prelims}.

\subsection{The Directed Homomorphism Basis}\label{sec:hom_basis}
The first ingredient of our work is the so-called \emph{homomorphism basis} introduced by Curticapean, Dell, and Marx~\cite{CurticapeanDM17}, which establishes a common connection between (undirected) parameterised pattern counting problems. 
Although the original framework is for undirected graphs, it can be equally well be formulated for the directed case, as we are going to do.
Let $\vec{H}$ be a digraph. There is a function $\asub_{\vec{H}}$ of finite support from digraphs to rationals such that for each digraph~$\vec{G}$:
\begin{align}
    \#\subs{\vec{H}}{\vec{G}} &= \sum_{\vec{F}} \asub_{\vec{H}}(\vec{F}) \cdot \#\homs{\vec{F}}{\vec{G}} \label{eq:subs_to_homs}
\end{align}
This identity follows by well-known transformations based on inclusion-exclusion and M\"obius inversion (see e.g.\ Chapter 5.2.3.\ in ~\lovasz~\cite{Lovasz12}). 
It is also well known that $\asub_{\vec{H}}(\vec{F})\neq 0$ if and only if $\vec{F}$ is a quotient of $\vec{H}$.

These facts allow us to construct a reduction from the parameterized problem of computing $\#\subs{\vec{H}}{\vec{G}}$ to the parameterized problem of computing $\#\homs{\vec{H}}{\vec{G}}$ and vice versa. More precisely, one can show that computing $\#\subs{\vec{H}}{\vec{G}}$ is precisely as hard (in FPT-equivalence terms) as computing the hardest term $\#\homs{\vec{F}}{\vec{G}}$ in the summation of~\eqref{eq:subs_to_homs}.
One direction is obvious --- the time to compute $\#\subs{\vec{H}}{\vec{G}}$ is the sum of the times to compute all terms $\#\subs{\vec{F}}{\vec{G}}$, whose number is a function of $\vec H$.
The other direction is nontrivial, and was established for multiple variants of subgraph counting over the past years~\cite{ChenM16,CurticapeanDM17,DellRW19,RothSW21,BeraGLSS22}. Rather than extending those results to yet another variant (directed graphs), we observe that the constructive version of Dedekind's Theorem on the linear independence of characters yields a general interpolation method that subsumes all those results, including the one for directed graphs. We prove what follows (see Theorem~\ref{thm:dedekind} for a more complex but complete version):
\begin{theorem}[Simplified Version]\label{thm:dedekind_intro}
Let $(\mathrm{G},\ast)$ be a semigroup. Let furthermore $(\varphi_i)_{i\in[k]}$ with $\varphi_i: \mathrm{G} \to \Q$ be pairwise distinct and non-zero semigroup homomorphisms of $(\mathrm{G},\ast)$ into $(\Q,\cdot)$, that is, $\varphi_i(g_1\ast g_2)= \varphi_i(g_1)\cdot \varphi_i(g_2)$ for all $i\in[k]$ and $g_1,g_2\in \mathrm{G}$. Let $\phi : \mathrm{G} \to \Q$ be a function
\begin{equation}\label{eq:dedekind_intro}
    \phi : g \mapsto \sum_{i=1}^k a_i \cdot \varphi_i(g)\,, 
\end{equation}
where the $a_i$ are rational numbers. Then there is an efficient algorithm $\hat{\mathbb{A}}$ which is equipped with oracle access to $\phi$ and which computes the coefficients $a_1,\dots, a_k$.
\end{theorem}
\noindent In our setting, Theorem~\ref{thm:dedekind_intro} yields what follows. First, let $\mathrm{G}$ be the set of all digraphs and $*$ be the directed tensor product; one can check that $(\mathrm{G},*)$ is indeed a semigroup. Second, for any fixed $\vec{H}$ consider the function $\vec{G}\mapsto\#\homs{\vec{H}}{\vec{G}}$; one can check this is a semigroup homomorphism into $\mathbb{Q}$. Using Theorem~\ref{thm:dedekind_intro}, we can prove:
\begin{lemma}\label{lem:complexity_monotonicity_intro}
There exists a deterministic algorithm $\mathbb{A}$ with the following specifications:
\begin{itemize}\itemsep0pt
    \item The input of $\mathbb{A}$ is a pair $(\vec{{G}}',\iota)$ where $\vec{{G}}'$ is a digraph and $\iota : \mathrm{G} \to \mathbb{Q}$.
    \item $\mathbb{A}$ is equipped with oracle access to the function 
    \[ \vec{G} \mapsto \sum_{\vec{F}} \iota(\vec{F}) \cdot \#\homs{\vec{F}}{\vec{G}} \,,\]
    where the sum is over all (isomorphism classes of) digraphs.
    \item The output of $\mathbb{A}$ is the list with elements $(\vec{F},\#\homs{\vec{F}}{\vec{{G}}'})$ for each $\vec{F}$ with $\iota(\vec{F})\neq 0$.
    \item For some computable function $f$ the running time of $\mathbb{A}$ is bounded by $f(|\iota|)\cdot |\vec{{G}}'|^{O(1)}$
    \item The outdegree of every digraph $\vec{G}$ on which $\mathbb A$ invokes the oracle is at most $f(|\iota|)\cdot d(\vec{{G}}')$ where $d(\vec{{G}}')$ is the maximum outdegree of $\vec{{G}}'$.
\end{itemize}
\end{lemma}

\noindent To understand the meaning of Lemma~\ref{lem:complexity_monotonicity_intro}, let $\iota(\vec F)=\asub_{\vec{H}}(\vec F)$ for all $\vec F \in \mathrm{G}$, see~\eqref{eq:subs_to_homs}. Then Lemma~\ref{lem:complexity_monotonicity_intro} says that, if $\mathbb{A}$ has oracle access to $\#\subs{\vec H}{\cdot}$, then $\mathbb{A}$ can compute $\#\homs{\vec H}{\vec G'}$ efficiently and by computing $\#\subs{\vec H}{\vec G}$ only for $\vec G$ of outdegree not larger than that of $\vec G'$. This yields a parameterised reduction from $\#\dirsubsprobd(\vec C)$ to $\#\dirhomsprobd(\vec C')$, where $\vec C'$ is the set of all digraphs $\vec F$ such that $\asub_{\vec{H}}(\vec F) \ne 0$ for some $\vec H \in \vec C$. As stated above, $\asub_{\vec{H}}(\vec F) \ne 0$ if and only if $\vec F$ is a quotient of $\vec H$. We conclude that computing $\#\subs{\vec{H}}{\vec{G}}$ is at least as hard as computing $\#\homs{\vec{F}}{\vec{G}}$ for each $\vec F$ that is a quotient of $\vec H$. In other words we have a parameterised reduction from $\#\dirsubsprobd(\vec C)$ to $\#\dirhomsprobd(\vec C')$ where $\vec C'$ is the set of all quotients of $\vec C$. Together with the converse reduction (see above) this tells us that $\#\dirsubsprobd(\vec C)$ is \emph{precisely as hard as} $\#\dirhomsprobd(\vec C')$ where $\vec C'$ is the set of all quotients of $\vec C$. Thus, classifying the complexity of $\#\dirsubsprobd(\vec C)$ boils down to understanding the complexity of $\#\dirhomsprobd(\vec C')$ where $C'$ is again the set of all quotients of $\vec C$. Answering this question turns out to be the most challenging task in this work.

\subsection{Upper bounds: a reduction to \#CSP}\label{sec:ub_intro}
To understand the complexity of $\#\dirhomsprobd(\vec C')$ where $\vec C$ is the set of all quotients of $\vec C$, we take two steps. First, we show that the problem can be reduced to \#CSP, the problem of counting the solutions to a constraint satisfaction problem. Second, we show that the fractional cover number of $\vec C$ bounds the fractional hypertree width of the \#CPS instances obtained from $\vec C'$, which makes the problem fixed-parameter tractable by existing results.

\subsubsection{A reduction to \#CSP} Let $\vec{H}$ and $\vec{G}$ be digraphs and let $d$ be the maximum outdegree of $\vec{G}$. Let furthermore $k=|\vec{H}|$ and $n=|\vec{G}|$. Recall that a source $S$ of $\vec{H}$ is a strongly connected component of $\vec{H}$ such that $S$ cannot be reached from any other strongly connected component. Let $S_1,\dots,S_\ell$ be the sources of $\vec{H}$, and let $s_i \in S_i$ for each $i\in[\ell]$. Finally, let $R_i$ be the set of all vertices of $\vec{H}$ that can be reached from $s_i$ via a directed path --- note that $S_i$ is fully contained in $R_i$.
Clearly each arc of $\vec{H}$ is fully contained in at least one of the $R_i$. Writing $\vec{H}[R_i]$ for the subgraph of $\vec{H}$ induced by $R_i$, one can see that every map $\varphi: V(\vec{H}) \to V(\vec{G})$ satisfies:
\begin{equation}\label{eq:towards_CSP}
    \varphi \in \homs{\vec{H}}{\vec{G}} \;\Leftrightarrow\; \forall i \in [\ell]: \varphi|_{R_i} \in \homs{\vec{H}[R_i]}{\vec{G}}\,,
\end{equation}
where $\varphi|_{R_i}$ is the restriction of $\varphi$ on $R_i$. In other words, $\varphi$ is a homomorphism if and only if it induces a partial homomorphism from $\vec{H}[R_i]$ for each $i\in[\ell]$.

The observation above allows us to reduce the computation of $\homs{\vec{H}}{\vec{G}}$ to counting the solutions of a certain constraint satisfaction problem. Start by fixing an arbitrary order over $V(\vec H)$, so that every $R_i$ appears as an ordered tuple. Now, for each $i\in[\ell]$, we enumerate all partial homomorphisms $\varphi|_{R_i} \in \homs{\vec{H}[R_i]}{\vec{G}}$. It is well known that this can be done in time $f(k,d) \cdot n^{O(1)}$: simply guess the image $v$ of $s_i$ in $V(\vec G)$, and perform a brute force search over the $d^{O(k)}$ vertices of $\vec G$ reachable from $v$ in $k$ steps~\cite{Chiba&1985,Bressan21,BeraGLSS22}. Now for every $i\in[\ell]$ consider the set of all (the images of) the maps in $\homs{\vec{H}[R_i]}{\vec{G}}$. This is a set of ordered tuples of vertices of $\vec G$, i.e., a relation over $V(\vec G)$. We denote this relation by $\mathrm{R}_i$. It is not hard to see that the homomorphisms from $\vec H$ to $\vec G$ are precisely those maps from $V(\vec H)$ to $V(\vec G)$ that for every $i \in [\ell]$ send $R_i$ to an element of $\mathrm{R}_i$, and that counting those maps is an instance of a counting constraint satisfaction problem (\#CSP). 

\subsubsection{Bounding the cost of solving \#CSP over quotients}
Recall the reachability hypergraph $\reach{\vec H}$ of $\vec H$: the hypergraph whose vertex set is $V(\vec H)$ and whose edge set is $\{R_i:i \in [\ell]\}$. A well-known result due to Grohe and Marx~\cite{GroheM14} states that counting the solutions to the CSP instance above is fixed-parameter tractable whenever $\reach{\vec H}$ has bounded fractional hypertreewidth, where the parameter is $|\vec H|$; in fact,~\cite{GroheM14} shows that there is an algorithm that solves the problem in time $f(k,d) \cdot |V(\vec G)|^{\mathsf{fhtw}(\mathcal{R})+O(1)}$.
Now recall from Section~\ref{sec:results} the fractional cover number $\rho^\ast(\vec{H})$ of $\vec{H}$. We prove:
\begin{lemma}\label{lem:subs_upper_intro}
Let $\vec{H}$ be a digraph, let $\vec{F}$ be a quotient graph of $\vec{H}$, and let $\reach{\vec{F}}$ be the reachability hypergraph of $\vec{F}$. Then $\mathsf{fhtw}(\reach{\vec{F}})\leq \rho^\ast(\vec{H})$.
\end{lemma}
\noindent The intuition behind the proof of Lemma~\ref{lem:subs_upper_intro} is that (i) taking the quotient of a digraph cannot increase its fractional cover number, and (ii) the fractional hypertreewidth of a hypergraph is bounded by its fractional edge cover number (which is the fractional cover number of $\vec H$).

Together with the observations above, this implies that we can compute $\#\subs{\vec H}{\vec G}$ in time $f(|\vec H|,d(\vec g)) \cdot |\vec G|^{\rho^*(\vec H)}$, thus proving Theorem~\ref{thm:main_intro_bounds} and the tractability part of Theorem Theorem~\ref{thm:main_intro} for $\#\dirsubsprobd$. It remains to prove the intractability part of Theorem~\ref{thm:main_intro}, which we do in the next sections.

Let us again consider $\vec{H}=\Delta_1^k$ as a toy example, that is, $\vec{H}$ is the disjoint union of $k$ triangles, each of which is cyclically oriented. We can use the principle of inclusion and exclusion to reduce the computation of $\#\subs{\vec{H}}{\vec{G}}$ to the computation of terms $\#\homs{\vec{F}}{\vec{G}}$ where $\vec{F}$ is a quotient of $\vec{H}$. Now, it can easily be observed that each quotient of $\vec{H}$ is a disjoint union of strongly connected components $S_1,\dots,S_\ell$. Unfolding our general reduction to $\#\textsc{CSP}$, for each of the strongly connected components $S$, we only have to guess the image $v$ of one vertex $s\in S$ in $\vec{G}$. Then the image of each additional vertex in $S$ must be reachable from $v$ by a directed path of length at most $k$. Since the outdegree of $\vec{G}$ is at most $d$, there are thus at most $d^{O(k)}$ possibilities for the images of the remaining vertices. Thus, for each strongly connected components $S$, we can compute $\#\homs{\vec{F}[S]}{\vec{G}}$ in time $d^{O(k)}\cdot |\vec{G}|$. Finally, we have $\#\homs{\vec{F}}{\vec{G}} = \prod_{i=1}^\ell \#\homs{\vec{F}[S_i]}{\vec{G}}$.

\subsection{Lower bounds}\label{sec:lb_intro}
The goal of this section is to prove that, roughly speaking, if $\rho^*(\vec{H})$ is large then $\vec{H}$ has a quotient $\vec{F}$ such that computing $\#\homs{\vec{F}}{\vec{G}}$ is hard when parameterized by $|\vec F|+d(\vec G)$. To this end we seek a reduction from \#CSP to $\#\dirhomsprobd$, i.e., in the opposite direction of Section~\ref{sec:ub_intro}. However, while that direction was relatively easy, since every digraph can be easily encoded as a set of relations, the direction we seek here is significantly harder. Indeed, it is not clear at all how an instance of \#CSP can be ``encoded'' as a pair of directed graphs $(\vec H,\vec G)$ if we can choose $\vec H$ only from the class $\vec C$ for which we want to prove hardness.

\subsubsection{Encoding \#CSP instances via canonical DAGs}
To bypass the obstacle above, we start by considering classes of \emph{canonical DAGs}. A digraph $\vec{H}$ is a canonical DAG if it is acyclic and every vertex is either a source (i.e., it has indegree $0$) or a sink (i.e., it has outdegree $0$). Note that this implies that $\vec H$ is bipartite, with (say) all sources on the left and all sinks on the right. If $\vec C$ is a class of canonical DAGs, then it is easy to reduce \#CSP to $\#\dirhomsprobd(\vec C)$ while preserving all parameters. To see why, let $(\scH,\scG)$ be a pair of hypergraphs (the instance of \#CSP).  Define $\vec H$ by letting $V(\vec H)=V(\scH) \cup \{x_e: e \in E(\scH)\}$, and adding $(x_e,u)$ to $E(\vec H)$ for every $e \in E(\vec \scH)$ and every $u \in e$. Define $\vec G$ similarly as a function of $\scG$. One can then show, using the \emph{color-prescribed} version of homomorphism counting (defined in Section~\ref{sec:prelims}), that $\#\homs{\scH}{\scG}$ can be computed in FPT time with $|\scH|$ as a parameter if we can compute $\#\homs{\vec H}{\vec G}$ in FPT time with $|\vec H|+d(\vec G)$ as a parameter.

Recall then the \emph{contour} $\reach{\vec H}$ of $\vec H$ from Section~\ref{sec:results}. It is immediate to see that, if $\vec H$ is a canonical DAG obtained from $\scH$ as described above, then $\reach{\vec H}=\scH$. This is precisely the intuitive role of the contour --- to encode the structure of the reachability sets of $\vec H$ (ignoring sources).
Indeed, using contours we can then state our main reduction. Let $\vec C$ be a class of canonical DAGs, and let $\#\textsc{CSP}(\reduct{\vec{C}})$ be the restriction of $\#\textsc{CSP}$ to instances whose left-hand hypergraph (i.e., $\scH$) is isomorphic to a contour of $\reduct{\vec{C}}$. Using as a starting point a reduction due to Chen et al.\ \cite{Chenetal19}, we prove that $\#\textsc{CSP}(\reduct{\vec{C}})$ reduces to $\#\dirhomsprobd(\vec C)$ via parameterised Turing reductions. Now, under ETH, $\#\textsc{CSP}(\reduct{\vec{C}}) \notin \ccFPT$ when the adaptive width of $\reduct{\vec{C}}$ is unbounded~\cite{Marx13}, unless ETH fails. By the reduction above, then, we obtain:
\begin{lemma}\label{lem:hom_hard_cont_aw}
$\#\dirhomsprobd(\vec C) \notin \ccFPT$ for every class $\vec C$ of canonical DAGs whose contours have unbounded adaptive width, unless ETH fails.
\end{lemma}
\noindent We now seek to lift this result from canonical DAGs to arbitrary directed graphs.

\subsubsection{Lifting hardness to arbitrary digraphs via Monotone Reversible Minors}
Starting from Lemma~\ref{lem:hom_hard_cont_aw}, we prove a hardness result for $\#\dirhomsprobd(\vec C)$ for general classes of digraphs $\vec C$. To this end we need to reduce from $\#\dirhomsprobd(\vec C)$ to $\#\dirhomsprobd(\vec C')$ where $\vec C'$ is a class of canonical DAGs, so that we can apply Lemma~\ref{lem:hom_hard_cont_aw}; clearly, the reduction must imply that $\#\dirhomsprobd(\vec C')$ has unbounded adaptive width.

Towards this end we introduce a kind of graph minors for digraphs, which we call \emph{monotone reversible (MR) minors}. A digraph $\vec H'$ is a MR minor of $\vec H$ if it is obtained from $\vec H$ by a sequence of the following operations:
\begin{itemize}\itemsep0pt
    \item deleting a sink, i.e., a strongly connected component from which no other vertices can be reached
    \item deleting a loop
    \item contracting an arc
\end{itemize}
\noindent Note that, unlike standard minors, deletion of arbitrary vertices and arbitrary arcs are not allowed. This allows us to prove:
\begin{lemma}\label{lem:intro_MR}
Let $\vec{C}$ be a class of digraphs and let $\vec{D}$ be a class of MR minors of $\vec{C}$. Then there exists a parameterised Turing reduction from $\#\dirhomsprobd(\vec D)$ to $\#\dirhomsprobd(\vec C)$.
\end{lemma}
\noindent Lemma~\ref{lem:intro_MR} explains the ``reversible'' part of MR minors---we can efficiently ``revert'' the operations that yielded a MR of a digraph; for the ``monotone'' see the next section. The heart of the proof of Lemma~\ref{lem:intro_MR} proves the claim for the color-prescribed version of the problems; this implies the reduction for the original problems via standard interreducibility arguments arguments. 
As a consequence of Lemma~\ref{lem:intro_MR} we obtain:
\begin{lemma}\label{lem:main_hardness_homs_intro}
Let $\vec{C}$ be a recursively enumerable class of digraphs and let $\vec{C}'$ be a class of canonical DAGs that are MR minors of digraphs in $\vec{C}$. If $\reduct{\vec{C}'}$ has unbounded adaptive width then $\#\dirhomsprobd(\vec{C}) \notin \ccFPT$, unless ETH fails.
\end{lemma}

\subsubsection{Lifting hardness from homomorphisms to subgraphs}
Recall the arguments of Section~\ref{sec:hom_basis}: to prove that $\#\dirsubsprob(\vec C)$ is hard when $\rho^*(\vec C)=\infty$, we essentially have to prove that every digraph $\vec{H}$ with high fractional cover number $\rho^*(\vec H)$ has a quotient $\vec{F}$ that is hard. By the arguments of the previous section, to show that such a quotient $\vec{F}$ is hard it is enough to show that $\vec{F}$ has an MR minor $\vec{F}'$ which is a canonical DAG whose contour $\reduct{\vec{F}'}$ has high adaptive width. We indeed prove that such a quotient exists. To this end, we consider two cases. Recall that $\alpha(\scH)$ and $\alpha^*(\scH)$ denote respectively the independence number and the fractional independence number of a hypergraph $\scH$.
\begin{itemize}
    \item[(a)] $\alpha(\reduct{\vec{H}})$ is large. In this case we can show that $\vec{H}$ contains a large matching whose edges are ``isolated enough'' for us to construct a quotient $\vec{F}$ that admits, as MR minor, the $1$-subdivision $\vec{F}'$ of a large complete graph, where the arcs of $\vec F'$ are directed away from the subdivision vertices. It is easy to see that $\vec F'$ is a canonical DAG, and that $\reduct{\vec F'}$ is the complete graph itself, which has large adaptive width.
    \item[(b)] $\alpha(\reduct{\vec{H}})$ is small. We then choose as quotient $\vec{F}$ the graph $\vec{H}$ itself. Recall that, by definition, $\rho^*(\vec{H})=\rho^*(\reach{\vec H})$. By LP duality the fractional cover number equals the fractional independence number, that is, $\rho^*(\reach{\vec H})=\alpha^*(\reach{\vec H})$. Using Theorem~\ref{thm:independence_Gap_Intro}, we deduce that the adaptive width of $\reach{\vec H}$ is within constant factors of $\alpha^*(\reach{\vec H})$, and thus of $\rho^*(\vec H)$. By carefully exploiting this fact, we can explicitly construct an MR minor $\vec{F}'$ of $\vec{H}$ that is both a canonical DAG and has high adaptive width.
\end{itemize}
Thus, in both cases we can show that if $\rho^*(\vec H)$ is large then $\vec H$ admits an MR minor that is a canonical DAG of large treewidth. Formally, we obtain:
\begin{lemma}\label{lem:subs_to_MR_aw}
Let $\vec C$ be a class of digraphs such that $\rho^*(\vec C )=\infty$. Then the class $\vec C'$ of all canonical DAGs that are MR minors of quotients of $\vec C$ has unbounded adaptive width.
\end{lemma}
As a consequence, this implies that $\#\dirsubsprobd(\vec C) \notin \ccFPT$ whenever $\rho^*(\vec C )=\infty$, unless ETH fails. This concludes the overview of the proof of the lower bounds for $\#\dirsubsprobd(\vec C)$. 

Before proceeding with the next section, we wish to point out that, in the course of our proofs, we will also see that for computing the fractional cover number and the source count of a directed graph $\vec{H}$, it is \emph{always} sufficient to consider the DAG $\Hsim$ (see Figure~\ref{fig:main_intro}). In other words, the complexity of $\#\dirsubsprobd(\vec C)$ and $\#\dirindsubsprobd(\vec{C})$ solely depends on the structure of the DAGs obtained by contracting the strongly connected components of the patterns in $\vec{C}$.

\subsection{Related Work and Outlook}\label{sec:related_work}
Our results are closely related to the recent surge of works on pattern counting in degenerate graphs~\cite{Bera-ITCS20,BeraS20,Bera-SODA21,BressanR21,GishbolinerLSY22,BeraGLSS22}: An undirected graph $G$ has degeneracy $d$ if there is an acyclic orientation $\vec{G}$ of $G$ with outdegree at most $d$. In the context of pattern counting in degenerate graphs, one is given undirected graphs $H$ and $G$, and the goal is to compute the number of copies (resp.\ induced copies) of $H$ in $G$, parameterised by the size of the pattern $H$ and the degeneracy $d$ of $G$. These problems have been completely classified with respect to linear time tractability~\cite{BeraGLSS22} and with respect to fixed-parameter tractability~\cite{BressanR21}. 

The crucial difference to the results in this work is that, in our setting, the orientations of $H$ and $G$ are already fixed. Notably, this increases the set of tractable instances when compared to the degenerate setting: For example, counting copies of an undirected graph $H$ in an undirected graph $G$, parameterised by $|H|$ and the degeneracy $d$ of $G$, is FPT if and only if the induced matching number of $H$ is small~\cite{BressanR21}. Now fix acyclic orientations $\vec{H}$ and $\vec{G}$ of $H$ and $G$ such that the outdegree of $\vec{G}$ is at most $d$. One might think that the directed problem also is FPT if and only if $H$ (i.e., the underlying undirected graph of $\vec{H}$) has small induced matching number. However, this is not true: We have shown in this work that we can count copies of a digraph $\vec{H}$ in a digraph $\vec{G}$ in FPT time (parameterised by $|\vec{H}|$ and $d(\vec{G})$) if and only if the fractional cover number of $\vec{H}$ is small --- we will see that this holds even if the host $\vec{G}$ is a DAG. While $H$ having small induced matching number certainly implies that the fractional cover number of $\vec{H}$ is small, the other direction does not hold: For example, if $\vec{H}$ contains a source that is adjacent to all non-sources, then the fractional cover number is $1$, although the induced matching number can be arbitrarily large.

This work also sheds some new light on the problem of counting homomorphisms into degenerate graphs: Bressan~\cite{Bressan21} has shown that we can count homomorphisms from $H$ to $G$ in FPT time (parameterised by $|H|$ and the degeneracy of $G$) if the so-called \emph{dag treewidth} of $H$ is small; it is currently open whether the other direction holds as well~\cite{BressanR21,BeraGLSS22}. The dag treewidth of $H$ is just the maximum (non-fractional) hypertreewidth of the reachability hypergraph of any acyclic orientation of $\vec{H}$. Our reduction to $\#\textsc{CSP}$ implies that it is sufficient for the reachability hypergraphs to have small fractional hypertreewidth, which yields a fractional version of dag treewidth. However, it is not clear whether unbounded dag treewidth and unbounded fractional dag treewidth are equivalent, the reason for which is the fact that all acyclic orientations have to be considered. We leave this as an open problem for future work.

\pagebreak

\section{Preliminaries}\label{sec:prelims}
Given a set $S$, we set $S^2= S \times S$, and we write $S^{(2)}$ for the set of all unordered pairs of distinct elements of $S$. Let $f:A\times B \to C$ be a function and let $a\in A$. We write $f(a,\ast) : B \to C$ for the function that maps $b\in B$ to $f(a,b)$.

\subsection{Graphs and Directed Graphs}
We denote graphs by $F,G,H$, and directed graphs by $\vec{F},\vec{G},\vec{H}$. Graphs and digraphs are encoded via adjacency lists, and we write $|G|$ (resp.\ $|\vec{G}|$) for the length of the encoding. In the remainder of the paper we will call directed graphs just ``digraphs''. We use $\{u,v\}$ for undirected edges, and $(u,v)$ for directed edges, which we also call \emph{arcs}. Furthermore, we will use $C$ to denote classes of graphs, and $\vec{C}$ to denote classes of digraphs. Our graphs do not contain multi-edges; however, we allow forward-backward arcs $(u,v)$ and $(v,u)$ in digraphs. Additionally, our undirected graphs do not contain loops (edges from a vertex to itself) unless stated otherwise. For technical reasons, we will allow loops in digraphs.\footnote{We will see and state explicitly, that all of our hardness results will also entail corresponding hardness in the restricted case of digraphs without loops.} 

A directed acyclic graph (DAG) is a digraph without directed cycles. A \emph{source} of a DAG is a vertex with in-degree $0$, and a \emph{sink} of a DAG is a vertex with outdegree $0$ (and an isolated vertex is simultaneously a source and and a sink).
Given a directed (not necessarily acyclic) graph $\vec{H}$ and a vertex $v\in V(\vec{H})$, we write $\reachv{v}$ for the set of vertices reachable from $v$ by a directed path; this includes $v$ itself. Given a set of vertices $S\subseteq V(\vec{H})$, we set
\begin{align*}
    \reachv{S}:=\bigcup_{v\in S} \reachv{v}\,.
\end{align*}
Let $H$ be a graph and $\sigma$ a partition of $V(H)$. The \emph{quotient graph} of $H$ w.r.t.\ $\sigma$, denoted by $H/\sigma$, is defined as follows: $V(H/\sigma)$ consists of the blocks of $\sigma$, and $\{B_1,B_2\} \in E(H/\sigma)$ if and only if $\{v_1,v_2\}\in E(H)$ for some $v_1\in B_1$ and $v_2\in B_2$. If $H/\sigma$ does not contain loops then it is called a \emph{spasm}~\cite{CurticapeanDM17}. These definitions can be adapted in the obvious way for digraphs.

\begin{definition}[The DAG $\Hsim$]
Let $\vec{H}$ be a digraph and let $x,y\in V(\vec{H})$. We denote by $\sim$ the equivalence relation over $V(\vec H)$ whose classes are the strongly connected components of $\vec H$. 
We denote by $\Hsim$ the DAG obtained by deleting loops from the quotient of $\vec H$ with respect to the partition of $V(\vec{H})$ given by $\sim$.
\end{definition}

Next we introduce some notions that will be used in our classifications. 
\begin{definition}[Directed split]
The \emph{directed split} $\dsplit{H}$ of a graph $H$ is the digraph obtained from the $1$-subdivision of $H$ by orienting all edges towards $V(H)$. 
\end{definition}

\paragraph*{Homomorphisms and Colourings}
A \emph{homomorphism} from $H$ to $G$ is a map $\varphi : V(H) \to V(G)$ such that $\{\varphi(u),\varphi(v)\} \in E(G)$ whenever $\{u,v\} \in E(H)$. The set of all homomorphisms from $H$ to $G$ is denoted by $\homs{H}{G}$. An \emph{$H$-colouring} of $G$ is a homomorphism $c \in \homs{G}{H}$. An $H$-\emph{coloured} graph is a pair $(G,c)$ where $G$ is a graph and $c$ an $H$-colouring of $G$. A homomorphism $\varphi \in \homs{H}{G}$ is color-prescribed (by $c$) if $c(\varphi(v))=v$ for all $v \in V(H)$. We write $\homs{H}{(G,c)}$ for the set of all homomorphisms from $H$ to $G$ color-prescribed by $c$. These definitions can be adapted for digraphs in the obvious way; we emphasise that a homomorphism $\varphi$ between digraphs must preserve the direction of the arcs, i.e., $(\varphi(u),\varphi(v)) \in E(\vec{G})$ whenever $(u,v) \in E(\vec{H})$.

\paragraph*{Subgraphs and Induced Subgraphs}
A \emph{subgraph} of a graph $G=(V,E)$ is a graph with vertices $\hat{V}\subseteq V$ and edges $\hat{E} \subseteq \hat{V}^{(2)}\cap E$. We write $\subs{H}{G}$ for the set of all subgraphs of $G$ that are isomorphic to $H$. Similarly, a subgraph of a digraph $\vec{G}=(V,E)$ is a digraph with vertices $\hat{V}\subseteq V$ and arcs $\hat{E} \subseteq \hat{V}^{2}\cap E$, and we denote by $\subs{\vec{H}}{\vec{G}}$ the set of all subgraphs of $\vec{G}$ that are isomorphic to $\vec{H}$.

Given a subset of vertices $S$ of a graph $H$, we write $H[S]$ for the graph \emph{induced} by $S$, that is $V(H[S]):=S$ and $E(H[S]):= E(H)\cap S^{(2)}$. The subgraph $\vec{H}[S]$ of a digraph induced by a vertex set $S$ is defined correspondingly: $V(\vec{H}[S]):=S$ and $E(\vec{H}[S]):= E(\vec{H})\cap S^{2}$. We write $\indsubs{H}{G}:=\{S \subseteq V(G)~|~G[S]\cong H\}$ for the set of all induced subgraphs of $G$ that are isomorphic to $H$. Again, the notion $\indsubs{\vec{H}}{\vec{G}}$ is defined similarly for digraphs.

\subsection{Hypergraphs}
A hypergraph is a pair $\scH=(V,E)$ where $V$ is a finite set and $E \subseteq 2^V \setminus \{\emptyset\}$. The elements of $E$ are called hyperedges or simply edges. Given $X \subseteq V$, the subhypergraph of $\scH$ induced by $X$ is the hypergraph $\scH[X]$ with vertex set $X$ and edge set $\{e \cap X : e \in \scE\} \setminus \{\emptyset\}$. The \emph{arity} of a hypergraph $\mathcal{G}$, denoted by $a(\mathcal{G})$ is the maximum cardinality of a hyperedge. We denote hypergraphs with the symbols $\scH,\scG,\ldots$. 

\begin{definition}[Reachability Hypergraph]
Let $\vec{H}$ be a digraph, let $S_1,\dots,S_k\subseteq V(\vec{H})$ be the sources of $\Hsim$, and for each $i\in[k]$ let $s_i\in S_i$. The \emph{reachability hypergraph} $\reach{\vec{H}}$ of $\vec{H}$ has vertex set $V(\vec{H})$ and edge set $\{e_i=\reachv{s_i} : i \in [k]\}$.
\end{definition}
\noindent Note that $\reach{\vec H}$ is well defined, since $S_i$ is a strongly connected component of $\vec H$ and so the choice of $s_i \in S_i$ is irrelevant. If $\vec H$ is a DAG, then $\reach{\vec H}$ is the reachability hypergraph in the usual sense.

The following special case of DAGs, defined via reachability hypergraphs, will turn out to be crucial for our lower bounds.
\begin{definition}[Canonical DAGs]
Let $\mathcal{R}$ be a reachability hypergraph. For every $e \in E(\mathcal{R})$ fix some $s_e \in V(\mathcal{R})$ such that $s_e$ is contained only in $e$. The \emph{canonical DAG} $\vec{H}$ of $\mathcal{R}$ is defined by $V(\vec{H})=V(\mathcal{R})$ and $E(\vec{H})=\{(s_e,v) : e \in E(\mathcal R), v \in e \setminus \{s_e\}\}$. \end{definition}

Note that if $\vec H$ is a canonical DAG of $\mathcal R$ then $\mathcal{R}$ is the reachability hypergraph of $\vec H$. If a DAG $\vec H$ is the canonical DAG of its own reachability hypergraph, then we just say $\vec H$ is a canonical DAG. Equivalently, a DAG $\vec{H}$ is a canonical DAG if every vertex is either a source or a sink.

\begin{definition}[Contour]\label{def:contour}
Let $\vec{H}$ be a digraph, let $S_1,\dots,S_k\subseteq V(\vec{H})$ be the sources of $\Hsim$, and for each $i\in[k]$ let $s_i\in S_i$.
The \emph{contour} $\reduct{\vec{H}}$ of $\vec{H}$ is the hypergraph obtained from $\reach{\vec{H}}$ by deleting $S_i$ from $e_i$ for each $i\in[k]$.
\end{definition}

\paragraph*{Invariants and Width Measures}
In what follows, we are using the notation of~\cite{GroheM14} and~\cite{Marx13}, and we recall the most important definitions.

\begin{definition}[Tree decompositions]
Let $\mathcal{H}$ be a hypergraph. A \emph{tree decomposition} of $\mathcal{H}$ is a pair of a tree $\mathcal{T}$ and a set of subsets of $V(\mathcal{H})$, called \emph{bags}, $\mathcal{B}=\{B_t\}_{t\in V(\mathcal{T})}$ such that the following conditions are satisfied:
\begin{enumerate}
    \item $\bigcup_{t\in V(\mathcal{T})} B_t = V(\mathcal{H})$.
    \item For every hyperedge $e\in E(\mathcal{H})$ there is a bag $B_t$ such that $e\subseteq B_t$.
    \item For every vertex $v\in V(\mathcal{H})$, the subgraph $\mathcal{T}[\{t~|~v\in B_t\}]$ of $\mathcal{T}$ is connected.
\end{enumerate}
\end{definition}

\begin{definition}[$f$-width]
Let $\mathcal{H}$ be a hypergraph, let $f:2^{V(\mathcal{H})}\rightarrow\mathbb{R}_{+}$, and let $(\mathcal{T},\mathcal{B})$ be a tree decomposition of $\mathcal{H}$. The $f$-\emph{width} of $(\mathcal{T},\mathcal{B})$ is defined as follows:
\[\fwidth{f}(\mathcal{T},\mathcal{B}) := \max_{t\in V(\mathcal{T})}f(B_t) \,.\]
The $f$-width of $\mathcal{H}$ is the minimum $f$-width of any tree decomposition of $\mathcal{H}$.
\end{definition}
\noindent For example, the \emph{treewidth} of a (hyper-)graph is just its $f$-width for the function $f(B):=|B|-1$.

Given a hypergraph $\mathcal{H}$ and a vertex-subset $X$ of $\mathcal{H}$, the \emph{edge cover number} $\rho_\mathcal{H}(X)$ of $X$ is the minimum number of hyperedges of $\mathcal{H}$ required to cover each vertex in $X$. The \emph{edge cover number} of $\mathcal{H}$, denoted by $\rho(\mathcal{H})$, is defined as $\rho_\mathcal{H}(V(\mathcal{H}))$.

A \emph{fractional} version of the edge cover number of defined similarly: Given $\mathcal{H}$ and $X$ as above, a function $\gamma:E(\mathcal{H})\rightarrow[0,\infty]$ is a fractional edge cover of $X$ if for each $v\in X$ we have $\sum_{e: v\in e} \gamma(e) \geq 1$. The fractional edge cover number $\rho_\mathcal{H}^\ast(X)$ of $X$ is defined to be the minimum of $\sum_{e\in E(\mathcal{H})} \gamma(e)$ among all fractional edge covers $\gamma$ of $X$. The \emph{fractional edge cover number} of $\mathcal{H}$, denoted by $\rho^\ast(\mathcal{H})$, is defined as $\rho^\ast_\mathcal{H}(V(\mathcal{H}))$.

\begin{definition}[Generalised and Fractional Hyper-Treewidth]
The generalised hyper-treewidth of $\mathcal{H}$, denoted by $\mathsf{htw}(\mathcal{H})$, is the $\rho_\mathcal{H}$-width of $\mathcal{H}$. The fractional hyper-treewidth of $\mathcal{H}$, denoted by $\mathsf{fhtw}(\mathcal{H})$, is the $\rho_\mathcal{H}^\ast$-width of $\mathcal{H}$.
\end{definition}

An \emph{independent set} of a hypergraph $\mathcal{H}$ is a set $I$ of vertices such that, for each $u,v\in I$ with $u\neq v$, there is no hyperedge containing $u$ and $v$. The \emph{independence number} of $\mathcal{H}$, denoted by $\alpha(\mathcal{H})$, is the size of a maximum independent set of $\mathcal{H}$.
A \emph{fractional independent set} of a hypergraph $\mathcal{H}$ is a mapping $\alpha^\ast: V(\mathcal{H}) \rightarrow [0,1]$ such that for each $e\in E(\mathcal{H})$ we have
\[\sum_{v\in e} \alpha^\ast(v) \leq 1 \,. \]
The \emph{fractional independence number} of $\mathcal{H}$, denoted by $\alpha^\ast(\mathcal{H})$ is the maximum of $\sum_{v\in e} \alpha^\ast(v)$ among all fractional independent sets $\alpha^\ast$.
For a subset $X$ of vertices in $V(\mathcal{H})$, we set $\alpha^\ast(X)=\sum_{v\in X}\alpha^\ast(v)$. We remark that, by LP duality, the fractional independence number and the fractional edge cover number are equal (see,~\cite{fractionalGraphTheory}):

\begin{fact}\label{fact:LP_duality}
Let $\mathcal{H}$ be a hypergraph. We have $\alpha^\ast(\mathcal{H})=\rho^\ast(\mathcal{H})$.
\end{fact}

We continue with the notion of adaptive width, which is equivalent\footnote{Here, ``equivalent'' means that a class of hypergraphs has bounded adaptive width if and only if it has bounded submodular width.} to submodular width as shown by Marx~\cite{Marx13}. 

\begin{definition}[Adaptive width]
Let $\mathcal{H}$ be a hypergraph. The \emph{adaptive width} of $\mathcal{H}$ is 
\[\aw(\mathcal{H}) := \sup\left\{\alpha^\ast\text{-width}(\mathcal{H}) \mid \alpha^\ast \text{ is a fractional independent set of } \mathcal{H} \right\} \,.\]
\end{definition}

\begin{lemma}[\cite{Marx13}]\label{lem:width_measures}
Let $\mathcal{C}$ be a class of hypergraphs. Then
\begin{align*}
    ~&~~ \mathcal{C} \text{ has unbounded adaptive width}\\
    \Rightarrow &~~ \mathcal{C} \text{ has unbounded fractional hyper-treewidth}\\
    \Rightarrow &~~ \mathcal{C} \text{ has unbounded generalised hyper-treewidth}\,.
\end{align*}
Furthermore, all of the above implications are strict, that is, there are classes $\mathcal{C}_1$ and $\mathcal{C}_2$ such that $\mathcal{C}_1$ has bounded adaptive width but unbounded fractional hyper-treewidth, and $\mathcal{C}_2$ has bounded fractional hyper-treewidth but unbounded generalised hyper-treewidth.
\end{lemma}

Throughout this paper, we will be interested in the independence number $\alpha$, the fractional independence number $\alpha^\ast$, and the adaptive width $\aw$ of the contours of digraphs $\vec{H}$. The following lemma shows that it is equivalent to consider the contour of $\Hsim$ for those invariants. 
\begin{lemma}\label{lem:invariants_are_invariant}
Let $\vec{H}$ be a digraph. We have
\begin{enumerate}
    \item $\alpha(\reduct{\vec{H}})=\alpha(\reduct{\Hsim})$.
    \item $\alpha^\ast(\reduct{\vec{H}})=\alpha^\ast(\reduct{\Hsim})$ and $\rho^\ast(\reduct{\vec{H}})=\rho^\ast(\reduct{\Hsim})$.
    \item $\aw(\reduct{\vec{H}})=\aw(\reduct{\Hsim})$.
\end{enumerate}
\end{lemma}
\begin{proof}
Recall that $\Hsim$ is obtained from $\vec{H}$ by contracting each strongly connected component into a single vertex. In the reachability hypergraph, and thus in the contour, this corresponds to identifying blocks of vertices $B=\{v_1,\dots,v_\ell\}$ satisfying that all vertices in $B$ are contained in the same (non-empty) set of hyperedges, that is, there is a non-empty set of hyperedges $E_B$ such that for each $i\in[\ell]$, the set of hyperedges containing $v_i$ is $E_B$. 

Hence it is sufficient to show that identifying the vertices in $B$ --- we call the resulting vertex $v_B$ --- does not change any of the invariants $\alpha$, $\alpha^\ast$, and $\aw$. For what follows, let us write $\mathcal{H}$ for the contour $\reduct{\vec{H}}$ of $\vec{H}$, and let us write $\mathcal{H}'$ for the hypergraph obtained from $\mathcal{H}$ obtained by contracting $B$ to $v_B$. Furthermore, let $E'_B$ be the set of hyperedges that contain $v_B$ and observe that $E'_B$ can be obtained from $E_B$ by contracting $B$ to $v_B$ for each edge $e\in E_B$. 
\begin{enumerate}
    \item Goal: $\alpha(\mathcal{H})=\alpha(\mathcal{H}')$. Let $S\subseteq V(\mathcal{H})$ be a maximum independent set of $\mathcal{H}$. Note that at most one vertex in $B$ can be contained in $S$. If no vertex of $S$ is contained in $B$, then $S$ is an independent set of $\mathcal{H}'$. Otherwise, assume $v_i\in S$ for some $i\in[\ell]$. Then $S':= (S\setminus\{v_i\})\cup \{v_B\}$ is an independent set of $\mathcal{H}'$. This shows that $\alpha(\mathcal{H})\leq \alpha(\mathcal{H}')$. 
    
    For the other direction, let $S'$ be a maximum independent set of $\mathcal{H}'$. If $v_B\in S'$, then we set $S:=(S'\setminus\{v_B\})\cup \{v_1\}$. Otherwise, we set $S:=S'$. Clearly, $S$ is an independent set of $\mathcal{H}$ and thus  $\alpha(\mathcal{H})\geq \alpha(\mathcal{H}')$.
    \item \label{bul:2} Goal: $\alpha^\ast(\mathcal{H})=\alpha^\ast(\mathcal{H}')$ (Note that this is equivalent to $\rho^\ast(\mathcal{H})=\rho^\ast(\mathcal{H}')$ by Fact~\ref{fact:LP_duality}). Let $\mu$ be a fractional independent set of $\mathcal{H}$ of maximum weight. Define
    \[\mu'(v) := \begin{cases} \sum_{i=1}^\ell \mu(v_i) & v=v_B \\ \mu(v) & v\neq v_B \end{cases} \,.\]
    We claim that $\mu'$ is a fractional independent set of $\mathcal{H}'$. Let $e'\in E(\mathcal{H}')$ and let $e$ be the corresponding edge in $\mathcal{H}$, that is, $e=e'$ if $e'\notin E'_B$, and $e'$ is obtained from $e$ by contracting $B$ into $v_B$ otherwise. Depending on whether $e\in E_B$ we have that either $B\subseteq e$ or $B\cap e =\emptyset$. In both cases, by definition of $\mu'$, we have that $\sum_{v'\in e'} \mu'(v') = \sum_{v \in e} \mu(v) \leq 1$. Thus $\mu'$ is a fractional independent set of $\mathcal{H}'$. Since, clearly, $\mu'$ has the same total weight as $\mu$, we have that $\alpha^\ast(\mathcal{H})\leq\alpha^\ast(\mathcal{H}')$.
    
    For the other direction, let $\mu'$ be a fractional independent set of $\mathcal{H}'$ of maximum weight. Define 
    \[ \mu(v) := \begin{cases} \mu'(v_B) & v=v_1 \\ 0 & v\in B\setminus\{v_1\} \\ \mu'(v) & v\notin B \end{cases} \,.\]
    We claim that $\mu$ is a fractional independent set of $\mathcal{H}$. Similarly as in the first direction, let $e\in E(\mathcal{H})$ and let $e'$ be the corresponding edge in $\mathcal{H}'$. Again, depending on whether $e\in E_B$, we have that either $B\subseteq e$ or $B\cap e =\emptyset$, and in both cases, by definition of $\mu$, we have $ \sum_{v \in e} \mu(v) = \sum_{v'\in e'} \mu'(v') \leq 1$. Thus $\mu$ is a fractional independent set of $\mathcal{H}$. Since, clearly, $\mu$ has the same total weight as $\mu'$, we have that $\alpha^\ast(\mathcal{H})\geq\alpha^\ast(\mathcal{H}')$.
    \item Goal: $\aw(\mathcal{H})=\aw(\mathcal{H}')$. Any tree decomposition $(\mathcal{T},\mathcal{B})$ of $\mathcal{H}$ can be transformed to a tree decomposition $(\mathcal{T}',\mathcal{B}')$ of $\mathcal{H}'$ as follows: In each bag that contains a vertex in $B$, we delete all vertices in $B$ and add $v_B$. Clearly, the union over all bags in $\mathcal{B}'$ is the set of all vertices of $\mathcal{H}'$, and each hyperedge of $\mathcal{H}'$ is fully contained in some bag. For the last condition necessary for $(\mathcal{T}',\mathcal{B}')$ being a tree decomposition of $\mathcal{H}'$, we have to show that for each $v'\in V(\mathcal{H}')$, the subgraph $\mathcal{T}'_v$ of $\mathcal{T}$ consisting of the bags containing $v'$ is connected. If $v'\neq v_B$ then this property immediately follows from $(\mathcal{T},\mathcal{B})$ being a tree decomposition. For $v=v_B$ we use that $\mathcal{T}_{v_i}$ is connected for each $i\in[\ell]$ since $(\mathcal{T},\mathcal{B})$ is a tree decomposition of $\mathcal{H}$. Now $\mathcal{T}'_{v_B}$ corresponds, by definition of $(\mathcal{T}',\mathcal{B}')$, to the union of all the $\mathcal{T}_{v_i}$. However, since $E_B\neq \emptyset$, there is a hyperedge $e$ of $\mathcal{H}$ fully containing $B$. Since $e$ must be fully contained in a bag of $(\mathcal{T},\mathcal{B})$, all of the $\mathcal{T}_{v_i}$ overlap in at least one vertex, and thus $\mathcal{T}'_{v_B}$ is connected, proving that $(\mathcal{T}',\mathcal{B}')$ is indeed a tree decomposition of $\mathcal{H}'$. For what follows, let $\tau$ denote the function that maps a tree decomposition $(\mathcal{T},\mathcal{B})$ of $\mathcal{H}$ to a tree decomposition $(\mathcal{T}',\mathcal{B}')$ of $\mathcal{H}'$ as defined above.
    
    In the other direction, each tree decomposition $(\mathcal{T}',\mathcal{B}')$ of $\mathcal{H}'$ can be made a tree decomposition of $\mathcal{H}$ by substituting $v_B$ with the vertices in $B$. In this direction, it is clear that this yields a tree decomposition $(\mathcal{T},\mathcal{B})$ of $\mathcal{H}$. For what follows, let $\tau'$ denote the function that maps a tree decomposition $(\mathcal{T}',\mathcal{B}')$ of $\mathcal{H}'$ to a tree decomposition $(\mathcal{T},\mathcal{B})$ of $\mathcal{H}$ as defined above.
    
    Now let $\aw(\mathcal{H})= a$. We prove that $\aw(\mathcal{H}')\leq a$. Let $\mu'$ be a fractional independent set of $\mathcal{H}'$. We have to show that there is a tree decomposition $(\mathcal{T}',\mathcal{B}')$ of $\mathcal{H}'$ with $\mu'\mhyphen\mathsf{width}(\mathcal{T}',\mathcal{B}')\leq a$. To this end, let $\mu$ be the fractional independent of $\mathcal{H}$ obtained from $\mu$ as in~\ref{bul:2}. Since $\aw(\mathcal{H})= a$, there exists a tree decomposition $(\mathcal{T},\mathcal{B})$ of $\mathcal{H}$ with $\mu\mhyphen\mathsf{width}(\mathcal{T},\mathcal{B})\leq a$. Let $(\mathcal{T}',\mathcal{B}') := \tau(\mathcal{T},\mathcal{B})$. By definition of $\tau$ and $\mu$, the $\mu'\mhyphen\mathsf{width}$ of $(\mathcal{T}',\mathcal{B}')$ is at most the $\mu\mhyphen\mathsf{width}$ of $(\mathcal{T},\mathcal{B})$, which is at most $a$, concluding the first direction.
    
    For the second direction, let $\aw(\mathcal{H}')= a'$. We prove that $\aw(\mathcal{H})\leq a'$ similarly as in the first direction: Starting with a fractional independent set $\mu$ of $\mathcal{H}$, we consider $\mu'$ as constructed in~\ref{bul:2}, and we obtain a tree decomposition $(\mathcal{T}',\mathcal{B}')$ of $\mathcal{H}'$ with $\mu'\mhyphen\mathsf{width}(\mathcal{T}',\mathcal{B}')\leq a'$. We set $(\mathcal{T},\mathcal{B}) := \tau'(\mathcal{T}',\mathcal{B}')$ and observe that by definition of $\tau'$ and $\mu'$, the $\mu\mhyphen\mathsf{width}$ of $(\mathcal{T},\mathcal{B})$ is at most the $\mu'\mhyphen\mathsf{width}$ of $(\mathcal{T}',\mathcal{B}')$, which is at most $a'$, concluding the second direction and thus the proof.
\end{enumerate}
\end{proof}

\subsection{Relational Structures}\label{sec:rel_structs_and_CSP}
A \emph{signature} $\tau$ is a (finite) tuple of relation symbols $(R_i)_{i\in[\ell]}$ with arities $(a_i)_{i\in[\ell]}$. The arity of $\tau$, denoted by $\arity(\tau)$ is the maximum of the $a_i$. A \emph{relational structure} $\mathcal{A}$ of signature $\tau$ is a tuple $(V,R_1^\mathcal{A},\dots,R_\ell^\mathcal{A})$ where $V$ is a finite set of elements, called the \emph{universe} of $\mathcal{A}$, and $R_i^\mathcal{A}$ is a relation on $V$ of arity $a_i$ for each $i\in[\ell]$. We emphasize that $R_i^\mathcal{A}$ is not necessarily symmetric, and that tuples might contain repeated elements. We will mainly use the symbols $\mathcal{A}$ and $\mathcal{B}$ to denote relational structures. Further, we assume that a structure $\mathcal{A}$ is encoded in the standard way, i.e., the universe and the relations are encoded as lists. We denote by $|\mathcal{A}|$ the length of the encoding of $\mathcal{A}$.

Given two relational structures $\mathcal{A}$ and $\mathcal{B}$ over the same signature $\tau$ with universes $U$ and $V$, a \emph{homomorphism} from $\mathcal{A}$ to $\mathcal{B}$ is a mapping $\varphi:U \rightarrow V$ such that, for each $i\in [\arity(\tau)]$ and for each tuple $t\in U^{a_i}$ we have
\[t\in R_i^\mathcal{A} \Rightarrow \varphi(t)\in R_i^\mathcal{B} \,.\]
We write $\homs{\mathcal{A}}{\mathcal{B}}$ for the set of homomorphisms from $\mathcal{A}$ to $\mathcal{B}$.

The hypergraph $\mathcal{H}(\mathcal{A})$ of $\mathcal{A}$ has as vertices the universe $V$ of $\mathcal{A}$, and for each tuple $t=(v_1,\dots,v_a)$ of elements of $V$, we add an hyperedge $e_t=\{v_1,\dots,v_a\}$ if and only if $t$ is an element of a relation of $\mathcal{A}$. To avoid notational clutter, we will define the treewidth, the hypertreewidth, the fractional hypertreewidth and the submodular width of a structure as the respective width measure of its hypergraph. Similarly, a tree decomposition of a structure refers to a tree decomposition of its hypergraph.

\subsection{Parameterised and Fine-Grained Complexity Theory}
A \emph{parameterised counting problem} is a pair $(P,\kappa)$ of a counting problem $P:\{0,1\}^\ast \to \mathbb{N}$ and a computable function $\kappa:\{0,1\}^\ast \to \mathbb{N}$, called the \emph{parameterisation}. Consider for example the parameterised clique counting problem:

\begin{parameterizedproblem}
\problemname{$\#\textsc{Clique}$}
\probleminput{a pair of a graph $G$ and a positive integer $k$}
\problemoutput{the number of $k$-cliques in $G$}
\problemparameter{$k$, that is, $\kappa(G,k):= k$}
\end{parameterizedproblem}

An algorithm for a parameterised (counting) problem is called a \emph{fixed-parameter tractable} (FPT) algorithm if there is a computable function $f$ such that, on input $x$, its running time can be bounded by $f(\kappa(x))\cdot |x|^{O(1)}$. A parameterised (counting) problem is called \emph{fixed-parameter tractable} if it can be solved by an FPT algorithm.

A \emph{parameterised Turing-reduction} from $(P,\kappa)$ to $(P',\kappa')$ is an FPT algorithm for $(P,\kappa)$ with oracle access to $P'$, additionally satisfying that there is a computable function $g$ such that, on input $x$, the parameter $\kappa'(y)$ of any oracle query is bounded by $g(\kappa(x))$. We write $(P,\kappa) \fptred (P',\kappa')$ if a parameterised Turing-reduction exists.

We say that $(P,\kappa)$ is $\#\W{1}$\emph{-hard} if $\#\textsc{Clique}\fptred (P,\kappa)$. The class $\#\W{1}$ can be considered a parameterised counting equivalent of $\mathrm{NP}$, and we refer the interested reader to Chapter~14 in the standard textbook of Flum and Grohe~\cite{FlumG06}
for a comprehensive introduction. It is known that $\#\W{1}$-hard problems are not fixed-parameter tractable unless standard assumptions, such as ETH, fail:

\begin{definition}[The Exponential Time Hypothesis (ETH)~\cite{ImpagliazzoP01}]
The Exponential Time Hypothesis (ETH) asserts that $3$-$\textsc{SAT}$ cannot be solved in time $\mathsf{exp}(o(n))$, where $n$ is the number of variables.
\end{definition}

\begin{theorem}[Chen et al.\ \cite{Chenetal05,Chenetal06}]
Assume that ETH holds. Then there is no function $f$ such that $\#\textsc{Clique}$ can be solved in time $f(k)\cdot |G|^{o(k)}$. 
\end{theorem}
Note that the previous theorem rules out an FPT algorithm for $\#\textsc{Clique}$ (and thus all $\#\W{1}$-hard problems), unless ETH fails.

\subsubsection{Parameterised Counting Problems}\label{sec:prelims_probs}
The following parameterized problems are central to the present work. In what follows $\vec{C}$ denotes a class of directed graphs, and $\mathcal{C}$ denotes a class of hypergraphs.

\begin{parameterizedproblem}
\problemname{$\#\dirhomsprobd(\vec C)$}
\probleminput{a pair of digraphs $(\vec H, \vec G)$ with $\vec H \in \vec C$}
\problemoutput{\#\homs{\vec H}{\vec G}}
\problemparameter{$|\vec H|+d$ where $d$ is the maximum outdegree of $\vec G$}
\end{parameterizedproblem}

\begin{parameterizedproblem}
\problemname{$\#\dirsubsprobd(\vec C)$}
\probleminput{a pair of digraphs $(\vec H, \vec G)$ with $\vec H \in \vec C$}
\problemoutput{\#\subs{\vec H}{\vec G}}
\problemparameter{$|\vec H|+d$ where $d$ is the maximum outdegree of $\vec G$}
\end{parameterizedproblem}

\begin{parameterizedproblem}
\problemname{$\#\dirindsubsprobd(\vec C)$}
\probleminput{a pair of digraphs $(\vec H, \vec G)$ with $\vec H \in \vec C$}
\problemoutput{\#\indsubs{\vec H}{\vec G}}
\problemparameter{$|\vec H|+d$ where $d$ is the maximum outdegree of $\vec G$}
\end{parameterizedproblem}

\begin{parameterizedproblem}
\problemname{$\#\cpdirhomsprobd(\vec{C})$}
\probleminput{a digraph $\vec{H}\in \vec{C}$ and an $\vec{H}$-coloured digraph $(\vec{G},c)$}
\problemoutput{$\#\homs{\vec{H}}{(\vec{G},c)} $}
\problemparameter{$|\vec{H}|+d$ where $d$ is the maximum outdegree of $\vec G$}
\end{parameterizedproblem}

\begin{parameterizedproblem}
\problemname{$\#\textsc{CSP}(\mathcal{C})$}
\probleminput{a pair of relational structures $(\mathcal{A}, \mathcal{B})$ over the same signature with $\mathcal{H}(\mathcal{A}) \in \mathcal{C}$}
\problemoutput{$\#\homs{\mathcal{A}}{\mathcal{B}} $}
\problemparameter{$|\mathcal{A}|$}
\end{parameterizedproblem}

It was shown by Grohe and Marx~\cite{GroheM14} that the decision version of $\#\textsc{CSP}(\mathcal{C})$ can be solved in polynomial time if the fractional hypertreewidth of $\mathcal{C}$ is bounded. More precisely, they discovered an algorithm that solves the decision problem in time
\[(|\mathcal{A}| + |\mathcal{B}|)^{r+O(1)} \,,\]
assuming that a tree decomposition of $\mathcal{A}$ of $\rho^\ast$-width at most $r$ is given (see Theorem 3.5 and Lemma 4.9 in~\cite{GroheM14}); recall that the $\rho^\ast$-width of a tree decomposition is the maximum fractional edge cover number of a bag. In particular, they show that the partial solutions of each bag can be enumerated in time $(|\mathcal{A}| + |\mathcal{B}|)^{r+O(1)}$. Thus the dynamic programming algorithm that solves the decision version immediately extends to counting. Finally, since computing such an optimal tree decomposition can be done in time only depending on $\mathcal{A}$, we obtain the following overall running time for the counting problem:
\begin{theorem}\label{thm:csp_algo}
Let $\mathcal{A}$ and $\mathcal{B}$ be relational structures over the same signature and let $r$ be the fractional hypertreewidth of $\mathcal{A}$. There is a computable function $f$ such that we can compute $\#\homs{\mathcal{A}}{\mathcal{B}}$ in time
\[f(|\mathcal{A}|)\cdot |\mathcal{B}|^{r+O(1)} \,.\]
In particular, $\#\textsc{CSP}(\mathcal{C})$ is fixed-parameter tractable if $\mathcal{C}$ has bounded fractional hypertreewidth.
\end{theorem}

\section{The Directed Homomorphism Basis and Dedekind's Theorem}
In this section, we will revisit the interpolation technique for evaluating linear combinations of homomorphism counts due to Curticapean, Dell and Marx~\cite{CurticapeanDM17}, and we will extend their framework from undirected graphs to digraphs. We wish to point out that most of the results presented in this section are easy consequences and generalisations of methods known in the literature~\cite{Lovasz12,ChenM16,CurticapeanDM17,DellRW19}; and we only provide the details for reasons of self-containment. For the purpose of this section, we assume that finitely supported functions $\iota$ from digraphs to rationals are encoded as a list of elements $(\vec{F},\iota(\vec{F}))$ for all $\vec{F}$ with $\iota(\vec{F})\neq 0$. We write $|\iota|$ for the encoding length of $\iota$. 

To begin with, given a digraph $\vec{H}$, recall that $\#\subs{\vec{H}}{\star}$ and $\#\indsubs{\vec{H}}{\star}$ denote the functions that map a digraph $\vec{G}$ to $\#\subs{\vec{H}}{\vec{G}}$ and $\#\indsubs{\vec{H}}{\vec{G}}$, respectively.

For our reductions, we express both functions $\#\subs{\vec{H}}{\star}$ and $\#\indsubs{\vec{H}}{\star}$ as linear combinations of homomorphism counts from digraphs. We start with $\#\subs{\vec{H}}{\star}$ and point out that, similarly to the argument in~\cite{CurticapeanDM17}, the existence of the following transformation follows from M\"obius Inversion over the partition lattice as shown by \lovasz~(see Chapter 5.2.3 in~\cite{Lovasz12}).\footnote{In fact, Lemma~\ref{lem:subs_transformation} and Lemma~\ref{lem:indsubs_transformation} are special cases of more general transformations for counting answers to conjunctive queries with disequalities and negations~\cite{DellRW19}.}

\begin{lemma}\label{lem:subs_transformation}
Let $\vec{H}$ be a digraph. There exists a (unique and computable) function $\asub_{\vec{H}}$ from digraphs to rationals such that
\[\#\subs{\vec{H}}{\star} = \sum_{\vec{F}} \asub_{\vec{H}}(\vec{F}) \cdot \#\homs{\vec{F}}{\star}\,,\]
where the sum is over all (isomorphism classes of) digraphs $\vec{F}$.
Moreover, the function $\asub_{\vec{H}}$ has finite support, and satisfies $\asub_{\vec{H}}(\vec{F})\neq 0$ if and only if $\vec{F}$ is a quotient graph of $\vec{H}$.  
\end{lemma}

A similar transformation is known for $\#\indsubs{\vec{H}}{\star}$ which relies on \emph{arc supergraphs}.
\begin{definition}[Arc supergraphs]
Let $\vec{H}_1=(V_1,E_2)$ and $\vec{H}_2=(V_2,E_2)$ be digraphs without loops. We say that $\vec{H}_2$ is an \emph{arc supergraph} of $\vec{H}_1$ if $V_1=V_2$ and $E_1 \subseteq E_2$.
\end{definition}

In the first step, using the inclusion-exclusion principle, $\#\indsubs{\vec{H}}{\star}$ can be cast as a linear combination of subgraph counts (the proof is analogous to the undirected setting; see~\cite{CurticapeanDM17} and~\cite[Chapter 5.2.3]{Lovasz12}):
\begin{lemma}\label{lem:indsubs_transformation_intermediate}
Let $\vec{H}$ be a digraph. There exists a (unique and computable) function $\aindsub^\ast_{\vec{H}}$ from digraphs to rationals such that
\[\#\indsubs{\vec{H}}{\star} = \sum_{\vec{F}'} \aindsub^\ast_{\vec{H}}(\vec{F}') \cdot \#\subs{\vec{F}'}{\star}\,,\]
where the sum is over all (isomorphism classes of) digraphs $\vec{F}'$.
Moreover, the function $\aindsub^\ast_{\vec{H}}$ has finite support and satisfies $\aindsub^\ast_{\vec{H}}(\vec{F}')\neq 0$ if and only if $\vec{F}'$ is an arc supergraph of $\vec{H}$. 
\end{lemma}

In combination, the previous two lemmas allow us to cast $\#\indsubs{\vec{H}}{\star}$ as a linear combination of homomorphism counts.
\begin{lemma}\label{lem:indsubs_transformation}
Let $\vec{H}$ be a digraph. There exists a (unique and computable) function ${\aindsub}_{\vec{H}}$ from digraphs to rationals such that
\[\#\indsubs{\vec{H}}{\star} = \sum_{\vec{F}} {\aindsub}_{\vec{H}}(\vec{F}) \cdot \#\homs{\vec{F}}{\star}\,,\]
where the sum is over all (isomorphism classes of) digraphs $\vec{F}$.
Moreover, ${\aindsub}_{\vec{H}}$ satisfies the following conditions:
\begin{enumerate}
    \item\label{bul:b1} ${\aindsub}_{\vec{H}}$ has finite support.
    \item\label{bul:b2} If ${\aindsub}_{\vec{H}}(\vec{F})\neq 0$ then $\vec{F}$ is a quotient of an arc supergraph of $\vec{H}$.
    \item\label{bul:b3} If $\vec{F}$ is an arc supergraph of $\vec{H}$, then ${\aindsub}_{\vec{H}}(\vec{F})\neq 0$.
\end{enumerate}
\end{lemma}
\begin{proof}
We first apply the transformation from induced subgraphs to subgraphs as in Lemma~\ref{lem:indsubs_transformation_intermediate}, second, we apply the transformation from subgraphs to homomorphisms as in Lemma~\ref{lem:subs_transformation}. We obtain
\begin{equation}
    \#\indsubs{\vec{H}}{\star} = \sum_{\vec{F}'} \aindsub^\ast_{\vec{H}}(\vec{F}') \cdot \sum_{\vec{F}} \asub_{\vec{F}'}(\vec{F}) \cdot \#\homs{\vec{F}}{\star}\,.
\end{equation}
The coefficients $\aindsub_{\vec{H}}$ are then obtained by collecting for isomorphic terms, that is
\begin{equation}\label{eq:collect_coeffs_indsub}
    \aindsub_{\vec{H}}(\vec{F}) = \sum_{\vec{F}'} \aindsub^\ast_{\vec{H}}(\vec{F}') \cdot \asub_{\vec{F}'}(\vec{F})\,.
\end{equation}

Note that~\ref{bul:b1} and~\ref{bul:b2} follow immediately from the properties of $\asub$ and $\aindsub^\ast$. It remains to show~\ref{bul:b3}: To this end, let $\vec{F}$ be an \emph{arc supergraph} of $\vec{H}$. Then, for each arc supergraph $\vec{F}'$ of $\vec{H}$, the only quotient graph of $\vec{F}'$ that can be isomorphic to $\vec{F}$ is $\vec{F}'$ itself, since all other quotients have fewer vertices. By Lemma~\ref{lem:subs_transformation} we hence have that arc supergraphs $\vec{F}$ and $\vec{F}'$ of $\vec{H}$ satisfy that $\asub_{\vec{F}'}(\vec{F})\neq 0$ implies $\vec{F}\cong \vec{F}'$. Finally, Lemma~\ref{lem:indsubs_transformation_intermediate} asserts that $\aindsub^\ast_{\vec{H}}(\vec{F}')\neq 0$ if and only if $\vec{F}'$ is an arc supergraph of $\vec{H}$. Thus, using that $\vec{F}$ is an arc supergraph of $\vec{H}$, we have that~\eqref{eq:collect_coeffs_indsub} simplifies to
$\aindsub_{\vec{H}}(\vec{F})=\aindsub^\ast_{\vec{H}}(\vec{F})$, which is non-zero by Lemma~\ref{lem:indsubs_transformation_intermediate}. This concludes the proof.
\end{proof}

\noindent \textbf{A Remark on loops:} Readers familiar with~\cite{CurticapeanDM17} might notice that the quotient graphs in Lemma~\ref{lem:subs_transformation} and the arc supergraphs in Lemma~\ref{lem:indsubs_transformation} may have loops although loops are forbidden in~\cite{CurticapeanDM17}. The reason for this is the fact that we allow digraphs to contain loops, whereas undirected graphs in~\cite{CurticapeanDM17} are not allowed to contain loops. However, we emphasize that our hardness results will also apply to the restricted case of digraphs without loops. This will be made explicit in the respective sections.

\medskip

Next we show that the interpolation method in~\cite{CurticapeanDM17} for evaluating linear combinations of homomorphism counts transfers to the directed setting as well. The algorithm will be the same as in~\cite{CurticapeanDM17}; however, our correctness proof will be both more concise and more general at the same time by relying on a classical result of Dedekind.

\begin{lemma}\label{lem:complexity_monotonicity}
There exists a deterministic algorithm $\mathbb{A}$ with the following specification:
\begin{itemize}
    \item The input of $\mathbb{A}$ is a pair of a digraph $\vec{{G}}'$ with outdegree $d$ and a function $\iota$ from digraphs to rationals of finite support.
    \item $\mathbb{A}$ is equipped with oracle access to the function 
    \[ \vec{G} \mapsto \sum_{\vec{F}} \iota(\vec{F}) \cdot \#\homs{\vec{F}}{\vec{G}} \,,\]
    where the sum is over all (isomorphism classes of) digraphs.
    \item The output of $\mathbb{A}$ is the list with elements $(\vec{F},\#\homs{\vec{F}}{\vec{{G}}'})$ for each $\vec{F}$ with $\iota(\vec{F})\neq 0$.
\end{itemize}
Additionally, for some computable function $f:\N\to\N$ the running time of $\mathbb{A}$ is bounded by $f(|\iota|)\cdot |\vec{{G}}'|^{O(1)}$ and the outdegree of every oracle query $\vec{G}$ is at most $f(|\iota|)\cdot d$.

Similar algorithms exist for the restricted cases of digraphs without loops and DAGs.
\end{lemma}

The proof of the previous lemma requires some additional set-up. For what follows, we let $\vec{\mathcal{U}}^\circ$ be the class of all digraphs, and we let $\vec{\mathcal{U}}$ be the class of all digraphs without loops, and we let $\vec{\mathcal{D}}$ be the class of all DAGs. Next, consider the following operation on digraphs:

\begin{definition}[Tensor product]\label{def:tensor}
The tensor product $\vec{G}\otimes \vec{F}$ of two digraphs $\vec G$ and $\vec F$ is the digraph with $V(\vec{G}\otimes \vec{F}) = V(\vec{G})\times V(\vec{F})$ and with $((a,b),(c,d)) \in E(\vec{G}\otimes \vec{F})$ iff $(a,c) \in E(\vec{G})$ and $(b, d) \in E(\vec{F})$.
\end{definition}

A \emph{semigroup} is a pair of a set $\mathrm{G}$ and an associative operation $\ast$ on $\mathrm{G}$. Now observe that the tensor product of digraphs is clearly associative, that is, $\vec{G}\otimes (\vec{F} \otimes \vec{H})\cong (\vec{G}\otimes \vec{F}) \otimes \vec{H}$ given by the isomorphism $(u,(v,w))\mapsto ((u,v),w)$. Observe further that the tensor product of two digraphs without loops does not contain loops, and that the tensor product of two DAGs is again a DAG. Consequently, we obtain three semigroups:

\begin{observation}
$(\vec{\mathcal{U}}^\circ,\otimes)$, $(\vec{\mathcal{U}},\otimes)$, and $(\vec{\mathcal{D}},\otimes)$ are semigroups.
\end{observation}
\noindent The fact that the tensor product induces a semigroup will allow us to invoke constructive version of Dedekind's Theorem on the linear independence of characters
from Artin~\cite[Theorem 12]{Artin48}.\footnote{Note that Artin states Dedekind's Theorem for the case of $(\mathrm{G},\ast)$ being a group, rather than a semigroup. However, the proof only needs associativity of the operation $\ast$ and thus applies for the more general case of semigroups as well.}
Since Artin does not state the constructive version explicitly - it only follows from their proof - we provide a self-contained argument in Appendix~\ref{app:dedekind} for the reader's convenience.
\begin{theorem}\label{thm:dedekind}
Let $(\mathrm{G},\ast)$ be a semigroup. Let $(\varphi_i)_{i\in[k]}$ with $\varphi_i: \mathrm{G} \to \Q$ be pairwise distinct semigroup homomorphisms of  $(\mathrm{G},\ast)$ into $(\Q,\cdot)$, that is, $\varphi_i(g_1\ast g_2)= \varphi_i(g_1)\cdot \varphi_i(g_2)$ for all $i\in[k]$ and $g_1,g_2\in \mathrm{G}$. Let $\phi : G \to \Q$ be a function
\begin{equation}\label{eq:dedekind}
    \phi : g \mapsto \sum_{i=1}^k a_i \cdot \varphi_i(g)\,, 
\end{equation}
where the $a_i$ are rational numbers.
Suppose furthermore that the following functions are computable:
\begin{enumerate}
    \item The operation $\ast$.
    \item The mapping $(i,g) \mapsto \varphi_i(g)$.
    \item A mapping $i\mapsto g_i$ such that $\varphi_i(g_i)\neq 0$.
    \item A mapping $(i,j) \mapsto g_{i,j}$ such that $\varphi_i(g_{i,j})\neq \varphi_j(g_{i,j})$ whenever $i\neq j$.
\end{enumerate}
Then there is a constant $B$ only depending on the $\varphi_i$ (and not on the $a_i$), and an algorithm  $\hat{\mathbb{A}}$ such that the following conditions are satisfied:
\begin{itemize}
    \item $\hat{\mathbb{A}}$ is equipped with oracle access to $\phi$.
    \item $\hat{\mathbb{A}}$ computes $a_1,\ldots,a_k$.
    \item Each oracle query $\hat{g}$ only depends on the $\varphi_i$ (and not on the $a_i$).
    \item The running time of $\hat{\mathbb{A}}$ is bounded by $O\left( B \cdot \sum_{i=1}^k \log a_i \right)$
\end{itemize}
\end{theorem}

We aim to apply Theorem~\ref{thm:dedekind} to prove Lemma~\ref{lem:complexity_monotonicity}. However, this requires us to establish the following properties:

\begin{lemma}\label{lem:semigroup_hom}\label{lem:distinct_hom} The following conditions are satisfied:
\begin{enumerate}
    \item For every $\vec H \in \vec{\mathcal{U}}^\circ$ the function $\#\homs{\vec H}{\star}$ is a homomorphism from $(\vec{\mathcal{U}}^\circ, \otimes)$ into $(\Q,\cdot)$, that is, $\#\homs{\vec H}{\vec G \otimes \vec F} = \#\homs{\vec H}{\vec G} \cdot \#\homs{\vec H}{\vec F}$ for all $\vec G, \vec F \in \vec{\mathcal{U}}^\circ$.
    \item $\#\homs{\vec H_1}{\star} \ne \#\homs{\vec H_2}{\star}$ whenever $\vec H_1\in \vec{\mathcal{U}}^\circ$ and $\vec H_2\in \vec{\mathcal{U}}^\circ$ are non-isomorphic.
\end{enumerate}
The same holds true in the restricted cases of digraphs without loops and DAGs, that is, the same holds true if $\vec{\mathcal{U}}^\circ$ is substituted by $\vec{\mathcal{U}}$ or $\vec{\mathcal{D}}$.
\end{lemma}
\begin{proof}
The first claim is immediate since, by definition of $\otimes$, each $\varphi \in \homs{\vec H}{\vec G \otimes \vec F}$ decomposes (via projection) to $\varphi_1\in \homs{\vec H}{\vec G}$ and $\varphi_2\in \homs{\vec H}{\vec F}$, which induces a bijection.
 
For the second claim we follow a classical argument by \lovasz~(see Chapter~5.4 in~\cite{Lovasz12}).
For any two $\vec{F},\vec{G} \in \vec{\mathcal{U}}^\circ$ define:
\begin{align}
    \#\mathsf{Sur}(\vec{F}\to \vec{G}):=\{\varphi \in \#\homs{\vec{F}}{\vec{G}}~|~ \varphi \text{ is vertex-surjective}\}
\end{align}
For any $S \subseteq V(\vec{G})$ let $\vec{G}[S]$ be the subgraph of $\vec{G}$ induced by $S$. By inclusion and exclusion:
\begin{align}
    \#\mathsf{Sur}(\vec F \to \vec{G}) &= \sum_{S\subseteq V(\vec{G})} (-1)^{|V(\vec{G})\setminus S|} \cdot \#\homs{\vec F}{\vec{G}[S]}\label{eq:sur_ie}
\end{align}
Now assume for contradiction that $\#\homs{\vec{H}_1}{\star}=\#\homs{\vec{H}_2}{\star}$ for two non-isomorphic $\vec{H}_1, \vec{H}_2 \in \vec{\mathcal{U}}^\circ$. Then for $\vec G = H_1$~\eqref{eq:sur_ie} yields $\#\mathsf{Sur}(\vec{H_2} \to \vec{H_1})=\#\mathsf{Sur}(\vec{H_1} \to \vec{H_1})>0$, and for $\vec G=\vec H_2$ it yields $\#\mathsf{Sur}(\vec{H_1} \to \vec{H_2}) = \#\mathsf{Sur}(\vec{H_2} \to \vec{H_2})>0$. Hence there are surjective homomorphisms of $\vec{H}_1$ into $\vec{H}_2$ and of $\vec{H}_2$ into $\vec{H}_1$. But then $\vec{H}_1$ and $\vec{H}_2$ are isomorphic, which yields the desired contradiction.
\end{proof}

We are now able to proof Lemma~\ref{lem:complexity_monotonicity}.

\begin{proof}[Proof of Lemma~\ref{lem:complexity_monotonicity}]
Let $\vec{G}'$ and $\iota$ be the input.
We apply Dedekind's Theorem (Theorem~\ref{thm:dedekind}) to the semigroup $(\vec{\mathcal{U}}^\circ,\otimes)$ and mappings $\varphi_{\vec{F}}:= \#\homs{\vec{F}}{\star}$ for the digraphs $\vec{F}$ in the support $\mathsf{supp}(\iota) $ of $\iota$, that is, $k=|\mathsf{supp}(\iota)|$.
Concretely, assume that $\vec{F}_1,\dots,\vec{F}_k$ are elements of the support of $\iota$ (the $\vec{F}_i$ are pairwise non-isomorphic) and observe that we can use our oracle to compute the following function
\[ \vec{H} \mapsto \sum_{i=1}^k \iota(\vec{F}_i) \cdot \#\homs{\vec{F}_i}{\vec{G}\otimes\vec{H}} \,.\]
Using the properties of the tensor product, this rewrites to
\[ \vec{H} \mapsto  \sum_{i=1}^k  (\iota(\vec{F}_i) \cdot \#\homs{\vec{F}_i}{\vec{G}}) \cdot \#\homs{\vec{F}_i}{\vec{H}}\,.\]
Now we set $a_i:= \iota(\vec{F}_i) \cdot \#\homs{\vec{F}_i}{\vec{G}}$ and $\varphi_i :=\#\homs{\vec{F}_i}{\star}$. 
Note that Lemma~\ref{lem:semigroup_hom} makes sure that the $\varphi_i$ are indeed pairwise distinct semigroup homomorphisms. Next, clearly, all functions in 1.\ to 4.\ in Theorem~\ref{thm:dedekind} are computable in our setting. Hence we can use Algorithm $\hat{\mathbb{A}}$ from Theorem~\ref{thm:dedekind} to obtain $a_1,\dots,a_k$. The number of steps required by $\hat{\mathbb{A}}$ is bounded by $O(B \cdot \sum_{i=1}^k \log a_i)$, where $B$ does not depend on the $a_i$. Thus we can bound $B$ by a function in $|\iota|$. Furthermore, we can generously bound
$\log a_i \leq \log (\iota(\vec{F}_i) \cdot |\vec{G}'|^{|\iota|}) \leq f'(|\iota|)\cdot \log |\vec{G}'|$ for some computable function $f'$.

Next, when simulating an oracle query $\vec{H}$ posed by $\hat{\mathbb{A}}$, we have to use our own oracle to query $\vec{G}' \otimes \vec{H}$. Fortunately, Theorem~\ref{thm:dedekind} guarantees that $\vec{H}$ only depends on the $\varphi_i$, that is, on the $\vec{F}_i$, and thus only on $\iota$. Thus, constructing $\vec{G}' \otimes \vec{H}$ can be done in time $|\vec{G}'|^{O(1)} \cdot f''(|\iota|)$ for some computable function $f''$. Additionally, the outdegree of $\vec{G}' \otimes \vec{H}$ is, by definition of the tensor product, bounded by the outdegree of $\vec{G}'$, i.e., $d$, times the outdegree of $\vec{H}$, which only depends on $\iota$.
Finally, having obtained the $a_i$, we obtain the terms $\#\homs{\vec{F}_i}{\vec{G}'}$ by dividing by $\iota(\vec{F_i})$, which is well-defined, since all $\iota(\vec{F}_i)$ are non-zero.

Hence, there is a computable function $f$ such that the following two desired properties are true:
\begin{itemize}
    \item The total running time of our algorithm is bounded by $f(|\iota|)\cdot |\vec{G}'|^{O(1)}$, and
    \item the outdegree of every oracle query is bounded by $d\cdot f(|\iota|)$.
\end{itemize}
Noting that the same arguments apply in the restricted cases of digraphs without loops and DAGs, we can conclude the proof.
\end{proof}

\section{Counting Homomorphisms}
Recall the problem $\#\dirhomsprobd(\vec C)$ from Section~\ref{sec:prelims_probs}. Section~\ref{sec:dirhom_UB} shows that, when the reachability hypergraphs of the graphs in $\vec C$ have bounded fractional hypertreewidth, $\#\dirhomsprobd(\vec C)$ is fixed-parameter tractable. Section~\ref{sec:dirhom_red} instead gives parameterized reductions under what we call \emph{monotone reversible minors}, a new and restricted version of digraph minors. In the next sections we will leverage these results for $\#\dirsubsprobd(\vec C)$ and $\#\dirindsubsprobd(\vec C)$.

\subsection{Upper Bounds}\label{sec:dirhom_UB}
Let $\vec C$ be a class of of digraphs, and let $\reach{\vec C} = \{\reach{\vec H} : \vec H \in \vec C\}$. This section proves that, if $\reach{\vec C}$ has bounded fractional hypertreewidth, $\mathsf{fhtw}(\reach{\vec C})<\infty$, then $\#\dirhomsprobd(\vec C)$ is fixed-parameter tractable.
To this end we give a parameterized reduction from $\#\dirhomsprobd(\vec C)$ to $\#\CSP(\reach{\vec C})$ that preserves reachability hypergraphs, and then invoke Theorem~\ref{thm:csp_algo}.

\begin{definition}\label{def:struct_A_H}\label{def:struct_B_HG}
Let $\vec H$ be a digraph, let $S_1,\ldots,S_\ell$ be the sources of $\Hsim$, and fix any $s_i \in S_i$ for each $i\in[\ell]$. Fix any ordering $\prec$ of the vertices of $\vec{H}$, and, for each $i\in[\ell]$, let $t_i$ be the tuple formed by sorting $\reachv{s_i}$ according to $\prec$. Furthermore, let $\vec{G}$ be a digraph. 
The relational structures $\mathcal{A}[\vec{H}]$ and $\mathcal{B}[\vec{H},\vec{G}]$ are defined as follows:
\begin{itemize}
    \item  $\mathcal{A}[\vec{H}]=(V(\vec H), R_1^\mathcal{A}, \ldots, R_{\ell}^\mathcal{A})$, where $R_i^\mathcal{A}=\{t_i\}$ for all $i \in [\ell]$.
    \item $\mathcal{B}[\vec{H},\vec{G}] = (V(\vec{G}),R_1^\mathcal{B},\ldots,R_{\ell}^\mathcal{B})$, where \[R_i^\mathcal{B} = \{\phi(t_i) : \phi \in \homs{\vec{H}[\reachv{s_i}]}{\vec{G}} \}\] for all $i \in [\ell]$.
\end{itemize}
\end{definition}

\begin{lemma}\label{lem:hom_graph_struct}
$\homs{\vec{H}}{\vec{G}}=\homs{\mathcal{A}[\vec{H}]}{\mathcal{B}[\vec{H},\vec{G}]}$.
\end{lemma}
\begin{proof}
To avoid notational clutter, we set $e_i:=\reachv{s_i}$.
Let $\varphi\in\homs{\vec{H}}{\vec{G}}$ and fix any $i \in [\ell]$. Clearly, the restriction of $\varphi$ to $\vec H[e_i]$ is in $\homs{\vec{H}[e_i]}{\vec{G}}$. Therefore by Definition~\ref{def:struct_A_H} we have $\varphi(t_i) \in R_i^\mathcal{B}$. Hence $\varphi$ is a homomorphism from $\mathcal{A}[\vec{H}]$ to $\mathcal{B}[\vec{H},\vec{G}]$, that is, $\varphi \in \homs{\mathcal{A}[\vec{H}]}{\mathcal{B}[\vec{H},\vec{G}]}$.

Now let $\varphi \in \homs{\mathcal{A}[\vec{H}]}{\mathcal{B}[\vec{H},\vec{G}]}$ and fix any $i \in [\ell]$. By definition, $\varphi(t_i) \in R_i^\mathcal{B}$. By Definition~\ref{def:struct_B_HG} this implies $\varphi(t_i)=\phi(t_i)$ for some $\phi \in \homs{\vec H[e_i]}{\vec G}$, thus $(\varphi(u),\varphi(v)) \in E(\vec{G})$ for all $(u,v) \in E(\vec H[e_i])$. Since this holds for all $i \in [\ell]$ and since $E(\vec{H}) = \cup_{i\in[\ell]} E(\vec{H}[e_i])$, we have $(\varphi(u),\varphi(v)) \in E(\vec{G})$ for all $(u,v) \in E(\vec H)$. Thus $\varphi \in \homs{\vec{H}}{\vec{G}}$.
\end{proof}

Further, we note the following immediate consequence of the definition of $\mathcal{A}[\vec{H}]$.
\begin{observation}
$\mathcal{H}[\mathcal{A}[\vec{H}]]=\reach{\vec{H}}$.
\end{observation}

Next we show that $\mathcal{A}[\vec{H}]$ and $\mathcal{B}[\vec{H},\vec{G}]$ can be constructed efficiently.
\begin{lemma}\label{lem:efficient_reduction}
There exists a computable function $f$ such that, given any two DAGs $\vec{H}$ and $\vec{G}$, the relational structures $\mathcal{A}[\vec{H}]$ and $\mathcal{B}[\vec{H},\vec{G}]$ can be computed in time $f(k,d)\cdot n^{O(1)} $ where $k=|V(\vec{H})|$, $n=|V(\vec{G})|$, and $d$ is the maximum outdegree of $\vec{G}$. Moreover, $|\mathcal{B}[\vec{H},\vec{G}]|\leq f(k,d)\cdot O(|\vec{G}|)$.
\end{lemma}
\begin{proof}
It is straightforward that $\mathcal{A}[\vec{H}]$ can be constructed in time only depending on $k$. For what follows, we again set $e_i:=\reachv{s_i}$. Furthermore, set $a_i :=|e_i|$.
To construct $\mathcal{B}[\vec{H}, \vec G]$, and more precisely every relation $\mathcal{R}_i^{\mathcal B}$, we use the following standard technique. For every $v \in V(\vec{G})$ let $N_{k-1}(v)$ be the set of vertices reachable from $v$ by a directed path of length at most $k-1$. Note that for any $\phi \in \homs{\vec H[e_i]}{\vec G}$ we have that $\vec H[e_i]$ is connected and contains at most $k$ vertices. Furthermore, each vertex in $\vec{H}[e_i]$ is reachable by a directed path from $s_i$. Thus $\phi(t_i)$ contains only vertices of $N_{k-1}(v)$ where $v=\phi(s_i)$.

Therefore to list all $\phi(t_i)$ with $\phi \in \homs{\vec H[e_i]}{\vec G}$ we take every $v \in V(\vec G)$ in turn, we compute $N_{k-1}(v)$, and for every $a_i$-tuple $z_i \in (N_{k-1}(v))^{a_i-1}$ whose first element is $v$, we add $z_i$ if and only if the map $\phi$ defined by $\phi(t_i)=z_i$ preserves all edges of $\vec H[e_i]$, which holds if and only if $\phi \in \homs{\vec H[e_i]}{\vec G}$. 
Finally, recalling that $d$ is the outdegree of $\vec{G}$, we have $|N_{k-1}(v)| \le \sum_{j=0}^{k-1} d^{j}$, which only depends on $k$ and $d$. 
It is thus immediate to see that we can compute all $\mathcal{R}_i^{\mathcal B}$ in time $f(k,d) \cdot n^{O(1)}$, and that the overall size of $\mathcal{B}[\vec{H},\vec{G}]$ is bounded by $f(k,d)\cdot O(|\vec{G}|)$, for some computable function $f$.
\end{proof}
\begin{theorem}\label{thm:homs_algo}
For some computable function $f$ there is an algorithm that, given any pair of digraphs $(\vec H,\vec G)$, computes $\#\homs{\vec{H}}{\vec{G}}$ in time $f(|\vec{H}|,d)\cdot |\vec{G}|^{r+O(1)}$,
where $d$ is the maximum outdegree of $\vec G$ and $r=\mathsf{fhtw}(\reach{\vec{H}})$. Therefore $\#\dirhomsprobd(\vec{C}) \in \ccFPT$ if $\mathsf{fhtw}(\reach{\vec{C}}) < \infty$.
\end{theorem}
\begin{proof}
Given an instance $(\vec{H},\vec{G})$, let $n=|V(\vec{G})|$ and $k=|V(\vec{H})|$. We compute $\mathcal{A}[\vec{H}]$ and $\mathcal{B}[\vec{H},\vec{G}]$ as in Lemma~\ref{lem:efficient_reduction}. In particular, we obtain $|\mathcal{B}[\vec{H},\vec{G}]| \leq g(k,d)\cdot O(|\vec{G}|)$ for some computable function $g$. Finally, by Lemma~\ref{lem:hom_graph_struct}, we have $\homs{\vec{H}}{\vec{G}}=\homs{\mathcal{A}[\vec{H}]}{\mathcal{B}[\vec{H},\vec{G}]}$, the latter of which can be computed, by Theorem~\ref{thm:csp_algo}, in time
\[f'(|\mathcal{A}[\vec{H}]|) \cdot |\mathcal{B}[\vec{H},\vec{G}]|^{r+O(1)} \leq f'(|\mathcal{A}[\vec{H}]|) \cdot g(k,d)^{r+O(1)} \cdot |\vec{G}|^{r+O(1)}\]
for some computable function $f'$. Since $|\mathcal{A}[\vec{H}]|$ depends only on $k$ and not on $n$, the proof is concluded.
\end{proof}

\subsection{Coloured Homomorphisms and Reductions via MR Minors}\label{sec:dirhom_red}
We start by introducing \emph{monotone reversible (MR) minors} of a digraph $\vec H$. 

\begin{definition}[Monotone Reversible Minors]
Let $\vec{H}$ be a digraph and consider the following operations:
\begin{itemize}
    \item[]\emph{Sink deletion:}
    delete all vertices in $T$ for some sink $T$ of $\Hsim$. The resulting graph is denoted by $\vec{H}\setminus T$.
    \item[]\emph{Contraction:} identify $u$ and $v$ for some $uv \in E(\vec H)$. We emphasise that a contraction does not yield a loop. The resulting graph is denoted by $\vec{H}/(u,v)$.
    \item[]\emph{Loop deletion:} delete a loop $(u,u)$. The resulting graph is denoted by $\vec{H}\setminus(u,u)$.
\end{itemize}
A \emph{monotone reversible minor} (``MR minor'') of $\vec{H}$ is a digraph that can be obtained from $\vec{H}$ by a sequence of sink deletions, contractions and loop deletions.
\end{definition}

\begin{observation}\label{obs:wminor_obs}
Let $\vec{H}$ be a digraph.
\begin{enumerate}
    \item[(M1)] If $\vec{H}$ does not have loops, then no MR minor of $\vec{H}$ does.
    \item[(M2)] $\Hsim$ is an MR minor of $\vec{H}$.
\end{enumerate}
\end{observation}

We show that the parameterized complexity of $\#\dirhomsprobd$ is monotone (i.e., nonincreasing) under taking MR~minors. The following three lemmas establish the fact separately for sink deletions, for contractions, and for loop deletion. For technical reasons, we will prove this property for the colour-prescribed variant; we will be able to remove the colours in our hardness reductions later.
\begin{lemma}\label{lem:technical_minors_deletions}
    There exists an algorithm $\mathbb{A}_1$ that satisfies the following constraints:
    \begin{enumerate}
        \item $\mathbb{A}_1$ expects as input a digraph $\vec{H}$, a sink $T$ of $\Hsim$, and a surjectively $\vec{H}\setminus T$-coloured digraph $(\vec{G}',c')$ of outdegree $d'$.
        \item The running time of $\mathbb{A}_1$ is bounded by $\mathsf{poly}(|\vec{H}|,|\vec{G}'|)$.
        \item $\mathbb{A}_1$ outputs a surjectively $\vec{H}$-coloured digraph $(\vec{G},c)$ of size at most $O(|H|\cdot |\vec{G}'|)$ such that the outdegree of $\vec{G}$ is bounded by $d'+|\vec{H}|$, and 
        \[ \#\homs{\vec{H}\setminus T}{(\vec{G}',c')} = \#\homs{\vec{H}}{(\vec{G},c)} \,.\]
    \end{enumerate}
\end{lemma}
\begin{proof}
Let $T=\{t_1,\dots,t_k\}$. For every $i\in[k]$ let $V_i$ be the set of vertices $v\in V(\vec{H})\setminus T$ such that $(v,t_i)\in E(\vec{H})$. Note that, since $T$ is a sink in $\Hsim$, there are no arcs from $T$ to $V(\vec{H})\setminus T$ in $\vec{H}$.

We construct $(\vec{G},c)$ from $(\vec{G}',c')$ as follows: First, we add to $\vec{G}'$ the set $T$ and all arcs in $T^2\cap E(\vec{H})$. Second, for every $i\in[k]$, we add to $\vec{G}$ all arcs from $c'^{-1}(V_i)$ to $t_i$ --- note that $c'^{-1}(V_i)$ is the set of vertices of $\vec{G}'$ that are coloured by $c'$ with a vertex in $V_i$. Finally, we let $c$ agree with $c'$ on all vertices of $\vec{G}'$, and we set $c(t_i)=t_i$ for all $i\in[k]$.
Clearly, this construction can be done in time $\mathsf{poly}(|\vec{H}|,|\vec{G}'|)$, and the resulting $(\vec{G},c)$ is of size at most $O(|H|\cdot |\vec{G}'|)$. Furthermore, the outdegree of $\vec{G}$ is bounded by $d'+|\vec{H}|$. By construction, and the fact that $c'$ is a vertex-surjective homomorphism from $\vec{G}'$ to $\vec{H}\setminus T$, it is also immediate that $c$ is a vertex-surjective homomorphism from $\vec{G}$ to $\vec{H}$. Finally, consider the function $b$ that maps a colour-prescribed homomorphism $\varphi\in \homs{\vec{H}\setminus T}{(\vec{G}',c')}$ to the function $\psi$ that agrees with $\varphi$ on $V(\vec{H}\setminus T)$ and that maps $t_i$ to $t_i$ for each $i\in[k]$. By construction, it is easy to see that $b$ must be a bijection from $\homs{\vec{H}\setminus T}{(\vec{G}',c')}$ to $\homs{\vec{H}}{(\vec{G},c)}$. This concludes the proof. 
\end{proof}

\begin{lemma}\label{lem:technical_minors_contractions}
    There exists an algorithm $\mathbb{A}_2$ that satisfies the following constraints:
    \begin{enumerate}
        \item $\mathbb{A}_2$ expects as input a digraph $\vec{H}$, an arc $(u,v)$ of $\vec{H}$, and a surjectively $\vec{H}/(u,v)$-coloured digraph $(\vec{G}',c')$ of outdegree $d'$.
        \item The running time of $\mathbb{A}_2$ is bounded by $\mathsf{poly}(|\vec{H}|,|\vec{G}'|)$.
        \item $\mathbb{A}_2$ outputs a surjectively $\vec{H}$-coloured digraph $(\vec{G},c)$ of size at most $O(|H|\cdot |\vec{G}'|)$ such that the outdegree of $\vec{G}$ is bounded by $2d'+1$, and 
        \[ \#\homs{\vec{H}/(u,v)}{(\vec{G}',c')} = \#\homs{\vec{H}}{(\vec{G},c)} \,.\]
    \end{enumerate}
\end{lemma}
\begin{proof}
We write $uv$ for the vertex in $\vec{H}/(u,v)$ that corresponds to the contraction of $u$ and $v$. Recall that $\vec{G}'$ is surjectively $\vec{H}/(u,v)$-coloured by $c'$, and set $V_{uv}:=c'^{-1}(uv)$. Let us now provide the construction of $(\vec{G},c)$:
\begin{enumerate}
    \item[(a)] We start with $\vec{G}'$ and delete the set $V_{uv}$ (including all incident arcs).
    \item[(b)] We add two copies of $V_{uv}$; one is denoted by $V_u$, and the other one is denoted by $V_v$. For each vertex $w\in V_{uv}$, we denote the copy of $w$ in $V_u$ by $w_u$, and the copy of $w$ in $V_v$ by $w_v$.
    \item[(c)] For each $w\in V_{uv}$, we add an arc $(w_u,w_v)$.
    \item[(d)] For each $x\in V(\vec{H})\setminus\{u,v\}$ we proceed as follows
    \begin{itemize}
        \item If $(x,u)\in E(\vec{H})$, then, for every $y\in c'^{-1}(x)$ and $w\in V_{uv}$, we add an arc from $y$ to $w_u\in V_u$ if and only if $(y,w) \in E(\vec{G}')$.
        \item If $(u,x)\in E(\vec{H})$, then, for every $y\in c'^{-1}(x)$ and $w\in V_{uv}$, we add an arc from $w_u\in V_u$ to $y$ if and only if $(w,y) \in E(\vec{G}')$.
        \item If $(x,v)\in E(\vec{H})$, then, for every $y\in c'^{-1}(x)$ and $w\in V_{uv}$,we add an arc from $y$ to $w_v\in V_v$ if and only if $(y,w) \in E(\vec{G}')$.
        \item If $(v,x)\in E(\vec{H})$, then, for every $y\in c'^{-1}(x)$ and $w\in V_{uv}$, we add an arc from $w_v\in V_v$ to $y$ if and only if $(w,y) \in E(\vec{G}')$.
    \end{itemize}
    \item[(e)] Finally, $c$ agrees with $c'$ on $V(\vec{H})\setminus\{u,v\}$, and $c$ maps every vertex in $V_u$ to $u$ and it maps every vertex in $V_v$ to $v$.
\end{enumerate}
It is immediate that $c$ is a surjective $\vec{H}$-colouring of $\vec{G}$. Furthermore, it is clear that the running time is bounded by $\mathsf{poly}(|\vec{H}|,|\vec{G}'|)$, and that the size of $(\vec{G},c)$ is at most $O(|H|\cdot |\vec{G}'|)$. Let us consider the outdegree: For $x\in V(\vec{H})\setminus\{u,v\} = V(\vec{H}/(u,v)) \setminus\{uv\}$, the outdegree of every vertex $v\in c^{-1}(x)=c'^{-1}(x)$ in $\vec{G}$ is bounded by twice the outdegree of $v$ in $\vec{G}'$ (see (d) above). Furthermore, the outdegree of every vertex $w_u\in c^{-1}(u)$ in $\vec{G}$
is bounded by the outdegree of $w$  in $\vec{G}'$ plus $1$ (see (d) above and note that the arcs added in (c) above can increase it by $1$). Finally, the outdegree of every vertex $w_v\in c^{-1}(v)$ in $\vec{G}$ is bounded by the outdegree of $w$ in $\vec{G}'$ (see (d) above).
Consequently, we can bound the outdegree of $\vec{G}$ by $2d'+1$.

Finally, note that (c) above makes sure that any homomorphism $\psi$ in $\homs{\vec{H}}{(\vec{G},c)}$ must map $u$ and $v$ to the same copy of a vertex $w\in V_{uv}$. Furthermore, (d) makes sure that the mapping $\varphi$ that agrees with $\psi$ on $V(\vec{H})\setminus\{u,v\}$ and that maps $uv$ to $w$ (where $w_u$ and $w_v$ are the images of $u$ and $v$ under $\psi$), is a homomorphism in $\homs{\vec{H}/(u,v)}{(\vec{G}',c')}$. On the other hand every homomorphism $\varphi \in \homs{\vec{H}/(u,v)}{(\vec{G}',c')}$ corresponds to the homomorphism $\psi \in \homs{\vec{H}}{(\vec{G},c)}$ that agrees with $\varphi$ on $V(\vec{H})\setminus\{u,v\}$, and that maps $u$ and $v$ to $w_u$ and $w_v$, where $w=\varphi(uv)$. Concretely, we obtain the desired bijection between $\homs{\vec{H}/(u,v)}{(\vec{G}',c')}$ and $\homs{\vec{H}}{(\vec{G},c)}$, concluding the proof.
\end{proof}

\begin{lemma}\label{lem:technical_minors_self-loops}
    There exists an algorithm $\mathbb{A}_3$ that satisfies the following constraints:
    \begin{enumerate}
        \item $\mathbb{A}_3$ expects as input a digraph $\vec{H}$, a loop $(u,u)$ of $\vec{H}$, and a surjectively $\vec{H}\setminus(u,u)$-coloured digraph $(\vec{G}',c')$ of outdegree $d'$.
        \item The running time of $\mathbb{A}_3$ is bounded by $\mathsf{poly}(|\vec{H}|,|\vec{G}'|)$.
        \item $\mathbb{A}_3$ outputs a surjectively $\vec{H}$-coloured digraph $(\vec{G},c)$ of size at most $O(|\vec{G}'|)$ such that the outdegree of $\vec{G}$ is bounded by $d'+1$, and 
        \[ \#\homs{\vec{H}\setminus(u,u)}{(\vec{G}',c')} = \#\homs{\vec{H}}{(\vec{G},c)} \,.\]
    \end{enumerate}
\end{lemma}
\begin{proof}
This is a very easy case: Obtain $\vec{G}$ from $\vec{G}'$ by adding a loop to each vertex of $\vec{G}'$ coloured by $c'$ with $u$. Furthermore, we set $c:=c'$. Clearly, $c$ is a surjective $\vec{H}$-colouring of $\vec{G}$, and the outdegree of $\vec{G}$ can increase by at most $1$. Furthermore, the construction immediately yields that  $\homs{\vec{H}\setminus(u,u)}{(\vec{G}',c')} = \homs{\vec{H}}{(\vec{G},c)}$, that is, a mapping $\varphi:V(\vec{H})(=V(\vec{H}')) \to V(\vec{G})(=V(\vec{G}'))$ is a homomorphism from $\vec{H}\setminus(u,u)$ to $(\vec{G}',c')$ if and only if it is a homomorphism from $\vec{H}$ to $(\vec{G},c)$
\end{proof}

In combination, the three lemmas above yield the following:
\begin{lemma}\label{lem:minor_reduction}
Let $\vec{C}$ be a recursively enumerable class of digraphs and let $\vec{C}'$ be a class of MR minors of graphs in $\vec{C}$. Then
\[\#\cpdirhomsprobd(\vec{C}') \fptred \#\cpdirhomsprobd(\vec{C}) \,.\]
\end{lemma}
\begin{proof}
Let $\vec{H}'$ and $(\vec{G}',c')$ be an input instance of $\#\cpdirhomsprobd(\vec{C}')$. We start by searching a graph $\vec{H}\in\vec{C}$ such that $\vec{H}'$ is an MR minor of $\vec{H}$. Note that this can be done in time only depending on $\vec{H}'$. Since $\vec{H}'$ can be obtained from $\vec{H}$ by a sequence of $\ell$ sink-deletions, contractions, and loop deletions, we can use algorithms $\mathbb{A}_1$, $\mathbb{A}_2$ and $\mathbb{A}_3$ from the previous three lemmas for a total of $\ell$ times. Note that the crucial property of the $\mathbb{A}_i$ is that the oracle queries always have size bounded by $f(|\vec{H}|)\cdot O(|\vec{G}'|)$. Hence, even after $\ell$ applications of the constructions, the total size will still be bounded by $f(|\vec{H}|)\cdot O(|\vec{G}'|)$. Since furthermore each individual application takes only polynomial time, we obtain, as desired a parameterised Turing-reduction.
\end{proof}

The final part of this subsection is the following lemma for removing the colours. Note that its proof is a simple application of the inclusion-exclusion principle and transfers verbatim from e.g.\ \cite[Lemma 2.49]{Roth19} (see also~\cite[Lemma 1.34]{Curticapean15}). We emphasise that the reduction only requires oracle queries for subgraphs of the input host-graph which cannot increase the outdegree; thus the reduction applies to our setting.
\begin{lemma}\label{lem:colour_removal}
Let $\vec{C}$ be a class of digraphs. We have
\[\#\cpdirhomsprobd(\vec{C}) \fptred \#\dirhomsprobd(\vec{C}) \,.\]
\end{lemma}

\subsection{Lower Bounds}
We start by extending the notion of adaptive width from hypergraphs to digraphs.

\begin{definition}[Adaptive width of digraphs]
The \emph{adaptive width} of a digraph $\vec{H}$, denoted by $\aw(\vec{H})$, is defined as the adaptive width of its contour. That is $\aw(\vec{H}):=\aw(\reduct{\vec{H}})$. 
\end{definition}

We proceed by proving intractability of $\#\cpdirhomsprobd(\vec{C})$ for classes $\vec{C}$ of canonical DAGs of unbounded adaptive width.

\begin{lemma}\label{lem:coloured_CSP_reduction}
Let $\vec{C}$ be a recursively enumerable class of canonical DAGs. If the adaptive width of $\vec{C}$ is unbounded then $\#\cpdirhomsprobd(\vec{C})$ is not fixed-parameter tractable, unless ETH fails.
\end{lemma}

The proof of Lemma~\ref{lem:coloured_CSP_reduction} uses a careful reduction from a version of the parameterised constraint satisfaction problem and is encapsulated in the following subsection.

\subsubsection{Reduction from a Constraint Satisfaction problem}
The starting point of our reduction is the following \emph{decision} problem and the corresponding hardness result proven by Chen et al.~\cite{DBLP:conf/ijcai/ChenGLP20} (building upon Marx~\cite{Marx13}); below, $\mathcal{S}$ denotes a class of relational structures.

\begin{parameterizedproblem}
\problemname{$\homsprob(\mathcal{S})$}
\probleminput{a pair of relational structures $(\mathcal A, \mathcal B)$ with $\mathcal{A}\in \mathcal{S}$}
\problemoutput{\textsc{true} iff there exists a homomorphism from ${\mathcal{A}}$ to ${\mathcal{B}}$}
\problemparameter{$|\mathcal A|$}
\end{parameterizedproblem}

\noindent Let us emphasise the subtle difference between the input restrictions of $\homsprob(\mathcal{S})$ and $\#\textsc{CSP}(\mathcal{C})$: In the former, we enforce that $\mathcal{A}$ is contained in $\mathcal{S}$. In the latter, we only restrict the hypergraphs of the relational structure $\mathcal{A}$ by requiring that the hypergraph $\mathcal{H}(\mathcal{A})$ is contained in $\mathcal{C}$.

For what follows, we call a relational structure $\mathcal{A}$ \emph{minimal under homomorphic equivalence} if there is no homomorphism from $\mathcal{A}$ to a proper substructure of $\mathcal{A}$.

\begin{theorem}[\cite{DBLP:conf/ijcai/ChenGLP20}]\label{thm:homS_hard}
Let $\mathcal{S}$ be a recursively enumerable class of relational structures that are minimal under homomorphic equivalence, and assume that ETH holds. If $\mathcal{S}$ has unbounded adaptive width then $\homsprob(\mathcal{S})$ is not fixed-parameter tractable.
\end{theorem}
To obtain hardness of $\#\dirhomsprobd(\vec{C})$ from Theorem~\ref{thm:homS_hard}, we show a chain of parameterized Turing reductions using two intermediate problems; one of them is $\#\cpdirhomsprobd(\vec{C})$, and for the second one we extend the notion of colour-prescribed homomorphisms to hypergraphs:
Let $\mathcal{H}$ be a hypergraph. An $\mathcal{H}$\emph{-colouring} of a hypergraph $\mathcal{G}$ is a homomorphism $c\in\homs{\mathcal{G}}{\mathcal{H}}$. An $\mathcal{H}$\emph{-coloured hypergraph} is a pair $(\mathcal G, c)$ where $\mathcal G$ is a hypergraph and $c$ is an $\mathcal{H}$-colouring of $\mathcal G$.
Given a $\mathcal{H}$-coloured hypergraph $(\mathcal{G},c)$, a map $\psi : V(\mathcal{H}) \to V(\mathcal{G})$ is $c$-\emph{colour-prescribed} if $c(\psi(v))=v$ for every $v\in V(\mathcal{H})$. The set of all $c$-colour-prescribed homomorphisms from $\mathcal{H}$ to $\mathcal{G}$ is denoted by $\homs{\mathcal{H}}{(\mathcal{G},c)}$.

We can now introduce the problem $\#\cphomsproba(\mathcal{C})$; here $\mathcal{C}$ is a class of hypergraphs:
\begin{parameterizedproblem}
\problemname{$\#\cphomsproba(\mathcal{C})$}
\probleminput{a hypergraph $\mathcal H \in \mathcal C$ and an $\mathcal H$-coloured hypergraph $(\mathcal G,c)$}
\problemoutput{$\#\homs{\mathcal{H}}{(\mathcal{G},c)}$}
\problemparameter{$|\mathcal H|+a(\mathcal{G})$ (recall that $a(\mathcal{G})$ is arity of $\mathcal{G}$)}
\end{parameterizedproblem}

\noindent In the rest of this subsection we prove:
\begin{align}
    \homsprob(\mathcal{S}) 
    \fptred
    \#\cphomsproba(\mathcal{C})
    \fptred
    \#\cpdirhomsprobd(\vec{C}) \label{eq:fptred_chain}
\end{align}
where $\mathcal{S}$ and $\mathcal{C}$ are carefully constructed from $\vec{C}$ so to preserve the (un)boundedness of adaptive width. The next two paragraphs prove the reductions of~\eqref{eq:fptred_chain} in order.

\paragraph*{$\homsprob(\mathcal{S}) \fptred \#\cphomsproba(\mathcal{C})$}
We reduce $\homsprob(\mathcal{S})$ to $\#\cphomsproba(\mathcal{C})$ for a certain class $\mathcal{S}=\mathcal{S}(\mathcal{C})$ described below. To this end we convert every hypergraph into a structure; this structure is the same of Definition~\ref{def:struct_A_H} -- only for hypergraphs. Without loss of generality, in what follows we assume $V(\mathcal{H}) = \{1,\ldots,k\}$ where $k=|V(\mathcal H)|$.
\begin{definition}
Let $\mathcal{H}$ be any hypergraph. The structure $\mathcal A[\mathcal H]$ has universe $V(\mathcal H)$ and, for every $e \in E(\mathcal H)$, contains a relation $R^\mathcal{A}_e$ whose only tuple is the set $e$ sorted in nondecreasing order.
\end{definition}

Since we use individual relation symbols for each $e\in E(\mathcal{H})$, the following is straightforward to prove:
\begin{claim}\label{clm:hom_minimal}
Let $\mathcal{H}$ be any hypergraph. Then $\mathcal A[\mathcal H]$ is minimal under homomorphic equivalence, and the hypergraph of $\mathcal A[\mathcal H]$ is $\mathcal H$.
\end{claim}

For any class $\mathcal C$ of hypergraphs let $\mathcal S[\mathcal C] = \{A[\mathcal H] ~|~ \mathcal H \in \mathcal C\}$. We prove:
\begin{lemma}\label{lem:cphoms_hard}
For every recursively enumerable class $\mathcal{C}$ of hypergraphs without isolated vertices, $\homsprob(\mathcal{S}[\mathcal{C}]) \fptred \#\cphomsproba(\mathcal{C})$.
\end{lemma}
\begin{proof}
Let $(\mathcal A, \mathcal B)$ be an instance of $\homsprob(\mathcal{S}[\mathcal{C}])$, and let $U(\mathcal{A})$ and $U(\mathcal{B})$, respectively, be the universes of $\mathcal{A}$ and $\mathcal{B}$. If $\mathcal{A}$ and $\mathcal{B}$ have different signatures then clearly the solution is NO. Otherwise, since $\mathcal{A}\in \mathcal{S}[\mathcal{C}]$ and $\mathcal{C}$ is recursively enumerable,  in time $f(|\mathcal{A}|)$ for some computable $f$ we find $\mathcal{H}\in \mathcal{C}$ such that $\mathcal{A}=\mathcal{A}[\mathcal{H}]$.
Then we construct the hypergraph $\mathcal{G}$ defined by:
\begin{itemize}
    \item $V(\mathcal{G})=U(\mathcal{B})\times V(\mathcal{H})$
    \item $E(\mathcal G) = \{(x_1,i_1),\dots,(x_\ell,i_\ell)\} : e = \{i_1,\dots,i_\ell\}\in E(\mathcal H), i_1<\ldots<i_{\ell}, (x_1,\dots,x_\ell)\in R^\mathcal{B}_e $ 
\end{itemize}
Let $c:V(\mathcal{G})\to V(\mathcal{H})$ be defined by $c((x,i)) = i$ for all $(x,i) \in V(\mathcal{G})$. Note that $c$ is an $\mathcal{H}$-colouring of $\mathcal{G}$.
Consider the instance $(\mathcal H, \mathcal (G,c))$ of $\#\cphomsproba(\mathcal{C})$. We claim:
\begin{align} \homs{\mathcal{A}}{\mathcal{B}}\neq \emptyset \;\Leftrightarrow\; \homs{\mathcal{H}}{(\mathcal{G},c)}\neq \emptyset
\end{align}
To prove that $\homs{\mathcal{A}}{\mathcal{B}}\neq \emptyset \,\Rightarrow\, \homs{\mathcal{H}}{(\mathcal{G},c)}\neq \emptyset$, suppose $\varphi\in \homs{\mathcal{A}}{\mathcal{B}}$. Define $\psi:V(\mathcal H) \to V(\mathcal G)$ by $\psi(i)=(\varphi(i),i)$ for all $i\in V(\mathcal{H})$. We claim that $\psi \in \homs{\mathcal{H}}{(\mathcal{G},c)}$. Let $e=\{i_1,\dots,i_\ell\} \in E(\mathcal{H})$ with $i_1<\dots<i_\ell$. By definition of $\mathcal{A}[\mathcal H]$ we have $(i_1,\dots,i_\ell)\in R^\mathcal{A}_e$; and since $\varphi$ is a homomorphism, $(\varphi(i_1),\dots,\varphi(i_\ell))\in R^\mathcal{B}_e$. By definition of $\mathcal{G}$ this implies that $\{(\varphi(i_1),i_1),\dots,(\varphi(i_\ell),i_\ell)\} \in E(\mathcal{G})$. Hence, $\psi \in \homs{\mathcal{H}}{\mathcal{G}}$. To see that $\psi$ is $c$-colour-prescribed, note that $c(\psi(i))=c((\varphi(i),i))=i$ for all $i \in V(\mathcal H)$.

To prove that $\homs{\mathcal{H}}{(\mathcal{G},c)}\neq \emptyset \,\Rightarrow\, \homs{\mathcal{A}}{\mathcal{B}}\neq \emptyset$, suppose $\psi\in \homs{\mathcal{H}}{(\mathcal{G},c)}$. Since $\psi$ is $c$-colour-prescribed, for each $i\in V(\mathcal{H})$ we have $\psi(i)=(x_i,i)$ for some $x_i \in U(\mathcal B)$.
Define $\varphi : U(\mathcal A) \to U(\mathcal B)$ by letting $\varphi(i)=x_i$ for all $i \in U(\mathcal A)$. We claim that $\varphi \in \homs{\mathcal{A}}{\mathcal{B}}$.
Let indeed $(i_1,\dots,i_\ell)\in R^\mathcal{A}_e$. By definition of $\mathcal{A}$ this implies $e=\{i_1,\dots,i_\ell\}\in E(\mathcal{H})$ and $i_1<\dots<i_\ell$. Since $\psi$ is $c$-color-prescribed, then $\psi(i_j)=(x_j,i_j)$ for all $j = 1,\ldots,\ell$; and since $\psi \in \homs{\mathcal H}{\mathcal G}$, then $ \{(x_1,i_1),\dots,(x_\ell,i_{\ell})\} \in E(\mathcal G)$. By definition of $\mathcal{G}$ this implies $(x_1,\dots,x_\ell)\in R^\mathcal{B}_e$, and since $\varphi(i_1,\dots,i_\ell)=(x_1,\dots,x_\ell)$, then $\varphi(i_1,\dots,i_\ell)\in R^\mathcal{B}_e$.

Finally, note that $|\mathcal H|=f(|\mathcal A|)$ and $a(\mathcal G) \le a(\mathcal H) \le |\mathcal H|$. Therefore $|\mathcal H|+a(\mathcal G)$ is a function of $|\mathcal A|$, hence the reduction preserves the parameter.
\end{proof}

\paragraph*{$\#\cphomsproba(\mathcal{C})\fptred\#\cpdirhomsprobd(\vec{C})$}

Recall that the contour of a digraph $\vec{H}$, denoted by $\reduct{\vec{H}}$, is obtained by deleting from the reachability hypergraph $\reach{\vec{H}}$ the vertices corresponding to sources in $\Hsim$ (see Definition~\ref{def:contour}). 
For our reduction we need to prove that the adaptive width of a contour is at least that of the original hypergraph.

\begin{lemma}\label{lem:reduct_aw}
Let $\vec{H}$ be a canonical DAG and let $\mathcal{R}$ be its reachability hypergraph. Then $\aw({\vec{H}})\geq \aw(\mathcal{R})$.
\end{lemma}
\begin{proof}
Recall that, by definition, $\aw({\vec{H}}) = \aw(\reduct{\vec{H}})$. Set $\mathcal{F}:=\reduct{\vec{H}}$ and let $s_1,\ldots,s_{\ell}$ be the sources of $\vec H$ and $e_1,\ldots,e_{\ell}$ the corresponding hyperedges in $\mathcal{R}$. Let $\mathcal F_0 = \mathcal R$, and for $i=1,\ldots,\ell$ let $\mathcal F_i$ be the hypergraph obtained from $\mathcal F_{i-1}$ by first deleting $s_i$ and then removing a copy of $e_i$ if more than one exist. Note that $\mathcal F_{\ell}=\mathcal F$. 
We prove that $\aw(\mathcal F_i) \ge \aw(\mathcal F_{i-1})$ for all $i=1,\ldots,\ell$, which implies $\aw(\mathcal F_{\ell}) \ge \aw(\mathcal F_0)$, that is, $\aw(\mathcal F) \ge \aw(\mathcal R)$. Since removing a copy of a multi-edge leaves adaptive with unchanged, we can assume $\mathcal F_i$ is obtained from $\mathcal F_{i-1}$ by just deleting $s_i$.

We rephrase and prove the claim $\aw(\mathcal F_i) \ge \aw(\mathcal F_{i-1})$ as follows. Let $\mathcal{H}$ be a hypergraph, let $e \in E(\mathcal{H})$, and let $\mathcal{H}'$ be the hypergraph obtained from $\mathcal{H}$ by adding a new vertex $v'$ and replacing $e$ with $e \cup \{v'\}$. We claim $\aw(\mathcal{H}) \ge \aw(\mathcal{H}')$.
To this end, we show that for every fractional independent set $\mu'$ of $\mathcal{H}'$ there exists a tree decomposition $(\mathcal{T}',\mathcal{B}')$ of $\mathcal{H}'$ such that $\fwidth{\mu'}(\mathcal{T}',\mathcal{B}') \leq \aw(\mathcal{H})$.
Let then $\mu'$ be a fractional independent set of $\mathcal{H}'$, and let $\mu$ be the restriction of $\mu'$ to $V(\mathcal{H})=V(\mathcal{H}')\setminus\{v'\}$. Note that $\mu$ is a fractional independent set of $\mathcal{H}$; thus, by definition of adaptive width, there is a tree decomposition $(\mathcal{T},\mathcal{B})$ of $\mathcal{H}$ such that $\fwidth{\mu}(\mathcal{T},\mathcal{B}) \le \aw(\mathcal H)$. By definition of tree decomposition, there exists $B_e\in \mathcal{B}$ such that $e \subseteq B_e$. Let $(\mathcal{T}',\mathcal{B}')$ be obtained from $(\mathcal{T},\mathcal{B})$ by appending appending the bag $B_{e'}=e'$ to $B_e$.

We claim that $(\mathcal{T}',\mathcal{B}')$ is a tree decomposition of $\mathcal{H'}$ and that $\fwidth{\mu'}(\mathcal{T}',\mathcal{B}') \leq \aw(\mathcal{H})$.
To see that $(\mathcal{T}',\mathcal{B}')$ is a tree decomposition of $\mathcal{H'}$ note that:
\begin{itemize}
    \item every edge of $\mathcal{H}'$ is a subset of some bag of $\mathcal B'$. Indeed, $e' \subseteq B_{e'}$, while every other edge is in $E(\mathcal H)$, and is thus a subset of some bag of $\mathcal{B}$ since $(\mathcal T, \mathcal{B})$ is a tree decomposition of $\mathcal H$
    \item for every $u \in V(\mathcal H')$ the subgraph $\mathcal{T}'_u$ of $\mathcal{T}'$ induced by the bags containing $u$ is connected. Indeed, if $u=v'$ then $\mathcal{T}'_u$ has only one vertex. Else, $\mathcal{T}'_u$ equals $\mathcal{T}_u$ if $u \notin e$ and $\mathcal{T}_u$ with an appended vertex otherwise. 
\end{itemize}
To prove $\fwidth{\mu'}(\mathcal{T}',\mathcal{B}')\leq \aw(\mathcal{H})$ consider any $B'\in \mathcal{B}'$. If $B'=e'$ then $\mu'(B') \le 1$ since $\mu'$ is a fractional independent set of $\mathcal H'$; and since $\aw \ge 1$, then $\mu'(B') \le \aw(\mathcal H')$. If $B'\ne e'$ then $B' \in \mathcal{B}$ and $v' \notin B'$, so $\mu'(B')=\mu(B')$. But $\mu(B') \le \fwidth{\mu}(\mathcal T,\mathcal B) \le \aw(\mathcal H)$, thus $\mu'(B') \le \aw(\mathcal H)$.
\end{proof}

We are ready for our parameterised Turing-reduction.
\begin{lemma}\label{lem:cphom_to_cpdirhom}
Let $\vec{C}$ be a recursively enumerable class of canonical DAGs and define $\hat{\mathcal{C}}= \{\Gamma(\vec{H}) \mid \vec{H}\in \vec{C}\}$.
Then $\#\cphomsproba(\hat{\mathcal{C}}) \fptred \#\cpdirhomsprobd(\vec{C}) $.
\end{lemma}
\begin{proof}
Let $(\mathcal H, (\mathcal G, c))$ be the input to $\#\cphomsproba(\hat{\mathcal{C}})$. As $\mathcal H \in \{\Gamma(\vec{H}) | \vec{H}\in \vec{C}\}$ and $\vec C$ is recursively enumerable, for some computable $f$ in time $f(|\mathcal H|)$ we find $\vec H \in \vec C$ such that $\mathcal H = \Gamma(\vec H)$. Let $s_1,\ldots,s_{\ell}$ be the sources of $\vec H$, and for every $s \in \{s_1,\ldots,s_{\ell}\}$ let $e_s$ be the set of non-source vertices reachable from $s$. Note that $e_i \in E(\mathcal H)$. 

We construct an $\vec H$-coloured DAG $(\vec G, c')$ such that $d(\vec G) \le a(\mathcal G)$ and 
\[|\homs{\Gamma}{(\mathcal{G},c)}| = |\homs{\vec{H}}{(\vec{G},c')}| \,.\] First, since $c \in \homs{\mathcal G}{\mathcal H}$, then $c(e) \in E(\mathcal H)$ for every $e \in E(\mathcal G)$. Let $S_{c(e)}$ contain every $s \in \{s_1,\ldots,s_{\ell}\}$ such that $\{s\} \cup c(e)$ is the reachable set of $s$ in $\vec H$. This implies that $(s,v) \in E(\vec H)$ for all $s \in S_{c(e)}$ and $v \in c(e)$.
Then define:
\begin{align}
    V(\vec G) &= V(\mathcal G) \cup \{x_{e,s} : e \in E(\mathcal G), s \in S_{c(e)}\}\\
    E(\vec G) &= \{(x_{e,s},v) : e \in E(\mathcal G), s \in S_{c(e)}, v \in e\}
\end{align}
and:
\begin{align}
    c'(v) &= c(v) \; :\; v \in V(\mathcal G)\\
    c'(x_{e,s}) &= s \;:\; e \in E(\mathcal G), s \in S_{c(e)}
\end{align}
Observe that $d(\vec G) \le a(\mathcal G)$ and that $\vec G$ and $c'$ can be constructed in FPT time.

Let us show that $c' \in \homs{\vec G}{\vec H}$. Let $(u,v) \in E(\vec G)$. By construction $u=x_{e,s}$ for $e \in E(\mathcal G)$, $s \in S_{c(e)}$, $v \in e$. By definition $c'(x_{e,s})=s$ and $c'(v)=c(v)$. Thus $(c'(x_{e,s}),c'(v))=(s,c(v))$, and as observed above $(s,c(v)) \in E(\vec H)$. Thus $(c(u),c(v)) \in E(\vec H)$ as desired.

Now we give a bijection between $\homs{\mathcal H}{(\mathcal{G},c)} $ and $\homs{\vec{H}}{(\vec{G},c')}$, proving that $|\homs{\mathcal H}{(\mathcal{G},c)}| = |\homs{\vec{H}}{(\vec{G},c')}|$.
First, let $\varphi \in \homs{\mathcal H}{(\mathcal{G},c)}$, and define the following extension $\psi : V(\vec H) \to V(\vec G)$ of $\varphi$:
\begin{align}
    \psi(v) &= \varphi(v) \;:\; v \in V(\mathcal H) \\
    \psi(s) &= x_{\varphi(e_s),s} \;:\; s \in \{s_1,\ldots,s_{\ell}\}
\end{align}
We claim that $\psi \in \homs{\vec{H}}{(\vec{G},c')}$. First, let us show that $\psi \in \homs{\vec H}{\vec G}$. Let $(s,v) \in E(\vec H)$. Then $(\psi(s),\psi(v)) = (x_{\varphi(e_s),s},\varphi(v))$. Since $e_s \in E(\mathcal H)$ and $\varphi \in \homs{\mathcal H}{\mathcal G}$, then $\varphi(e_s) = e$ for some $e \in E(\mathcal G)$; let us then write $(\psi(s),\psi(v))=(x_{e,s},\varphi(v))$.
Now, since $(s,v) \in E(\vec H)$, then $v \in e_s$, which implies $\varphi(v) \in \varphi(e_s) = e$. By construction of $\vec G$ this implies $(x_{e,s},\varphi(v)) \in E(\vec G)$. Therefore $\psi \in \homs{\vec H}{\vec G}$. 
To show that $\psi \in \homs{\vec H}{(\vec{G},c')}$, note that for all $v \in V(\mathcal H)$ we have $c'(\psi(v))=c(\varphi(v))=v$ by definition of $c'$ and $c$; while for all $s \in \{s_1,\ldots,s_{\ell}\}$ we have $c'(\psi(s))=c'(x_{\varphi(e_s),s}) = s$. Hence $\psi \in \homs{\vec H}{(\vec{G},c')}$.

Next, let $\psi \in \homs{\vec H}{(\vec{G},c')}$, and define $\varphi$ as the restriction of $\psi$ to $V(\mathcal H)$. Note that this is the inverse of the extension defined above, and as a consequence it is also $c$-prescibed. 
Therefore to establish our bijection we need only to prove that $\varphi \in \homs{\mathcal H}{(\mathcal{G},c)}$.
Consider any edge $\epsilon \in E(\mathcal H)$. By construction, $\epsilon=e_s$ for some $s \in \{s_1,\ldots,s_{\ell}\}$. Since $\psi \in \homs{\vec H}{\vec G}$, then $(\psi(s),\psi(v))\in E(\vec G)$ for all $v \in e_s$. By construction of $\vec G$ and since $\psi$ is $c'$-prescribed, this implies that $\psi(s) = x_{e,s}$ and $\psi(e_s) \subseteq e$ for some $e \in E(\mathcal G)$. The injectivity of $c$ and thus $c'$ on $e$, however, implies that $|e| \le |\psi(e_s)|$. Therefore, $\psi(e_s) = e$ and thus $\varphi(e_s) = \psi(e_s) \in E(\mathcal G)$.
\end{proof}

We are now able to prove Lemma~\ref{lem:coloured_CSP_reduction}; recall that we are required to show that $\#\cpdirhomsprobd(\vec{C})$ is (fixed-parameter) intractable whenever $\vec{C}$ is a class of canonical DAGs of unbounded adaptive width.
\begin{proof}[Proof of Lemma~\ref{lem:coloured_CSP_reduction}]
Let $\hat{\mathcal{C}} := \{\reduct{\vec{H}}~|~\vec{H}\in \vec{C} \}$ and note that $\hat{\mathcal{C}}$ does not have isolated vertices (a vertex contained in a hyperedge of cardinality $1$ is not isolated). Recall further that $\mathcal{S}[\hat{\mathcal{C}}] = \{\mathcal{A}[\mathcal{H}] ~|~\mathcal{H} \in \hat{\mathcal{C}} \}$. By Lemma~\ref{lem:cphoms_hard} and Lemma~\ref{lem:cphom_to_cpdirhom}, we have
\begin{align}
    \homsprob(\mathcal{S}[\hat{\mathcal{C}}]) 
    \fptred
    \#\cphomsproba(\hat{\mathcal{C}})
    \fptred
    \#\cpdirhomsprobd(\vec{C}) \label{eq:fptred_chain_in_proof}
\end{align}
\end{proof}
Since $\vec{C}$ has unbounded adaptive width, we obtain by Lemma~\ref{lem:reduct_aw} that $\hat{\mathcal{C}}$ and thus $\mathcal{S}[\hat{\mathcal{C}}]$ has unbounded adaptive width as well. By Claim~\ref{clm:hom_minimal}, each structure in $\mathcal{S}[\hat{\mathcal{C}}]$ is furthermore minimal under homomorphic equivalence. By Theorem~\ref{thm:homS_hard}, the problem $\homsprob(\mathcal{S}[\hat{\mathcal{C}}]) $ is thus not fixed-parameter tractable, unless ETH fails. The proof can thus be concluded by applying the chain of reductions~\eqref{eq:fptred_chain_in_proof}.

\subsection{The Intractability Result}
In combination with our MR minor operations and the removal of colours, we obtain the following result, which will be the basis of every intractability result in the forthcoming sections.
\begin{lemma}\label{lem:main_hardness_homs}
Let $\vec{C}$ be a recursively enumerable class of digraphs and let $\vec{C}'$ be a class of canonical DAGs that are MR minors of digraphs in $\vec{C}$. If $\vec{C}'$ has unbounded adaptive width, then $\#\dirhomsprobd(\vec{C})$ is not fixed-parameter tractable, unless ETH fails.
\end{lemma}
\begin{proof}
By Lemma~\ref{lem:coloured_CSP_reduction}, the problem $\#\cpdirhomsprobd(\vec{C}')$ is fixed-parameter intractable, unless ETH fails. By Lemma~\ref{lem:minor_reduction} and Lemma~\ref{lem:colour_removal}, we have that
\[ \#\cpdirhomsprobd(\vec{C}') \fptred \#\cpdirhomsprobd(\vec{C}) \fptred \#\dirhomsprobd(\vec{C})\,,\]
which concludes the proof.
\end{proof}

\section{Counting Subgraphs}
The following invariant will turn out to precisely capture the complexity for counting directed subgraphs in bounded outdegree graphs.
\begin{definition}[Fractional Cover Number]
A \emph{fractional cover} of a DAG $\vec{H}$ with sources $S$ is a function $\psi:S \to [0,\infty]$ such that for every vertex $v\in V(\vec{H})\setminus S$ we have
\[\sum_{\substack{s \in S \\ v\in R(s)}} \psi(s) \geq 1 \,.\]
The \emph{weight} of a fractional cover is $\sum_{s\in S} \psi(s)$, and we also require that the weight is at least $1$.\footnote{The technical reason for this is the corner case of $V(\vec{H})\setminus S$ being empty.}
The \emph{fractional cover number} of $\vec{H}$, denoted by $\rho^\ast(\vec{H})$, is the minimum weight of a fractional cover of $\vec{H}$.
The fractional cover number of a digraph $\vec{H}$ (not necessarily acyclic) is  $\rho^\ast(\vec{H})=\rho^\ast(\Hsim)$.
\end{definition}
In the previous definition we overloaded the symbol $\rho^\ast$, which also denotes the fractional edge cover number of hypergraphs. The motivation for reusing the symbol stems from the following observation:
\begin{observation}\label{obs:fractional_equivalence}
Let $\vec{H}$ be a digraph. We have $\rho^\ast(\vec{H}) = \rho^\ast(\reduct{\Hsim})$, that is, the fractional cover number of $\vec{H}$ is equal to the fractional edge cover number of the contour of the DAG $\Hsim$.
\end{observation}

\subsection{Upper Bounds}\label{sec:sub_UB}
To avoid notational clutter, given a fractional cover $\psi$ of $\Hsim$, we will call $\psi$ also a fractional cover of $\vec{H}$.

\begin{lemma}\label{lem:main_technical_subs}
Let $\vec{H}$ be a digraph and let $\vec{H}'$ be a quotient of $\vec{H}$. Then $\rho^\ast(\vec{H}')\leq \rho^\ast(\vec{H})$.
\end{lemma}
\begin{proof}
It suffices to show the following: Let $u$ and $v$ be (not necessarily adjacent) vertices of $\vec{H}$ and let $\vec{H}'$ be the graph obtained by identifying $u$ and $v$, that is, $\vec{H}'=\vec{H}/\sigma$, where $\sigma$ is the partition of $V(\vec{H})$ containing a block $\{u,v\}$ and singleton blocks $\{w\}$ for each $w\notin\{u,v\}$. Then $\rho^\ast(\vec{H}')\leq \rho^\ast(\vec{H})$.

To prove the previous claim, let $S_1,\dots,S_k$ be the sources of $\Hsim$ and let $W_1,\dots,W_\ell$ be the non-sources of $\Hsim$, that is $V(\Hsim)=\{S_1,\dots,S_k,W_1,\dots,W_\ell \}$. Thus $V(\vec{H})=\mathcal{S}\dot\cup\mathcal{W}$ where $\mathcal{S}:= S_1\cup\dots\cup S_k$ and $\mathcal{W}:= W_1\cup\dots\cup W_k$. In particular, the $S_i$ and the $W_i$ are the strongly connected components of $\vec{H}$.

Now let $\psi$ be a fractional cover of ${\Hsim}$ of weight $\rho^\ast(\vec{H})$.

We will proceed with a case-distinction on whether $u$ and $v$ (or both) are contained in the set of sources $\mathcal{S}:= S_1\cup\dots\cup S_k$. In each case, we will define a fractional cover $\psi'$ of $\SCCquot{H'}$ of total weight at most $\rho^\ast(\vec{H})$. Note that we will always assume that $u$ and $v$ are in distinct strongly connected components, since otherwise, $\Hsim = \SCCquot{H'}$ and thus the fractional cover number does not change. 

To avoid confusion, we denote $\mathbf{S}=\{S_1,\dots,S_k\}$ as the set of sources of $\Hsim$ and we denote $\mathbf{S}'$ as the set of sources of $\SCCquot{H'}$. Similarly, we denote $\mathbf{W}=\{W_1,\dots,W_k\}$ as the set of non-sources of $\Hsim$ and we denote $\mathbf{W}'$ as the set of non-sources of $\SCCquot{H'}$. Furthermore, given $S\in\mathbf{S}$ we write $R(S)$ for the set of vertices reachable from $S$ in $\Hsim$, and, given $S'\in \mathbf{S}'$ we write $R'(S')$ for the set of vertices reachable from $S'$ in $\SCCquot{H'}$.
\begin{itemize}
    \item[($u\in \mathcal{S}, v\in \mathcal{S}$)] Assume w.l.o.g.\ that $u\in S_1$ and $v\in S_2$. By identifying $u$ and $v$ we merge the strongly connected components $S_1$ and $S_2$; let us denote the resulting component by $\hat{S}$. Thus, $\SCCquot{H'}$ is obtained from $\Hsim$ by identifying $S_1$ and $S_2$, and calling the resulting vertex $\hat{S}$. Hence $\mathbf{S}'=\{\hat{S},S_3,\dots,S_k\}$ and $\mathbf{W}=\mathbf{W}'$. Now, for $S'\in\mathbf{S}'$, we define
    \[\psi'(S'):= \begin{cases} \psi(S') & S' \in \{S_3,\dots,S_k\} \\ \psi(S_1)+\psi(S_2) & S'=\hat{S}
    \end{cases}
    \]
    Clearly, we have
    \[ \sum_{S'\in \mathbf{S}'} \psi'(S') = \psi'(\hat{S}) +\sum_{i=3}^k \psi'(S_i) = \psi(S_1)+\psi(S_2) +\sum_{i=3}^k \psi'(S_i)= \sum_{S \in\mathbf{S}} \psi(S) = \rho^\ast(\vec{H})\,.\]
    It remains to show that $\psi'$ is a fractional cover of $\SCCquot{H'}$. To this end, let $W$ be a non-source of $\SCCquot{H'}$. Note that $W\in R(S_i)$ if and only if $W\in R'(S_i)$ for all $i\in\{3,\dots,k\}$. If $W$ is not reachable from $\hat{S}$, then $W$ is not reachable from either of $S_1$ or $S_2$ in $\Hsim$. Thus
    \[\sum_{\substack{S'\in \mathbf{S}'\\W \in R'(S')}} \psi'(S') = \sum_{\substack{S\in \mathbf{S}\\W \in R(S)}} \psi'(S) =\sum_{\substack{S\in \mathbf{S}\\W \in R(S)}} \psi(S) \geq 1\,.
    \]
    Otherwise, we have
    \[\sum_{\substack{S'\in \mathbf{S}'\\W \in R'(S')}} \psi'(S') = \psi(S_1) + \psi(S_2) + \sum_{\substack{S\in \mathbf{S}\\W \in R(S)\setminus\{S_1,S_2\}}} \psi'(S)  \geq \sum_{\substack{S\in \mathbf{S}\\W \in R(S)}} \psi(S) \geq 1\,,
    \]
    \item[($u\in \mathcal{S}, v\notin \mathcal{S}$)] Assume w.l.o.g.\ that $u \in S_1$ and let $V$ be the strongly connected component containing $v$. We have to consider the following two subcases:
    \begin{itemize}
        \item[\textbf{Case 1:}] The only source in $\Hsim$ from which $V$ can be reached is $S_1$. This means that $\SCCquot{H'}$ is obtained from $\Hsim$ by contracting to $S_1$ each $W$ that is reachable from $S_1$ and from which $V$ can be reached; this includes of course $S_1$ and $V$. The resulting vertex $\hat{S}$ is a source of $\SCCquot{H'}$ and, clearly, any fractional cover $\psi$ of $\Hsim$ becomes a fractional cover $\psi'$ of $\SCCquot{H'}$ by setting $\psi'(\hat{S})=\psi(S_1)$ and $\psi'(S_i)=\psi(S_i)$ for $i\geq 2$. Thus $\rho^\ast(\vec{H}')\leq \rho^\ast(\vec{H})$.
        
        \item[\textbf{Case 2:}] There is are sources $S_2,\dots,S_t$ in $\Hsim$ different from $S_1$ from which $V$ can be reached; specifically, assume that $S_2,\dots,S_t$ are all sources with this property. The identification of $u$ and $v$ in $\vec{H}$ then corresponds to the following operation in $\SCCquot{H'}$:
        \begin{itemize}
            \item[(a)] Similarly as in Case 1, we obtain a new vertex, called $\hat{W}$, by contracting all vertices in $\SCCquot{H'}$ that are reachable from $S_1$ and from which $V$ can be reached. 
            \item[(b)] In contrast to Case 1, $\hat{W}$ is not a source, since there will be arcs from $S_2,\dots,S_t$ to $\hat{W}$. However, $\hat{W}$ is not reachable from any source $S_i$ with $t<i\leq k$. 
        \end{itemize}
        Observe that the following holds for all vertices $W\in\mathbf{W}$ of $\Hsim$ that were not contracted to $\hat{W}$ in (a): If $W$ is reachable from $S_1$ in $\Hsim$, then $W$ is reachable from all sources $S_2,\dots,S_t$ in $\SCCquot{H'}$. 
        
        Noting that $\mathbf{S}'=\mathbf{S}\setminus\{S_1\}=\{S_2,\dots,S_t,\dots,S_k\}$, we define $\psi'$ as follows:
        \[\psi'(S'):= \begin{cases} \psi(S') & S'\in \{S_{t+1},\dots,S_k\} \\ \psi(S')+\psi(S_1)/(t-1) & S'\in \{S_2,\dots,S_t\}
    \end{cases}
    \]
    First we observe that, clearly,
    \[\sum_{S'\in\mathbf{S'}} \psi'(S') = \sum_{S\in\mathbf{S}} \psi(S) =\rho^\ast(\vec{H}) \,.\]
    Hence it remains to show that $\psi'$ is a fractional cover. To this end, let $W\in\mathbf{W}'$. Assume first that $W=\hat{W}$ and note that $\hat{W}$ is reachable from $S_2,\dots,S_t$ in $\SCCquot{H'}$: If it would be reachable from $S_i$ with $i>t$, then $V$ would have been reachable in $\Hsim$ from $S_i$, contradicting our choice of the $S_2,\dots,S_t$. Note further that $V$ is reachable from (precisely) $S_1,\dots,S_t$ in $\vec{H}$. Since $\psi$ is a fractional cover of $\Hsim$, we have
    \[\sum_{i=1}^t \psi(S_i) \geq 1 \,.\]
    Hence, we have that
    \[ \sum_{\substack{S'\in \mathbf{S}'\\\hat{W} \in R'(S')}} \psi'(S') = \sum_{i=2}^t \psi'(S_i) = \sum_{i=1}^t \psi(S_i) \geq 1\,.\]
    Next, assume that $W\neq \hat{W}$. For each $i\geq 2$, if $W$ is reachable from $S_i$ in $\SCCquot{H'}$ then $W$ is also reachable from $S_i$ in $\Hsim$. Thus, if $W$ is not reachable from $S_1$ in $\Hsim$, then 
    \[ \sum_{\substack{S'\in \mathbf{S}'\\W \in R'(S')}} \psi'(S') =  \sum_{\substack{S\in \mathbf{S}\\W \in R(S)}} \psi'(S)\geq 1\,.\]
    Finally, if $W$ is reachable from $S_1$ in $\Hsim$, we recall that $W$ must be reachable from all $S_2,\dots,S_t$ in $\SCCquot{H'}$. Since our definition of $\psi'$ adds to the value of those sources $\psi(S_1)/(t-1)$, we have
    \[ \sum_{\substack{S'\in \mathbf{S}'\\W \in R'(S')}} \psi'(S') = (t-1) \cdot \psi(S_1)/(t-1) + \sum_{\substack{S\in \mathbf{S}\setminus\{S_1\}\\W \in R(S)}} \psi(S) = \sum_{\substack{S\in \mathbf{S}\\W \in R(S)}} \psi(S)  \geq 1\,.\]
    This concludes Case 2.
    \end{itemize}
    \item[($u\notin \mathcal{S}, v\in \mathcal{S}$)] Symmetric to the previous case.
    \item[($u\notin \mathcal{S}, v\notin \mathcal{S}$)] Let $U$ and $V$ be the strongly connected components of $\vec{H}$ containing $u$ and $v$, respectively. Then $\SCCquot{H'}$ is obtained from $\Hsim$ by contracting $U$ and $V$, and all vertices between them, to a single vertex; here, a vertex $X$ is ``between'' $U$ and $V$ if there is a directed path from $U$ to $V$ (or vice versa) that contains $X$. Let us call the resulting vertex $\hat{W}$. Note that $\mathbf{S}=\mathbf{S}'$, that is, $\Hsim$ and $\SCCquot{H'}$ have the same sources. Note further that, for every $i\in\{1,\dots,k\}$ and non-source $W\neq\hat{W}$, if $W$ is reachable from $S_i$ in $\Hsim$, then it is also reachable from $S_i$ in $\SCCquot{H'}$. Furthermore, $\hat{W}$ is reachable from a source $S_i$ in $\SCCquot{H'}$ if one of the vertices that was contracted to $\hat{W}$ was reachable from $S_i$ in $\Hsim$. Thus, every fractional cover of $\Hsim$ must also be a fractional cover of $\SCCquot{H'}$, concluding this case.
\end{itemize}
With all cases resolved, the proof is complete.
\end{proof}

\begin{lemma}\label{lem:subs_quotients_small_fhtw}
Let $\vec{H}$ be a digraph. Then $\mathsf{fhtw}(\reach{\vec{H}})\leq \rho^\ast(\vec{H})$.
\end{lemma}
\begin{proof}
Let $S_1,\dots,S_k$ be the sources of $\Hsim$, and let $R_1,\dots,R_k$ be the hyperedges of $\reach{\vec{H}}$, that is, for each $i\in[k]$ the hyperedge $R_i$ includes all vertices in $\vec{H}$ that can be reached from $S_i$.

Consider the following tree decomposition $(\mathcal{T},\mathcal{B})$ of $\reach{\vec{H}}$:
\begin{itemize}
    \item We add one \emph{center} bag $B:=V(\vec{H})\setminus\left(\bigcup_{i=1}^k S_i \right)$, that is, $B$ contains all vertices of $\vec{H}$ not included in strongly connected components that become sources is $\Hsim$.
    \item For each $i\in[k]$ we add a bag $B_i:=R_i$, which is made adjacent to the center bag $B$.
\end{itemize}
Note first that this yields indeed a tree decomposition. Clearly, each vertex in $V(\reach{\vec{H}})=V(\vec{H})$ is contained in a bag, including isolated vertices of $\vec{H}$ (since those will become isolated sources in $\Hsim$). By definition, each hyperedge $R_i$ is fully contained in at least one bag. Finally, for every vertex $v\in V(\reach{\vec{H}})$, the subtree $\mathcal{T}_v=\mathcal{T}[B\in \mathcal{B}~|~v\in B]$ is connected: If $v$ is contained in $S_i$ for some $i\in[k]$, then $\mathcal{T}_v$ only consists of $B_i=R_i$. Otherwise, $v$ is contained in the center bag $B$. Hence $\mathcal{T}_v$ cannot be disconnected.

Now let us prove that each bag has a fractional edge cover of weight at most $\rho^\ast(\vec{H})$: For the center bag $B$, any fractional cover of $\vec{H}$ yields, by definition, a fractional edge cover of $B$ with the same weight (recall that the sources of $\Hsim$ correspond to the hyperedges of $\reach{\vec{H}}$). Hence, the fractional edge cover number of $B$ is bounded by $\rho^\ast(\vec{H})$. Finally, each bag $B_i=R_i$ can clearly be covered by one hyperedge. Thus the fractional edge cover number of $B_i$ is $1$.  
\end{proof}

\begin{theorem}\label{thm:main_algo_subs}
There is a computable function $f$ such that the following is true. Let $\vec{H}$ and $\vec{G}$ be digraphs, let $d$ be the maximum outdegree of $\vec{G}$, and let $r$ be the fractional cover number of $\vec{H}$. We can compute $\#\subs{\vec{H}}{\vec{G}}$ in time
\[f(|\vec{H}|,d)\cdot |\vec{G}|^{r+O(1)} \,. \]
Moreover, let $\vec{C}$ be a class of digraphs.
Then $\#\dirsubsprobd(\vec{C})$ is fixed-parameter tractable if ${\vec{C}}$ has bounded fractional cover number.
\end{theorem}
\begin{proof}
Given an instance $(\vec{H},\vec{G})$ we first cast the problem as a linear combination of homomorphism counts. Concretely, using Lemma~\ref{lem:subs_transformation}, we have
\begin{equation}\label{eq:subs_algo_basis}
    \#\subs{\vec{H}}{\vec{G}} = \sum_{\vec{F}} \asub_{\vec{H}}(\vec{F}) \cdot \#\homs{\vec{F}}{\vec{G}}\,,
\end{equation}
where the sum is over all (isomorphism classes of) digraphs $\vec{F}$ and the coefficients $\asub_{\vec{H}}(\vec{F})$ only depend on $\vec{H}$ and are non-zero if and only if $\vec{F}$ is a quotient graph of $\vec{H}$. Thus, we can proceed by computing all $\asub_{\vec{H}}(\vec{F})$ in time only depending on $\vec{H}$, and all terms $\#\homs{\vec{F}}{\vec{G}}$ with a non-zero coefficient using our algorithm for counting homomorphisms (Theorem~\ref{thm:homs_algo}): By Lemmas~\ref{lem:main_technical_subs} and~\ref{lem:subs_quotients_small_fhtw}, we have that, for each quotient $\vec{F}$ of $\vec{H}$,
\[\mathsf{fhtw}(\reach{\vec{F}}) \leq \rho^\ast(\vec{F}) \leq \rho^\ast(\vec{H}) = r \,.\]
Additionally, $|\vec{F}|\leq |\vec{H}|$. Thus, the computation of $\#\homs{\vec{F}}{\vec{G}}$ takes time $g(|\vec{H}|,d)\cdot |\vec{G}|^{r+O(1)}$, for some computable function $g$, by the running time bound given in Theorem~\ref{thm:homs_algo}, concluding the proof.
\end{proof}

\subsection{Lower Bounds}\label{sec:sub_LB}
Recall that $\alpha$ and $\alpha^\ast$ denote respectively the independence number and the fractional independence number. Let $\vec C$ be a class of digraphs, and let $\reduct{\vec C}$ the class of all contours of digraphs in~$\vec C$. We show that $\#\dirsubsprobd(\vec{C})$ is intractable when $\alpha^{\ast}(\reduct{\vec C}) = \infty$. Together with the upper bounds of Section~\ref{sec:sub_UB}, this yields a complete characterization of the tractability of $\#\dirsubsprobd(\vec{C})$.

To prove that $\#\dirsubsprobd(\vec{C})$ is hard when $\alpha^{\ast}(\reduct{\vec C}) = \infty$, we first look at the integral independence number $\alpha(\reduct{\vec C})$. We show that, if $\alpha(\reduct{\vec C}) = \infty$, then $\#\dirsubsprobd(\vec{C})$ is hard because $\vec C$ contains ``hard'' quotients. If instead $\alpha(\reduct{\vec C}) < \infty$ then we can show a reduction from $\#\dirhomsprobd(\vec{C}')$ where $\vec{C}'$ is a class of canonical DAGs with $\aw(\vec{C}')=\infty$, which implies hardness by Lemma~\ref{lem:main_hardness_homs}. 

First, using our interpolation result based on Dedekind's Theorem (Lemma~\ref{lem:complexity_monotonicity}), we establish a hardness result in two steps.

\begin{lemma}\label{lem:bottleneck_subs_setup}
Let $\vec{C}$ be a recursively enumerable class of digraphs and let $\vec{Q}$ be a class of quotient graphs of $\vec{C}$.
Then 
\[ \#\dirhomsprobd(\vec{Q}) \fptred \#\dirsubsprobd(\vec{C}) \,. \]
\end{lemma}
\begin{proof}
Let $\vec{H}'$ and $\vec{G}'$ be an input instance of $\#\dirhomsprobd(\vec{Q})$, and let $d$ be the outdegree of $\vec{G}'$. Search for a graph $\vec{H}\in \vec{C}$ such that $\vec{H}'$ is a quotient of $\vec{H}$ - note that this takes time only depending on $\vec{H}'$. By Lemma~\ref{lem:subs_transformation}, we have that
\[\#\subs{\vec{H}}{\star} = \sum_{\vec{F}} \asub_{\vec{H}}(\vec{F}) \cdot \#\homs{\vec{F}}{\star}\,.\]
Moreover, we have that $\asub_{\vec{H}}(\vec{H}')\neq 0$ since $\vec{H}'$ is a quotient of $\vec{H}$. Let us set $\iota = \asub_{\vec{H}}$. This allows us to invoke Lemma~\ref{lem:complexity_monotonicity} since we can then simulate the oracle required by Lemma~\ref{lem:complexity_monotonicity} using our own oracle for $\#\dirsubsprobd(\vec{C})$. The algorithm $\mathbb{A}$ in Lemma~\ref{lem:complexity_monotonicity} then returns all pairs $(\vec{F},\#\homs{\vec{F}}{\vec{G}'})$ with $\asub_{\vec{F}}\neq 0$; this includes $(\vec{H}',\#\homs{\vec{H}'}{\vec{G}'})$. All oracle queries posed by $\mathbb{A}$ have outdegree bounded by $f(|\iota|)\cdot d$, which guarantees that the parameter of each oracle call we forward to $\#\dirsubsprobd(\vec{C})$ only depends on $\vec{H}'$ (recall that the parameter is $|\vec{H}'|+d$). Moreover, the total running time is fixed-parameter tractable, concluding the proof.
\end{proof}

\begin{lemma}\label{lem:bottleneck_subs}
Let $\vec{C}$ be a recursively enumerable class of digraphs and let $\vec{Q}$ be the class of all quotient graphs of $\vec{C}$.
If the class of canonical DAGs that are MR minors of digraphs in $\vec{Q}$ has unbounded adaptive width, then $\#\dirsubsprobd(\vec{C})$ is not fixed-parameter tractable unless ETH fails.
\end{lemma}
\begin{proof}
Assume ETH holds.
By Lemma~\ref{lem:main_hardness_homs}, the problem $\#\dirhomsprobd(\vec{Q})$ is not fixed-parameter tractable. The claim thus follows by invoking Lemma~\ref{lem:bottleneck_subs_setup}.
\end{proof}

\subsubsection{The case of unbounded independence number}
Let us introduce:
\begin{definition}[Induced Matching Gadget]
Let $\vec{H}$ be a DAG. An \emph{induced matching gadget} of size $k$ of $\vec{H}$ is a set of arcs $(s_1,w_1),\dots,(s_k,w_k) \in E(\vec H)$ such that no two distinct $w_i,w_j$ are reachable from a single source of $\vec{H}$.
\end{definition}
We denote by $\mathsf{img}(\vec{H})$ the maximum size of an induced matching gadget in $\vec{H}$.  Given a directed (not necessarily acyclic) graph $\vec{H}$, we set $\mathsf{img}(\vec{H}):=\mathsf{img}(\Hsim)$.

\begin{lemma}\label{lem:ind_number_img}
Every digraph $\vec{H}$ satisfies $\alpha(\reduct{\Hsim})=\mathsf{img}(\vec{H})$.
\end{lemma}
\begin{proof}
First, we prove that $\mathsf{img}(\vec{H})\geq \alpha(\reduct{\Hsim})$. Let $\{v_1,\dots,v_{k}\}$ be an independent set of $\reduct{\Hsim}$, and let $s_1,\dots,s_{k}$ be sources of $\Hsim$ such that $v_i$ is reachable from $s_i$ for each $i\in[k]$. Note that those sources must exist: if $v_i$ is not reachable by any source, then $v_i$ is a source itself and thus not in $V(\reduct{\Hsim})$. Furthermore, those sources must be distinct: if $s_i=s_j$ for some $i \ne j$, then $v_i$ and $v_j$ can both be reached from $s_i=s_j$ and thus they are contained in a common edge of $\reduct{\Hsim}$, contradicting the fact that $\{v_1,\ldots,v_k\}$ is independent. 

For each $i\in[k]$ let $P_i$ be a shortest directed path from $s_i$ to $v_i$, and let $w_i$ be the successor of $s_i$ in $P_i$ (if $(s_i,v_i) \in E(\Hsim)$, then $w_i=v_i$).
Note that $w_1,\dots,w_{k}$ are pairwise distinct and form an independent set in $\reduct{\Hsim}$. Indeed, if $w_i=w_j$ for some $i\neq j$, then $v_i$ and $v_j$ can both be reached from $s_i$ and $s_j$ and are thus contained in a common hyperedge in $\reduct{\Hsim}$, a contradiction; if instead $\{w_1,\dots,w_{k}\}$ is not an independent set in $\reduct{\Hsim}$, then there is a source $s$ of $\Hsim$ from which both $w_i$ and $w_j$, and thus both $v_i$ and $v_j$, can be reached, which implies $v_i$ and $v_j$ are contained in a common edge of $\reduct{\Hsim}$, yielding again a contradiction.
Finally, observe that $(s_1,w_1),\dots,(s_{k},w_{k})$ is an induced matching gadget since $\{w_1,\ldots,w_{k}\}$ is an independent set in $\reduct{\Hsim}$.

Now we prove that $\alpha(\reduct{\Hsim}) \ge \mathsf{img}(\vec{H})$. Let $(s_1,w_1),\ldots,(s_k,w_k)$ be an induced matching gadget of $\vec H$, and for each $i \in [k]$ let $S_i$ and $W_i$ be the classes of $\sim$ containing respectively $s_i$ and $w_i$. Note that $S_i=\{s_i\}$ since $s_i$ is a source, so $W_i \ne S_i$; and $(S_i,W_i) \in E(\Hsim)$ since $(s_i,w_i) \in E(\vec H)$, hence $W_i$ is not a source of $\Hsim$. The definition of $V(\Gamma(\vec H))$ then implies $W_i \subseteq V(\Gamma(\vec H))$ for all $i \in [k]$, and so $\{w_1,\ldots,w_k\} \subseteq V(\Gamma(\vec H))$. Now observe that $\{w_1,\ldots,w_k\}$ is an independent set in $\Gamma(\Hsim)$; if this was not the case, then $w_i$ would be reachable in $\vec H$ from $w_j$, and thus from $s_j$, for some $j \ne i$, contradicting the definition of induced matching gadget.
\end{proof}

Next we show that large induced matching gadgets yield as quotients the directed splits of arbitrary graphs.
\begin{lemma}\label{lem:quotient_construction}
Let $F$ be an undirected graph with $\ell$ edges, and let $\vec{H}$ be a digraph with an induced matching gadget of size $2\ell$. Then $\dsplit{F}$ is an MR minor of a quotient graph of $\vec{H}$.
\end{lemma}
\begin{proof}
First we claim that we can assume w.l.o.g.\ that $\vec{H}$ is a DAG: In the very first step, we take the quotient of $\vec{H}$ corresponding to $\sim$. Note that the resulting graph is equal to $\Hsim$ except for possibly having loops. Note also that this operation does not change the size of a maximum induced matching gadget. Since all loops can in the end be deleted by the loop deletion operation of MR minors, to avoid notational clutter, we can assume that there are none.

Now, for each edge $e=\{u,v\}$ of $F$, we choose two arcs of the induced matching gadget that will correspond to $e$ in the construction of our quotient. We will denote those arcs by $(s_e^u,w_e^u)$ and $(s_e^v,w_e^v)$. Let us now define the partition $\sigma$ of $V(\vec{H})$ which will yield our quotient graph.
\begin{enumerate}
    \item For each vertex $v\in V(F)$, we add a block $B_v=\{w_e^v~|~e\in E(F)\}$.
    \item For each edge $e=\{u,v\}\in E(F)$, we add a block $B_e =\{s_e^u,s_e^v\}$.
    \item Each vertex of $\vec{H}$ not contained in any of the $B_v$ or $B_e$ becomes a singleton block.
\end{enumerate}
Let $\vec{H}'=\vec{H}/\sigma$, that is, $\vec{H}'$ is the quotient graph obtained from $\vec{H}$ by contracting each $B_v$ and each $B_e$ to a single vertex --- it will be convenient to also call those vertices $B_v$ and $B_e$. Observe that the subgraph of $\vec{H}'$ induced by the vertices $B_v$ for $v\in V(F)$ and $B_e$ for $e\in E(F)$ is isomorphic to $\dsplit{F}$: By construction of $\sigma$, it is clear that $\dsplit{F}$ is a \emph{subgraph} of the subgraph of $\vec{H}'$ induced by the $B_v$ and $B_w$. For isomorphism, we have to argue that there are no additional arcs: First, there cannot be any arc between $B_e$ and $B_{e'}$ for $e\neq e'$ since the $s_e^u$ and $s_e^v$ have been sources. Furthermore, there cannot be an arc between $B_{\{u,v\}}$ and $B_x$ for a vertex $x\notin\{u,v\}$ since this would only be possible if either $s_e^u$ or $s_e^v$ has an arc to some $w_{e'}^x$ for some $e\neq \{u,v\}$. However, in that case $w_{e'}^x$ and one of $w_{e}^u$ or $w_{e}^v$ would be reachable from the same source, contradicting the definition of an induced matching gadget. A similar argument shows that there cannot be an arc between $B_v$ and $B_u$ for two distinct vertices $u,v\in V(F)$. Also, we observe that the $B_e$ must be sources of $\vec{H}'$

Now perform the following operations on $\vec{H}'$ until none of them can be applied anymore:
\begin{itemize}
    \item Delete a sink that is not one of the $B_v$.
    \item Let $s$ be a source of $\vec{H}'$ not among the $B_e$, and let $y$ be a descendant of $s$ ($y$ might be one of the $B_v$). Contract the arc $(s,y)$.
\end{itemize}
This procedure stops if the only vertices remaining are the $B_e$ and the $B_v$. Crucially, the contraction of arcs from sources (not among the $B_e$) can never create additional arcs between the $B_e$ and the $B_v$ since, by definition of induced matching gadgets, no distinct pair of the $w_e^u$ is reachable from a common source.
\end{proof}

We are now able to establish hardness for the case of unbounded independence number of the contours.
\begin{lemma}\label{lem:subs_hard_independence}
Let $\vec{C}$ be a recursively enumerable class of digraphs, and let $\reduct{\vec{C}}$ be the contours of $\vec{C}$. If the independence number of $\reduct{\vec{C}}$ is unbounded, then $\#\dirsubsprobd(\vec{C})$ is not fixed-parameter tractable, unless ETH fails. 
\end{lemma}
\begin{proof}
By Lemma~\ref{lem:invariants_are_invariant} and Lemma~\ref{lem:ind_number_img}, the class $\vec{C}$ contains induced matching gadgets of unbounded size. Let $K$ be the family of all complete (undirected) graphs; clearly, the treewidth of $K$ is unbounded. Let furthermore $\dsplit{K}:= \{\dsplit{F}~|~F\in K\}$ be the set of all directed splits of complete graphs. Observe that $\dsplit{K}$ is a class of canonical DAGs, and observe further that $\reduct{\dsplit{K}}=K$: Given $\dsplit{F}\in \dsplit{K}$, each source $s$ of $\dsplit{F}$ corresponds to an edge $\{u,v\}$ of $F$, and the only two vertices reachable from $s$ in $\dsplit{F}$ are precisely $u$ and $v$. Hence, the reachability hypergraph is $3$-uniform and contains the hyperedges $\{s,u,v\}$. In the contour, we delete the former sources from each hyperedge, which then yields $F$ again (thinking of a graph as a $2$-uniform hypergraph).

Now, it is well-known that adaptive width and tree-width are equivalent for graphs. Hence the adaptive width of $\dsplit{K}$ must be unbounded (recall that the adaptive width of a digraph is defined to be the adaptive width of its contour).

Finally, by Lemma~\ref{lem:quotient_construction}, and using that $\vec{C}$ has induced matching gadgets of unbounded size, we obtain that the set of quotient graphs of digraphs in $\vec{C}$ admits as MR minors the canonical DAGs in $\dsplit{K}$. Since the adaptive width of the latter is unbounded, we can conclude the proof by applying Lemma~\ref{lem:bottleneck_subs}.
\end{proof}

\subsubsection{The case of bounded independence number}
Let $\scH$ be a nonempty hypergraph. Without loss of generality we may assume $V(\scH)=\cup E(\scH)$. Let $(T,B)$ be a tree decomposition of $\scH$.
For every $\{r,u\} \in E(T)$ let $T^r_u$ be the connected component of $T\setminus e$ containing $u$ but not $r$, and define:
\begin{align}
V^r_u = (\cup_{x \in V(T^r_u)}B_x) \setminus (B_r \cap B_u)
\end{align}
Note that $V(\scH) = V^r_u \;\dot\cup\; V^u_r \;\dot\cup\; (B_r\cap B_u)$.

The next result bounds the integrality gap of $\alpha^{\ast}(\scH)$ through $\aw(\scH)$.
\begin{lemma}\label{lem:alpha_aw}
$\alpha(\scH) \ge \frac{1}{2}+\frac{\alpha^{\ast}(\scH)}{4\aw(\scH)}$.
\end{lemma}
\begin{proof}
We use induction on $|E(\scH)|$. If $|E(\scH)| = 1$ then one can see that $\alpha(\scH)=\alpha^{\ast}(\scH)=\aw(\scH)$, so the claim holds.
Now suppose $|E(\scH)|>1$, and assume the claim holds for every hypergraph with less than $|E(\scH)|$ edges. Let $\mu:V(\scH) \to \R_{\ge 0}$ be a fractional independent set for $\scH$ with $\mu(V(\scH))=\alpha^{\ast}(\scH)$, and let $(T,\scB)$ be a tree decomposition for $\scH$ of smallest order (i.e., that minimizes $|V(T)|$) such that $\mu$-width$(T,\scB) \le \aw(\scH)$. Choose any $\{r,u\} \in E(T)$ and let $S = B_r \cap B_u$. Finally, let $\scC(S)$ be the set of connected components of $\scH \setminus S$. Note that no $e \in E(\scH)$ intersects two distinct elements of $\scC(S)$: indeed, by the properties of tree decompositions $S$ separates $V_u^r$ and $V_r^u$, and any such $e$ would intersect both $V^r_u$ and $V^u_r$, a contradiction.

We can now deduce the following facts. First, $\alpha(\scH) \ge \sum_{C \in \scC(S)} \alpha(\scH[C])$, since no $e \in E(\scH)$ intersects more than one element of $\scC(S)$.
Second, $|\scC(S)|\ge 2$; indeed, if $|\scC(S)| \le 1$, then $B_r \subseteq B_u$ (or vice versa) and thus we could replace $\{B_r,B_u\}$ with $B_u$ (or with $B_r$) without increasing $\mu$-width$(T,\scB)$, contradicting the minimality of $|V(T)|$.
Third, $|E(\scH[C])| < |E(\scH)|$ for all $C \in \scC(S)$; indeed, every $C \in \scC(S)$ is intersected by some $e \in E(\scH)$, and as noted above $|\scC(S)|\ge 2$ and no $e \in E(\scH)$ intersects more than one element of $\scC(S)$.
Using these facts and the inductive hypothesis on each $\scH[C]$, we obtain:
\begin{align}
    \alpha(\scH) &\ge \sum_{C \in \scC(S)} \alpha(\scH[C]) \\
    & \ge \sum_{C \in \scC(S)} \left(\frac{1}{2}+ \frac{\alpha^{\ast}(\scH[C])}{4 \aw(\scH[C])}\right) \\
    \\
    & \ge \sum_{C \in \scC(S)} \left(\frac{1}{2}+ \frac{\alpha^{\ast}(\scH[C])}{4 \aw(\scH)}\right) \\
    &\ge 1 + \frac{1}{4 \aw(\scH)}\sum_{C \in \scC(S)} \alpha^{\ast}(\scH[C]) 
\end{align}
Note that $\alpha^{\ast}(\scH[C]) \ge \mu(C)$ since the restriction of $\mu$ to $C$ is a fractional independent set for $\scH[C]$. Moreover $\sum_{C \in \scC(S)}\mu(C)=\mu(V(\scH))-\mu(S)$, and $\mu(S) \le \aw(\scH)$ by the choice of $(T,\scB)$. Therefore:
\begin{align}
    \alpha(\scH) &\ge 1 + \frac{1}{4 \aw(\scH)}\sum_{C \in \scC(S)} \mu(C) \\
    &= 1 + \frac{1}{4 \aw(\scH)}(\mu(V(\scH)) - \mu(S)) \\
    &\ge 1 + \frac{\alpha^{\ast}(\scH) -  \aw(\scH)}{4 \aw(\scH)}\\
    &> \frac{1}{2} + \frac{\alpha^{\ast}(\scH)}{4 \aw(\scH)}
\end{align}
which concludes the proof.
\end{proof}

\paragraph*{A construction of Canonical DAGs that preserves $\alpha^{\ast}$}

In what follows, recall that the vertices of the DAG $\Hsim$ are the strongly connected components of $\vec{H}$; we use capital letters to denote the vertices of $\Hsim$.
\begin{lemma}\label{lem:push_up_alpha_general}
Let $\vec H$ be a digraph and $(U,V) \in E(\Hsim)$ where $U$ is not a source of $\Hsim$. Then there is a fractional independent set $\hat{\mu}$ for $\reduct{\vec H}$ of maximum weight such that $\hat{\mu}(v)=0$ for every $v\in V$.
\end{lemma}
\begin{proof}
Let $\hat{\mu}^*:V(\reduct{\vec H}) \to \R_{\ge 0}$ be any fractional independent set for $\reduct{\vec H}$ of maximum weight. By Lemma~\ref{lem:invariants_are_invariant}, there exists a maximum fractional independent set $\mu^*:V(\reduct{\Hsim}) \to \R_{\ge 0}$, the weight of which equal to the weight of $\hat{\mu}^*$.

We define a fractional independent set $\mu$ of $\reduct{\Hsim}$ as follows:
\begin{align}
\mu(X) = \left\{
    \begin{array}{ll}
        \mu^*(X) & X \notin \{U,V\} \\
        \mu^*(U)+\mu^*(V) & X=U \\
        0 & X=V
    \end{array}\right.
\end{align}
Clearly $\mu$ and $\mu^*$ have the same weight. 
Now let $e \in E(\reduct{\Hsim})$. If $U \notin e$ then $\mu(e) \le \mu^*(e)$. Otherwise $\{U,V\} \subseteq e$, and since $\mu(U)+\mu(V)=\mu^*(U)+\mu^*(V)$, then again $\mu(e) \le \mu^*(e)$. Therefore $\mu$ is a fractional independent set for $\reduct{\Hsim}$. Note that $\mu$ must also be of maximum weight since otherwise, $\mu^*$ would not have been of maximum weight.

Now define $\hat{\mu}: V(\reduct{\vec{H}})\to \R_{\ge 0}$ as follows: For any strongly connected component $X=x_1,\dots,x_k$ of $\vec{H}$ we set $\hat{\mu}(x_1)=\mu(X)$ and $\hat{\mu}(x_i)=0$ for all $i\in\{2,\dots,k\}$. Clearly, $\hat{\mu}$ has the same total weight as $\mu$. Furthermore, note that the hyperedges of $\reduct{\Hsim}$ are obtained from the hyperedges of $\reduct{\vec{H}}$ by contracting each vertex set $X$ corresponding to a strongly connected component in $\vec{H}$ into a single vertex. For each such set $X$ and hyperedge $e\in E(\reduct{\vec{H}})$ we have that either $X\subseteq e$ or $e\cap X = \emptyset$ --- this follows from the definition of reachability hypergraphs and of the contour. Thus $\hat{\mu}$ is a fractional independent set of $\reduct{\vec{H}}$. Furthermore, since $\hat{\mu}$ has the same weight as $\mu$ and since $\mu$ is a fractional independent set of $\reduct{\Hsim}$ of maximum weight, we have by Lemma~\ref{lem:invariants_are_invariant} that $\hat{\mu}$ is of maximum weight as well. Finally, $\hat{\mu}(v)=0$ for each $v\in V$ since $\mu(V)= 0$, concluding the proof.
\end{proof}

\begin{lemma}\label{lem:from_dags_to_canonical_general}
For every digraph $\vec H$ there is a digraph $\vec F$ such that the following conditions are satisfied:
\begin{enumerate}
    \item $\vec{F}$ can be obtained from $\vec{H}$ via a sequence of sink deletions,
    \item $\alpha^{\ast}(\reduct{\vec F}) \ge \alpha^{\ast}(\reduct{\vec H})$, and
    \item $\SCCquot{F}$ is a canonical DAG.
\end{enumerate}
\end{lemma}
\begin{proof}
If $\Hsim$ is a canonical DAG then we set $\vec{F}=\vec{H}$.

If $\Hsim$ is not a canonical DAG, then there is an arc $(U,V) \in E(\Hsim)$ such that both $U$ and $V$ are not sources. Moreover we can assume $V$ is a sink (otherwise replace $U$ with $V$, and $V$ with one of its children). By Lemma~\ref{lem:push_up_alpha_general}, there is a fractional independent set $\hat{\mu}$ for $\reduct{\vec H}$ of maximum weight such that $\hat{\mu}(v)=0$ for every $v\in V$. Since $V$ is a sink of $\Hsim$, we can set $\vec{H}'= \vec{H} \setminus V$, that is, we perform a sink deletion. Let $\hat{\mu}'$ be the restriction of $\hat{\mu}$ to $V(\reduct{\vec{H}'})$. Note that $V(\reduct{\vec H'}) = V(\reduct{\vec H'}) \setminus V$; since $\hat{\mu}(v)=0$ for all $v\in V$, this implies that $\hat{\mu}'$ has the same weight as $\hat{\mu}$.

Moreover $\hat{\mu}'$ is clearly a fractional independent set for $\reduct{\vec H'}$, since for every $e' \in E(\reduct{\vec H'})$ there is $e \in E(\reduct{\vec H})$ such that $e' \subseteq e$, and $\hat{\mu}(e) = \hat{\mu}'(e')$. We conclude that $\alpha^{\ast}(\reduct{\vec H'}) \ge \alpha^{\ast}(\reduct{\vec H})$.
If $\SCCquot{H'}$ is a canonical DAG, then setting $\vec F=\vec H'$ concludes the proof; otherwise just repeat the argument on $\vec H'$.
\end{proof}

We are now able to prove the intractability part of our classification.
\begin{lemma}\label{lem:subs_hard}
Let $\vec{C}$ be a recursively enumerable class of digraphs of unbounded fractional cover number. Then $\#\dirsubsprobd(\vec{C})$ is not fixed-parameter tractable, unless ETH fails.
\end{lemma}
\begin{proof}
Let $\reduct{\vec{C}}$ be the class of contours of digraphs in $\vec{C}$. 
By Observation~\ref{obs:fractional_equivalence} and Lemma~\ref{lem:invariants_are_invariant} $\reduct{\vec{C}}$ has unbounded fractional edge cover number and thus, by Fact~\ref{fact:LP_duality}, $\reduct{\vec{C}}$ has unbounded fractional independence number.
The proof now considers two cases.

\noindent \textbf{Case 1:} The independence number of $\reduct{\vec{C}}$ is unbounded. Then the claim follows from the construction based on induced matching gadgets (Lemma~\ref{lem:subs_hard_independence}).

\noindent \textbf{Case 2:} The independence number of $\reduct{\vec{C}}$ is bounded. We show that $\#\dirhomsprobd(\vec{C})$ is not fixed-parameter tractable, unless ETH fails. The claim then follows since, by Lemma~\ref{lem:bottleneck_subs_setup} and the trivial fact that each digraph is a quotient graph of itself, we have
\[\#\dirhomsprobd(\vec{C}) \fptred \#\dirsubsprobd(\vec{C}) \,.\]
Let $\vec{C}'$ be the class of all digraphs $\vec{H}'$ such that 
\begin{enumerate}
    \item $\vec{H}'$ can be obtained by a sequence of sink deletions from a graph $\vec{H}\in \vec{C}$, disallowing deletions of sinks that are also sources, and
    \item $\SCCquot{H'}$ is a canonical DAG.
\end{enumerate}
Now note that a sink deletion in a digraph $\vec{H}$ corresponds to vertex-deletions in the contour. More precisely, let $\vec{H}'$ be obtained from $\vec{H}$ by deleting the sink $T$ of $\Hsim$. Then $\reduct{\vec{H}'}=\reduct{\vec{H}}[V(\vec{H})\setminus T]$. Thus, clearly, the independence number of $\reduct{\vec{H}'}$ is upper bounded by the independence number of $\reduct{\vec{H}}$. Thus, the independence number of the class $\reduct{\vec{C}'}$ of the contours of digraphs in $\vec{C}'$ is bounded (since the independence number of $\reduct{\vec{C}}$ is bounded by the assumption of this case.). 
Next, by Lemma~\ref{lem:from_dags_to_canonical_general}, we have that the fractional independence number of $\reduct{\vec{C}'}$ is still unbounded. In combination with Lemma~\ref{lem:alpha_aw}, this is only possible if the adaptive width of $\reduct{\vec{C}'}$ is unbounded.

Next, let $\SCCquot{C'}$ be the class of all DAGs $\SCCquot{H'}$ with $\vec{H}'\in \vec{C}'$. By definition of $\vec{C}'$, each element of $\SCCquot{C'}$ must be a canonical DAG. Now note that the contours of the canonical DAGs in $\SCCquot{C'}$ can be obtained from the contours of digraphs in $\vec{C}'$ by a sequence of contractions of vertices $u$ and $v$ such that $u$ and $v$ are contained in precisely the same hyperedges (since we contract strongly connected components into single vertices). Those contractions cannot decrease the adaptive width as shown in Lemma~\ref{lem:invariants_are_invariant}.

As a consequence, we have established the following three facts:
\begin{itemize}
    \item[(A)] The elements in $\SCCquot{C'}$ are MR minors of the digraphs in $\vec{C}$: $\vec{C}'$ is obtained by sink-deletions, and $\SCCquot{C'}$ is obtained by arc contractions --- each strongly connected component can be contracted into a single vertex using only arc contractions.
    \item[(B)] The adaptive width of $\SCCquot{C'}$ is unbounded.
    \item[(C)] Each element of $\SCCquot{C'}$ is a canonical DAG.
\end{itemize}
In combination, (A), (B), and (C) allow for the application of Lemma~\ref{lem:main_hardness_homs} which yields intractability as desired, concluding the proof.
\end{proof}

\subsection{Proof of the Classification}
We are now able to combine our upper and lower bounds and prove a complete and explicit classification:

\begin{theorem}\label{thm:main_result_subs}
Let $\vec{C}$ be a recursively enumerable class of digraphs and assume that ETH holds. Then the problem $\#\dirsubsprobd(\vec{C})$ is fixed-parameter tractable if and only if the fractional cover number of $\vec{C}$ is bounded.
\end{theorem}
\begin{proof}
The ``if'' direction is Theorem~\ref{thm:main_algo_subs}, and the ``only if'' direction is Lemma~\ref{lem:subs_hard}.
\end{proof}

Furthermore, we note that all of our hardness results apply also in case of digraphs without loops, and even for DAGs; the reason for this is two-fold: First the intractability boils down to counting homomorphisms from canonical DAGs to DAGs of small outdegree via our reduction from $\#\textsc{CSP}$ (see Lemma~\ref{lem:cphom_to_cpdirhom} and Lemma~\ref{lem:coloured_CSP_reduction}). Second, each intermediate step in our reduction, including the interpolation method based on Dedekind's Theorem (Lemma~\ref{lem:complexity_monotonicity}), does not create any additional cycles (including loops) in the host, provided our pattern is without loops and cycles. Concretely, we obtain the following variations of Theorem~\ref{thm:main_result_subs}.

\begin{theorem}\label{thm:main_result_subs_self-loop-free}
Let $\vec{C}$ be a recursively enumerable class of digraphs without loops and assume that ETH holds. Then the problem $\#\dirsubsprobd(\vec{C})$, restricted on host graphs without loops, is fixed-parameter tractable if and only if the fractional cover number of $\vec{C}$ is bounded.
\end{theorem}

\begin{theorem}\label{thm:main_result_subs_DAGs}
Let $\vec{C}$ be a recursively enumerable class of DAGs and assume that ETH holds. Then the problem $\#\dirsubsprobd(\vec{C})$, restricted on acyclic host graphs, is fixed-parameter tractable if and only if the fractional cover number of $\vec{C}$ is bounded.
\end{theorem}

\section{Counting Induced Subgraphs}
We will establish boundedness of the following invariant as sufficient and necessary condition for the fixed-parameter tractability of $\#\dirsubsprobd$.
\begin{definition}[$\scount{\vec{H}}$]
Given a digraph $\vec{H}$, we denote by $\scount{\vec{H}}$ the number of sources of the DAG $\Hsim$. 
\end{definition}

\subsection{Implications of Upper and Lower Bounds on $\scount{\vec{H}}$}

\begin{lemma}[Lower Bound]\label{lem:indsub_lower}
Let $F$ be an undirected graph with $k$ vertices and $\ell$ edges, and let $\vec{H}$ be a digraph with $\scount{\vec{H}}\geq k+\ell$. Then there exists an arc supergraph $\vec{H}'$ of $\vec{H}$ satisfying that $\dsplit{F}$ is an MR minor of $\vec{H}'$.
\end{lemma}
\begin{proof}
We will first construct $\vec{H}'$. To this end, observe that $\dsplit{F}$ has precisely $k+\ell$ vertices, and we assume w.l.o.g.\ that the vertex set of $\dsplit{F}$ is $[k+\ell]$. Since $\scount{\vec{H}}\geq k+\ell$ there exists set of sources $S_1,\dots,S_{k+\ell}$ of $\Hsim$. We emphasize that a set of sources must always also be an independent set, since there cannot be arcs between the sources. Recall that, by definition of $\Hsim$, the $S_i$ are strongly connected components of $\vec{H}$. Now pick $s_i\in S_i$ arbitrarily for each $i\in[k+\ell]$.
We obtain the graph $\vec{H}'$ from $\vec{H}$ as follows: Whenever $(i,j)$ is an arc of $\dsplit{F}$, we add an arc $(s_i,s_j)$ to $\vec{H}$. Observe that we do not create loops in this construction since $\dsplit{F}$ does not contain loops. For this proof, it will also be convenient to consider the digraph ${\vec{G}}$ obtained from $\Hsim$ by adding an arc $(S_i,S_j)$ whenever $(i,j)$ is an arc of $\dsplit{F}$. We will see that $\vec{G}=\SCCquot{H'}$.

Observe that $\vec{G}$ is still acyclic: Assuming otherwise, there must be a (not necessarily simple) directed cycle in $\vec{G}$. Since $\Hsim$ is acyclic, we have that at least on arc of the cycle must be one of the freshly added arcs $(S_i,S_j)$. Additionally, there must be a directed path $P$ from $S_j$ to $S_i$ in $\vec{G}$. Consider two cases: If $P$ contains any vertex $V$ which is not a source of $\Hsim$, then we obtain a contradiction immediately, since it is not possible to reach any source from $V$ by a directed path --- recall that we only added arcs between sources in the construction of $\vec{G}$. Otherwise, all vertices of $P$ are sources. However, in this case we created a cycle in $\vec{G}$ only consisting of sources, and thus, this cycle must correspond to a cycle in $\dsplit{F}$ by construction of $\vec{G}$. Since $\dsplit{F}$ is acyclic by definition, we obtain the contradiction.

Next we claim that the compositions into strongly connected components of $\vec{H}$ and $\vec{H}'$ are the same, that is, the relation $\sim$ has the same equivalence classes in both $V(\vec{H})$ and $V(\vec{H}')$. Assume for contradiction that this is not the case. 
Since adding arcs can only merge strongly connected components, there must be vertices $x$ and $y$ which are not in the same strongly connected component of $\vec{H}$, but they are in the same strongly connected component in $\vec{H}'$, that is, there is a (not necessarily simple) directed cycle in $\vec{H}'$ containing both $x$ and $y$. Let $X\neq Y$ be the connected components of $\vec{H}$ containing $x$ and $y$, respectively. Now identify all vertices in this cycle that are in the same connected component of $\vec{H}$ and observe that this creates a cycle in $\vec{G}$ (note that we needed the assumption that $X\neq Y$ to make sure that the entire cycle does not collapse to a single vertex in $\vec{G}$). Since $\vec{G}$ is acyclic, we obtain the contradiction. 

As a consequence, we infer that $\vec{G}$ is indeed the graph $\SCCquot{H'}$. Using this fact, we are able to show that $\dsplit{F}$ is an MR minor of $\vec{H}'$: First, contract each strongly connected component of $\vec{H}'$ into a single vertex. By definition this yields precisely $\SCCquot{H'} ~= \vec{G}$. Next we iteratively delete sinks until only the vertices $S_1,\dots,S_{k+\ell}$ remain and claim that the resulting graph is $\dsplit{F}$ as desired. To see this, observe that there is no vertex in $V(\vec{G})\setminus\{S_1,\dots, S_{k+\ell}\}$ from which one can reach any of the $S_1,\dots,S_k$; this is true since the $S_i$ have been sources in $\Hsim$. Therefore, in combination with the fact that $\vec{G}$ is acyclic, we have that there must always be a sink outside of the $S_1,\dots,S_{k+\ell}$ as long as there are still vertices outside of the $S_1,\dots,S_{k+\ell}$ remaining. Note that the latter property is invariant under the deletion of sinks outside of $S_1,\dots,S_{k+\ell}$ --- for clarification we note that even former sources of $\Hsim$ not included in $\{S_1,\dots, S_{k+\ell}\}$ will be deleted at some point, since they become sinks if all of their descendants have been deleted in previous iterations. Hence, we can conclude that at the end of the process of iteratively deleting sinks, we obtain the (induced) subgraph of $\vec{G}$ only consisting of the $S_1,\dots,S_{k+\ell}$ which is equal to $\dsplit{F}$ (recall that we added an arc between $S_i$ and $S_j$ if and only if there is an arc $(i,j)$ in $\dsplit{F}$, and since the $S_1,\dots,S_{k+\ell}$ have been sources in $\Hsim$ there are no further arcs between them.)
\end{proof}

Next we invoke our reduction based on Dedekind's Theorem to the case of counting induced subgraphs:
\begin{lemma}\label{lem:bottleneck_indsubs_setup}
Let $\vec{C}$ be a recursively enumerable class of digraphs and let $\vec{A}$ be a class of arc supergraphs of digraphs in $\vec{C}$.
Then 
\[ \#\dirhomsprobd(\vec{A}) \fptred \#\dirindsubsprobd(\vec{C}) \,. \]
\end{lemma}
\begin{proof}
Let $\vec{H}'$ and $\vec{G}'$ be an input instance of $\#\dirhomsprobd(\vec{A})$, and let $d$ be the outdegree of $\vec{G}'$. Search for a graph $\vec{H}\in \vec{C}$ such that $\vec{H}'$ is an arc supergraph of $\vec{H}$ - note that this takes time only depending on $\vec{H}'$. By Lemma~\ref{lem:indsubs_transformation}, we have that
\[\#\indsubs{\vec{H}}{\star} = \sum_{\vec{F}} \aindsub_{\vec{H}}(\vec{F}) \cdot \#\homs{\vec{F}}{\star}\,.\]
Moreover, we have that $\aindsub_{\vec{H}}(\vec{H}')\neq 0$ since $\vec{H}'$ is an arc supergraph of $\vec{H}$ (see condition 3.\ in Lemma~\ref{lem:indsubs_transformation}). Let us set $\iota = \aindsub_{\vec{H}}$. This allows us to invoke Lemma~\ref{lem:complexity_monotonicity} since we can then simulate the oracle required by Lemma~\ref{lem:complexity_monotonicity} using our own oracle for $\#\dirindsubsprobd(\vec{C})$. The algorithm $\mathbb{A}$ in Lemma~\ref{lem:complexity_monotonicity} then returns all pairs $(\vec{F},\#\homs{\vec{F}}{\vec{G}'})$ with $\aindsub_{\vec{F}}\neq 0$; this includes $(\vec{H}',\#\homs{\vec{H}'}{\vec{G}'})$. All oracle queries posed by $\mathbb{A}$ have outdegree bounded by $f(|\iota|)\cdot d$, which guarantees that the parameter of each oracle call we forward to $\#\dirindsubsprobd(\vec{C})$ only depends on $\vec{H}'$ (recall that the parameter is $|\vec{H}'|+d$). Moreover, the total running time is fixed-parameter tractable, concluding the proof.
\end{proof}

Now, relying on our reduction chain based on Dedekind's Theorem and on our hardness result for $\#\dirhomsprobd$, we can prove the following intractability result.

\begin{lemma}\label{lem:indsubs_hard}
Let $\vec{C}$ be a recursively enumerable class of digraphs and assume that ETH holds. If $\scount{\vec{C}}$ is unbounded then $\#\dirindsubsprobd(\vec{C})$ is not fixed-parameter tractable.
\end{lemma}
\begin{proof}
The setup is similar to the proof of Lemma~\ref{lem:subs_hard_independence}:
Let $K$ be the family of all complete (undirected) graphs. The treewidth of $K$ is unbounded. Let furthermore $\dsplit{K}:= \{\dsplit{F}~|~F\in K\}$ be the set of all directed splits of complete graphs. Observe that $\dsplit{K}$ is a class of canonical DAGs, and observe further that $\reduct{\dsplit{K}}=K$. Since adaptive width and treewidth are equivalent for graphs, the adaptive width of $\dsplit{K}$ must be unbounded (recall that the adaptive width of a digraph is defined to be the adaptive width of its contour).

Finally, by Lemma~\ref{lem:indsub_lower}, and using that $\scount{\vec{C}}$ is unbounded, we obtain that the set of arc supergraphs of digraphs in $\vec{C}$ admits as MR minors the canonical DAGs in $\dsplit{K}$. Since the adaptive width of the latter is unbounded, we can conclude the proof by applying Lemma~\ref{lem:bottleneck_indsubs_setup} and Lemma~\ref{lem:main_hardness_homs}.
\end{proof}

\begin{lemma}[Upper Bound]\label{lem:indsub_upper}
Let $\vec{H}$ be a digraph with $\scount{\vec{H}}\leq c$, and let $\vec{H}'$ be a quotient of an arc supergraph of $\vec{H}$. Then $|E(\mathcal{R}(\vec{H}'))|\leq c$.
\end{lemma}
\begin{proof}
Observe that neither of the operations of adding arcs to or identifying vertices of a digraph $\vec{F}$ can increase the number of sources of $\SCCquot{F}$. Thus $\scount{\vec{H}'}\leq c$ as well. By definition of reachability hypergraphs, we can immediately conclude that $|E(\mathcal{R}(\vec{H}'))|\leq c$ since we create one hyperedge for each source of $\SCCquot{H'}$.
\end{proof}

We obtain the following algorithm.
\begin{theorem}\label{thm:indsubs_algo}
There is a computable function $f$ such that the following is true. Let $\vec{H}$ and $\vec{G}$ be digraphs, let $d$ be the maximum outdegree of $\vec{G}$, and let $r=\scount{\vec{H}}$. We can compute $\#\indsubs{\vec{H}}{\vec{G}}$ in time
\[f(|\vec{H}|,d)\cdot |\vec{G}|^{r+O(1)} \,. \]
Moreover, let $\vec{C}$ be a class of digraphs.
Then $\#\dirindsubsprobd(\vec{C})$ is fixed-parameter tractable if $\scount{\vec{C}}$ is bounded.
\end{theorem}
\begin{proof}
By Lemma~\ref{lem:indsubs_transformation}, we have
\begin{equation}
    \#\indsubs{\vec{H}}{\vec{G}} = \sum_{\vec{F}} \aindsub_{\vec{H}}(\vec{F}) \cdot \#\homs{\vec{F}}{\vec{G}}\,,
\end{equation}
such that the following conditions are satifsied:
\begin{enumerate}
    \item\label{bul:b12} ${\aindsub}_{\vec{H}}$ has finite support and only depends on $\vec{H}$.
    \item\label{bul:b22} If ${\aindsub}_{\vec{H}}(\vec{F})\neq 0$ then $\vec{F}$ is a quotient of an arc supergraph of $\vec{H}$.
\end{enumerate}
By Lemma~\ref{lem:indsub_upper}, the second condition implies that the only terms $\#\homs{\vec{F}}{\vec{G}}$ surviving with a non-zero coefficient satisfy $|E(\mathcal{R}(\vec{F}))|\leq r$. Clearly, this also implies that $\mathsf{fhtw}(\reach{\vec{F}}) \leq r$. Note further that the size of any quotient of any arc supergraph of $\vec{H}$ is bounded by $|\vec{H}|^2$. Hence, using the algorithm for counting homomorphisms in Theorem~\ref{thm:homs_algo} for each term $\#\homs{\vec{F}}{\vec{G}}$ with a non-zero coefficient we can evaluate the linear combination in time \[f(|\vec{H}|,d)\cdot |\vec{G}|^{r+O(1)}\] for some computable function $f$. This concludes the proof.
\end{proof}

\subsection{Proof of the Classification}
We are now able to combine our upper and lower bounds and prove a complete and explicit classification:

\begin{theorem}\label{thm:main_result_indsubs}
Let $\vec{C}$ be a recursively enumerable class of digraphs and assume that ETH holds. Then the problem $\#\dirindsubsprobd(\vec{C})$ is fixed-parameter tractable if and only if $\scount{\vec{C}}$ is bounded.
\end{theorem}
\begin{proof}
The ``if'' direction is Theorem~\ref{thm:indsubs_algo}, and the ``only if'' direction is Lemma~\ref{lem:indsubs_hard}.
\end{proof}

Finally, similarly to the case of counting subgraphs, our proofs readily classify also the cases of digraphs without loops and DAGs:
\begin{theorem}\label{thm:main_result_indsubs_self-loop-free}
Let $\vec{C}$ be a recursively enumerable class of digraphs without loops and assume that ETH holds. Then the problem $\#\dirindsubsprobd(\vec{C})$, restricted on host graphs without loops, is fixed-parameter tractable if and only if $\scount{\vec{C}}$ is bounded.
\end{theorem}

\begin{theorem}\label{thm:main_result_indsubs_DAGs}
Let $\vec{C}$ be a recursively enumerable class of DAGs and assume that ETH holds. Then the problem $\#\dirindsubsprobd(\vec{C})$, restricted on acyclic host graphs, is fixed-parameter tractable if and only if $\scount{\vec{C}}$ is bounded.
\end{theorem}

\bibliography{conference.bib}

\pagebreak

\appendix

\section{Proof of Dedekind's Theorem}\label{app:dedekind}

\begin{theorem}[Theorem~\ref{thm:dedekind}, restated]\label{thm:dedekind_restated}
Let $(\mathrm{G},\ast)$ be a semigroup. Let $(\varphi_i)_{i\in[k]}$ with $\varphi_i: \mathrm{G} \to \Q$ be pairwise distinct semigroup homomorphisms of  $(\mathrm{G},\ast)$ into $(\Q,\cdot)$, that is, $\varphi_i(g_1\ast g_2)= \varphi_i(g_1)\cdot \varphi_i(g_2)$ for all $i\in[k]$ and $g_1,g_2\in \mathrm{G}$. Let $\phi : G \to \Q$ be a function
\begin{equation}\label{eq:dedekind_restated}
    \phi : g \mapsto \sum_{i=1}^k a_i \cdot \varphi_i(g)\,, 
\end{equation}
where the $a_i$ are rational numbers.
Suppose furthermore that the following functions are computable:
\begin{enumerate}
    \item The operation $\ast$.
    \item The mapping $(i,g) \mapsto \varphi_i(g)$.
    \item A mapping $i\mapsto g_i$ such that $\varphi_i(g_i)\neq 0$.
    \item A mapping $(i,j) \mapsto g_{i,j}$ such that $\varphi_i(g_{i,j})\neq \varphi_j(g_{i,j})$ whenever $i\neq j$.
\end{enumerate}
Then there is a constant $B$ only depending on the $\varphi_i$ (and not on the $a_i$), and an algorithm  $\hat{\mathbb{A}}$ such that the following conditions are satisfied:
\begin{itemize}
    \item $\hat{\mathbb{A}}$ is equipped with oracle access to $\phi$.
    \item $\hat{\mathbb{A}}$ computes $a_1,\ldots,a_k$.
    \item Each oracle query $\hat{g}$ only depends on the $\varphi_i$ (and not on the $a_i$).
    \item The running time of $\hat{\mathbb{A}}$ is bounded by $O\left( B \cdot \sum_{i=1}^k \log a_i \right)$
\end{itemize}
\end{theorem}
\begin{proof}
Let $g_i$ and $g_{i,j}$ be as in 3.\ and 4.\ in the statement of the theorem. 
The algorithm $\hat{\mathbb{A}}$ will perform recursion over $k$: If $k=1$, then we just output
\[\varphi_1(g_1)^{-1} \cdot \phi(g_1) = a_1\,.\]
If $k> 1$, we consider the element $g_{1,k}$ and observe that for all $g\in \mathrm{G}$ we have
\begin{equation}\label{eq:dedeproof1}
    \sum_{i=1}^k a_i \cdot \varphi_i(g_{1,k}\ast g) - \varphi_k(g_{1,k})\cdot \sum_{i=1}^k a_i \cdot \varphi_i(g) = \sum_{i=1}^{k-1} a_i \cdot (\varphi_i(g_{1,k})- \varphi_k(g_{1,k})) \cdot \varphi_i(g)\,. 
\end{equation}
Now set $\hat{a}_i:= a_i \cdot (\varphi_i(g_{1,k})- \varphi_k(g_{1,k}))$ for all $0<i<k$, and observe that \eqref{eq:dedeproof1} enables us to use our oracle to simulate an oracle for
\[g \mapsto \sum_{i=1}^{k-1} \hat{a}_i \cdot \varphi_i(g)\,.\]
By recursion, we can thus obtain the value $\hat{a}_1$. Since $\varphi_1(g_{1,k})\neq \varphi_k(g_{1,k})$, we are able to compute $a_1 = {\hat{a}_1}\cdot (\varphi_1(g_{1,k})- \varphi_k(g_{1,k}))^{-1}$. Knowing $a_1$, we can use our oracle to simulate an oracle for
\[g \mapsto \sum_{i=1}^k a_i \cdot \varphi_i(g) - a_1\cdot \varphi_1(g) = \sum_{i=2}^{k} a_i \cdot \varphi_i(g) \,.\]
Thus we can go into recursion again and compute $a_2,\dots,a_k$.

Note that all oracle queries only depend on the $\varphi_i$, and the same holds true for the recursion depth, and thus the number of arithmetic operations. This yields the desired running time; we emphasize that we need the factor of $\sum_{i=1}^{k} \log a_i$ in the running time to perform the arithmetic operations.
\end{proof}

\section{Classifications for Unbounded Outdegrees}\label{app:unbounded_degree}
In this final section we quickly describe how the classifications for counting homomorphisms, subgraphs and induced subgraphs are easy consequences of the works of Dalmau and Jonsson~\cite{DalmauJ04}, and of Curticapean, Dell and Marx~\cite{CurticapeanDM17}, using our interpolation method via Dedekind's Theorem.

Recall that, given a class of digraphs $\vec{C}$, the problems $\#\dirhomsprob(\vec{C})$, $\#\dirsubsprobd(\vec{C})$, and $\#\dirindsubsprob(\vec{C})$ ask, respectively, given as input a pair $\vec{H}\in \vec{C}$ and $\vec{G}$, to compute $\#\homs{\vec{H}}{\vec{C}}$, $\#\subs{\vec{H}}{\vec{C}}$, and $\#\indsubs{\vec{H}}{\vec{C}}$. The crucial difference to the problems considered so far is that the parameter is just $|\vec{H}|$, rather than $|\vec{H}|+d(\vec{G})$, that is, we do not assume anymore that the host digraph has small outdegree.

Dalmau and Jonsson~\cite{DalmauJ04} proved their classification for counting homomorphisms not only for graphs, but in the more general setting of bounded arity relational structures. Since digraphs are precisely the relational structures over the signature containing one binary relation symbol, we can apply their main result and obtain:
\begin{theorem}\label{thm:homsdicho_unrestricted}
Let $\vec{C}$ be a recursively enumerable class of digraphs and assume ETH holds. Then $\#\dirhomsprob(\vec{C})$ is fixed-parameter tractable if and only if $\vec{C}$ has bounded treewidth.
\end{theorem}

Using Theorem~\ref{thm:homsdicho_unrestricted} as the starting point, and relying on the transformation of subgraph and induced subgraph counts as a linear combination of homomorphism counts (see Lemma~\ref{lem:subs_transformation} and Lemma~\ref{lem:indsubs_transformation}), we can mimic the proof of the classification theorems in the undirected setting due to Curticapean, Dell and Marx~\cite{CurticapeanDM17} almost verbatim. The only difference is the application of Dedekind's Theorem (Theorem~\ref{thm:dedekind}) as an interpolation method for reducing the computation of a linear combination of directed homomorphism counts from the computation of its individual terms (Lemma~\ref{lem:complexity_monotonicity}). This yields the following results; for the first one, we define the vertex-cover number of a digraph as the vertex-cover number of its underlying undirected graph.

\begin{theorem}\label{thm:subsdicho_unrestricted}
Let $\vec{C}$ be a recursively enumerable class of digraphs and assume ETH holds. Then $\#\dirsubsprob(\vec{C})$ is fixed-parameter tractable if and only if $\vec{C}$ has bounded vertex-cover number.
\end{theorem}

\begin{theorem}\label{thm:indsubsdicho_unrestricted}
Let $\vec{C}$ be a recursively enumerable class of digraphs and assume ETH holds. Then $\#\dirindsubsprob(\vec{C})$ is fixed-parameter tractable if and only if $\vec{C}$ is finite.
\end{theorem}

\end{document}